\documentclass{article}
\usepackage[T1]{fontenc}
\usepackage{iclr2027_conference,times}
\usepackage{amsmath,amssymb,amsthm,booktabs,graphicx,hyperref,xurl,array,placeins,alphalph}
\hypersetup{hidelinks,
  pdftitle={The Sample Complexity of Quantum Entanglement Allocation},
  pdfauthor={Nathan Roll}}
\newtheorem{proposition}{Proposition}
\newtheorem{theorem}{Theorem}
\newtheorem{corollary}{Corollary}
\newcommand{\Tr}{\operatorname{Tr}}

\title{The Sample Complexity of\\Quantum Entanglement Allocation}
\author{Nathan Roll\\\normalfont Stanford University\\\texttt{nroll@stanford.edu}}
\iclrfinalcopy
\begin{document}
\maketitle
\fancyhead{}
\renewcommand{\headrulewidth}{0pt}
\begin{abstract}
How many past requests are needed to decide which qubits should share
entanglement? We show that the answer depends on the allocation
choices created by the queries: a larger memory can require no more
data. The memory stores a classical bit and answers requests through
a fixed detector that preserves coherence within each measured sector.
For independent commuting $X$- and $Z$-type Pauli queries, we characterize
the full attainable prediction-contrast region and construct encodings
that preserve the bit at every nonzero vertex. With sharp reports,
a $d$-qubit path and groups of at most $k$ qubits have minimax excess
error after $m$ requests proportional to
$k^{-1}\min\{1,\sqrt{d\log(k+1)/m}\}$, uniformly for $2\leq k<d$.
Connected biclique regions can grow without increasing sample demand
when depth, region count and connections per region stay bounded. Preparation noise
introduces a separate calibration requirement. We derive an exact
tradeoff with extra fresh detector calls and transfer the learning law
to structured transaction co-location. Population-risk experiments test
the statistical predictions. We also compare encodings on a native
15-qubit device and learned partitions on public purchase baskets.
The full chain wins on the device; frequency grouping outperforms
basket search in the largest-capacity retail setting.
\end{abstract}

\section{Introduction}
A quantum memory is prepared before the request it must answer is known.
When different requests favor different groups of entangled qubits,
the encoder has to choose where to spend its entanglement budget.
We ask how many past requests it needs to make that choice.
This design question complements mechanistic work on how trained
quantum language models retain context \citep{roll2026}.

The choice matters even when the memory stores just one classical bit.
Consider a three-qubit memory queried through $ZZI$, $XXX$ and $IZZ$, with probabilities
$1/4$, $1/2$ and $1/4$. A product code supports both $Z$ parities
perfectly, but its $X$ parity is uninformative. With groups capped at
two qubits, entangling the first pair supports $ZZI$ and $XXX$, reducing optimal prediction
error from $25\%$ to $12.5\%$. Both memories retain the bit without
loss. What changes is which requests they can answer. The request
frequencies determine which pair is better
(Figure~\ref{fig:access-mechanism86}). The encoder learns how to arrange
information it already has.

The detector determines which arrangements are useful. We consider a
supplied parity detector that preserves arbitrary coherence within
each parity sector, as required of an ideal native parity measurement
\citep{riste2013}. Measuring individual qubits and combining their
outcomes can reveal the parity, but destroys this coherence. We hold
the detector fixed and vary the entanglement in the stored state.

Our answer depends on the choices created by the query geometry.
For independent commuting $X$- and $Z$-type Pauli queries, we
characterize the attainable prediction contrasts under arbitrary
allowed partition mixtures (Theorem~\ref{thm:region86}). This reduces
learning over all allowed encoders, including mixed states, to learning
which query sets to support. Paths add choices as they grow; boundedly
connected biclique regions can grow without adding to the sample
requirement (Theorem~\ref{thm:learning86}).

Imperfect preparation also makes calibration necessary. We study the
request and calibration budgets, price the extra detector calls that
can replace entangled preparation, and transfer the learning law to
classical transaction co-location.

\begin{figure}[htbp]\centering
\includegraphics[width=\linewidth]{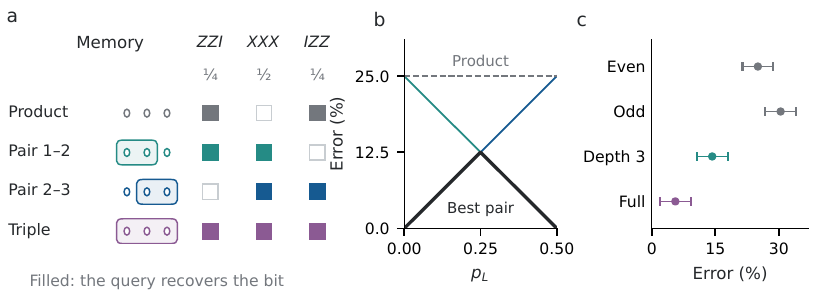}
\caption{(a) Three-qubit encodings at request probabilities $(1/4,1/2,1/4)$.
Groups share entanglement; filled squares recover the bit, empty squares
are uninformative. All four codes preserve the stored bit.
(b) Ideal prediction error when $\pi(XXX)=1/2$, $\pi(ZZI)=p_L$ and
$\pi(IZZ)=1/2-p_L$; black selects the better pair.
(c) Native 15-qubit commissioning on \texttt{ibm\_marrakesh} favors the
full chain. Even/odd are fixed product allocations. Each mean uses
1,920 shots over all 15 queries and both bits; bars are simultaneous
95\% Hoeffding bounds for four means, conditional on independent shots
in this one job. The allocations in (c) are fixed before acquisition.}
\label{fig:access-mechanism86}
\end{figure}

\section{The memory and its detector}
\label{sec:access-model86}
Let $C$ be uniform on $\{0,1\}$ and let $G=C\oplus Z$, with independent
$Z\sim\operatorname{Bernoulli}(h)$ and $0\leq h<1/2$. Write $D=1-2h$.
An encoder receives $G$, prepares a $d$-qubit state $\tau_G$, and
discards its classical record and environment. An independent request
$T=t$ arrives with probability $\pi_t$. The decoder receives $t$ and
the detector's classical report, and predicts $C$. It has no additional
copy, correlated input or later unrestricted measurement.

Given a collection $\mathcal F$ of allowed qubit partitions, let
$\mathcal S_{\mathcal F}$ be the convex hull of states product across
a partition in $\mathcal F$. Both codewords belong to this class, but
may use different partitions and mixtures. Allowing every partition
with blocks of size at most $k$ gives an entanglement-depth constraint.
The same collective detector is available to every encoder; $k$
constrains the stored state, not the detector's subsequent operations.

Orthogonal codeword supports preserve $G$: a global measurement
recovers it, giving unrestricted error $h$ and information
$1-H_2(h)$ bits about $C$. The optimization allows information-losing codes. We will show that
an optimum can always preserve the bit.

\subsection{Which reports preserve a parity sector?}
For a requested Hermitian Pauli observable $P$, put
$\Pi_\pm=(I\pm P)/2$. The detector's nonselective system channel must be
\begin{equation}
 \mathcal L_P(\rho)=\Pi_+\rho\Pi_+ + \Pi_-\rho\Pi_-.
 \label{eq:luders86}
\end{equation}
This L\"uders channel preserves all coherence within each sector.
An individual-outcome refinement can preserve the parity eigenvalue
while violating Eq.~\eqref{eq:luders86}.

Every instrument summing to $\mathcal L_P$ has outcome effects
$E_y=q(y|+)\Pi_++q(y|-)\Pi_-$ for a stochastic matrix $q$.
Thus the report contains only a postprocessed parity, however large
its alphabet. This follows from channel--measurement compatibility
\citep{heinosaari2013}; Appendix~\ref{app:detector86} gives a short
proof and an approximate version. The all-input requirement matters:
a shared syndrome service must preserve unknown coherent inputs even
when a particular client stores a classical bit.

We first use sharp parity reports. Symmetric, code-independent report noise
gives effects $(I\pm v_tP_t)/2$, with known $v_t\in[0,1]$.
For contrast $c_t=|\Tr[P_t(\tau_0-\tau_1)]|/2$, the optimal
classical decoder has risk
\begin{equation}
 R=\frac12-\frac D2\sum_t\pi_t v_t c_t.
 \label{eq:contrast-risk86}
\end{equation}
Small departures from exact channel preservation can also be bounded:
Appendix~\ref{app:detector86} gives a dimension-independent continuity
bound and a sufficient diamond-distance threshold for the separation.
Under the exact contract, the remaining task is to determine which
contrast vectors fit the entanglement budget.

\section{Which queries can an encoding support?}
\label{sec:region86}
Queries that commute globally can compete within separate qubit groups.
Consider independent, globally commuting Pauli generators
$P_1,\ldots,P_q$. Independence means no nonempty product is proportional
to the identity. A CSS family uses only $X$ and identity within each
$X$-type generator, and only $Z$ and identity within each $Z$-type
generator. For a partition $\Pi$, build a graph $H_\Pi$ on queries:
two queries are adjacent when their restrictions anticommute on a
block. This graph is bipartite, since generators of the same CSS type
cannot conflict. Let $\operatorname{STAB}(H)$ be the convex hull of
independent-set indicators, including the empty set.

To characterize the contrasts, we first reduce a pair of codewords to
one state. A tensor product of single-qubit Paulis flips all independent
generators at once and preserves every partition class. Binary
contrasts therefore have the same attainable region as single-state
absolute correlations (Appendix~\ref{app:region86}).

\begin{theorem}[Exact attainable contrast region]
\label{thm:region86}
For independent commuting CSS generators and any partition family
$\mathcal F$, the contrast vectors of all pairs
$\tau_0,\tau_1\in\mathcal S_{\mathcal F}$ are exactly
\begin{equation}
 \mathcal C_{\mathcal F}
 =\operatorname{conv}\!\bigcup_{\Pi\in\mathcal F}
                  \operatorname{STAB}(H_\Pi).
 \label{eq:region86}
\end{equation}
Every nonempty independent-set indicator has an explicit encoding
with orthogonal supports in the allowed class.
\end{theorem}

Local uncertainty bounds each conflicting pair by
$|\langle P_u\rangle|+|\langle P_v\rangle|\leq1$
\citep[Proposition~A.3]{zander2024}. Bipartite integrality and convexity turn these pairwise inequalities
into an upper bound on the whole region, including mixed codeword
pairs. To attain the bound, take an independent set $I$ and define
\begin{equation}
 \tau_g^I=2^{-d}\prod_{t\in I}\bigl(\mathbb I+(-1)^gP_t\bigr)
 \label{eq:projector-code86}
\end{equation}
where $\mathbb I$ is the identity. This normalized stabilizer projector
is separable across the chosen partition and has $c=\mathbf1_I$.
Every selected expectation has sign $(-1)^g$. Mixing the pairs thus
adds their contrasts without cancellation and fills the region in
Eq.~\eqref{eq:region86}. Each nonempty pair preserves $G$, although
a mixture of such pairs need not do so.

With $a_t=\pi_tv_t$, the optimal risk is
$1/2-(D/2)\max_{\Pi\in\mathcal F}\alpha_a(H_\Pi)$, where $\alpha_a$
is maximum independent-set weight. The workload objective is linear, so a best vertex is optimal and its
encoding preserves the bit. Mixtures can enlarge the feasible region
without improving this optimum. Appendix~\ref{app:region86} gives the
proof and counterexamples to dropping CSS or independence; the same appendix gives a worked lossless mixture.

\subsection{Which qubits should share entanglement?}
\label{sec:allocation86}
For graph-state queries, the contrast region becomes a concrete
allocation problem. A bipartite graph has generators
$K_i=X_i\prod_{j\sim i}Z_j$; Hadamards on one color convert them
to CSS form. The conflict graph contains exactly the edges crossing
the partition. Weighted bipartite vertex cover then gives the following rule.
Delete a minimum-weight set of query vertices until every remaining
connected component has at most $k$ vertices, and entangle within
those components. A deleted vertex means giving up a query; its qubit
stays in the memory. This is a variant of classical component-order connectivity
\citep{drange2014}: deletion costs carry workload weights, while
the surviving-component limit counts vertices.

With sharp reports, error $h$ is attainable exactly when every
connected component of the positive-weight query graph fits the
budget. For a uniform $d$-vertex path,
\begin{equation}
 R_k^*=h+\frac{D}{2d}\left\lfloor\frac d{k+1}\right\rfloor.
 \label{eq:path-depth86}
\end{equation}
Exact recovery of $G$ through every query requires depth $d$.
Constant-size groups suffice for a fixed excess error independent of $d$. Each path query touches at most
three qubits, so independent local outcome errors introduce at most
three visibility factors.

\section{Learning where entanglement belongs}
\label{sec:learning86}
The allocation rule assumes that the request frequencies are known.
We now give the learner only $m$ i.i.d.\ past query identities.
It chooses an encoding before future requests arrive and receives
no prediction labels.
For a known bipartite graph $\Gamma$ and its graph-state queries, let
$R_{\pi,k}^*$ be the best depth-$k$ risk for workload $\pi$, and define
$\mathfrak M_m(\Gamma,k)=\inf_{\widehat\tau}\sup_\pi
\mathbb E[R_\pi(\widehat\tau)-R_{\pi,k}^*]$.
The supremum ranges over all query distributions. The $m$ query
identities are the learner's only observations.

We learn the allocation by maximizing empirical workload weight over
\emph{inclusion-maximal} feasible query sets: no further query can be
added while preserving feasibility. Restricting to these sets preserves every nonnegative workload
optimum, and each set specifies an explicit lossless code.
With at most $2^q$ feasible indicators, finite-class concentration
gives excess error $O(D\sqrt{(q+\log(1/\delta))/m})$ with probability
$1-\delta$. This finite-class bound assumes exact empirical optimization, which
can remain combinatorial. To sharpen it, we need to count the
undominated allocation choices that the requests must distinguish.

In the opening example, the choice is which endpoint to pair with
the center. Local graph queries can create proportionally more such
choices as the memory grows. A complete bipartite graph $K_{a,b}$
has a different structure. At depth $k\leq\min(a,b)$, a retained
set either lies in one color or has at most $k$ vertices. The maximal
sets are the two full colors and mixed sets of size $k$, so the oracle
compares the color totals with the sum of the $k$ largest weights.

To separate the number of regions from their size, divide a bipartite
graph into $r$ complete-bipartite regions, linked through at most $p$
designated connection vertices per region. Each color must have at
least $k$ nonconnection vertices. Regions may have different sizes,
and physical blocks may cross their boundaries. The following result
compares this family with bounded-degree graphs.

\begin{theorem}[Geometry and the cost of allocation learning]
\label{thm:learning86}
Assume sharp reports and $m\geq1$ i.i.d.\ query samples.
On a path with $d\geq3$ vertices and $2\leq k\leq d-1$,
\begin{equation}
 \mathfrak M_m(P_d,k)
 =\Theta\!\left(\frac Dk\min\left\{1,
                 \sqrt{\frac{d\log(k+1)}m}\right\}\right).
 \label{eq:path-learning90}
\end{equation}
The constants are uniform in $d,k,m$.
For $r$ biclique regions with fixed connection bound $p$ and at
least $k$ nonconnection vertices per color,
\begin{equation}
 \mathfrak M_m(\Gamma,k)
 =\Theta_p\!\left(D\min\{1,\sqrt{rk/m}\}\right).
 \label{eq:region-learning86}
\end{equation}
The second law is uniform in $r$, $k$ and the region sizes.
A single region without connections recovers
$\mathfrak M_m(K_{a,b},k)=\Theta(D\min\{1,\sqrt{k/m}\})$.
All partitions within the depth budget are allowed. Both lower bounds
permit arbitrary randomized learners choosing mixed codeword pairs
and grant the optimal decoder. Information-preserving codes attain both rates. The path learner
uses a random periodic allocation at small budgets and empirical
maximization at larger budgets; regional empirical maximization
attains the second rate.
If $d\leq k$, the whole graph fits and allocation excess is zero.
\end{theorem}

On a path, one deletion every $k+1$ sites sacrifices at most
$1/(k+1)$ of the workload after averaging over offsets. The resulting small oracle loss lets us sharpen the concentration
bound. For the lower bound, separated
$(k+1)$-site segments each force a choice of which query to sacrifice;
identifying the least requested of $k+1$ alternatives supplies the
logarithmic factor (Appendix~\ref{app:growing-depth90}). Increasing depth reduces absolute error. Learning a sufficiently small
fixed fraction of that error scale still requires order
$d\log(k+1)$ requests. A sliding-minimum recurrence
solves the ideal empirical problem in $O(d)$ operations after counting.

The maximal biclique allocation class has VC dimension $k$,
although all feasible sets together have dimension $\max(a,b)$.
With $r$ connected regions, interior choices contribute at most $rk$
and connection vertices at most $rp$. A matching lower bound projects
onto $r$ independent capacity constraints, even when physical groups
cross regions (Appendix~\ref{app:modular86}). A chain of regions
admits exact empirical optimization in $O(rk^2)$ operations after sorting.
Completing empirical ties to maximal sets preserves training weight
and weakly improves every nonnegative population objective
(Appendix~\ref{app:completion90}).

The laws describe worst-case workloads. A fixed workload with a clear
gap between competing allocations can be learned faster. Giving every
query positive probability leaves the minimax value unchanged unless
a common probability floor is imposed. The regional reduction in sample
demand also comes with a physical cost: biclique queries grow with the
region, whereas bounded-degree queries touch at most $\Delta+1$ qubits.
Paths and chains of regions admit efficient exact empirical optimization;
the general graph theorem assumes access to an optimization oracle.

\section{Learning allocation under preparation errors}
\label{sec:noise86}
So far, supporting another query can only help. Preparation errors
change that tradeoff: the gate that supports one query may reduce
its neighbors' reliability. On a path, select vertices $S$, prepare
a graph state on each selected run, and put deleted sites in
$|0\rangle$. Encode $G$ by applying
$Z^G$ to every selected site. Suppose each retained CZ gate on edge
$e$ is independently followed by a $ZZ$ error with probability
$\epsilon_e\leq1/2$.
With $\gamma_e=1-2\epsilon_e$ and query visibility $v_i$,
\begin{equation}
 b_i(S)=\mathbf1\{i\in S\}\,v_i
        \prod_{\substack{e\ni i\\e\text{ retained}}}\gamma_e.
 \label{eq:noisy-contrast86}
\end{equation}
Here $b_i(S)$ includes reporting visibility, and the risk is
$1/2-(D/2)\sum_i\pi_i b_i(S)$.
Retaining an edge changes the value of neighboring queries, so its
cost cannot be absorbed into fixed workload weights.

The encoded bit can survive these preparation errors. It is enough
to retain one odd-sized component. The product of its stabilizer generators has eigenvalue
$(-1)^G$ and commutes with every retained-edge $ZZ$ error. The noisy
codewords remain in orthogonal sectors even if other components have
even size. An odd full chain also preserves $G$. A two-site run alone
fails because its $ZZ$ error flips the encoded bit.

We optimize within this family by tracking the selected-run length
and whether an odd run has closed. The dynamic program uses $O(dk)$
transitions; an independent $O(dk^2)$ interval implementation checks
its answers. The noisy
optimization is over this lossless circuit family; the ideal theorem
optimizes over the full physical class.

Learning this allocation requires information about the device as
well as the workload. Suppose we observe $m$ requests and $B$
independent calibration trials for each edge reliability and native
query visibility. This statistical model supplies each parameter
separately, at a cost of $B(2d-1)$ calibration observations in addition
to the request log. Appendix~\ref{app:service91} instead calibrates
directly observed binary responses. The comparator is the best allocation in the same noisy circuit
family; Appendix~\ref{app:two-resources90} also bounds errors caused
by channel mismatch and request smoothing.

The two-resource rate is sharp on depth-three paths. For the same
circuit family with $d\geq4$ and correctly specified noise, the minimax
excess from $m$ requests and $B$ trials per parameter is
$\Theta\!\left(D\min\{1,\sqrt{d/m}+B^{-1/2}\}\right)$
(Theorem~\ref{thm:joint-law90}). Four-site allocation choices force
the request term; two almost indistinguishable three-site channels
with different optimal allocations force the calibration term.

\section{Experiments}
\label{sec:experiments86}
The experiments first test the learning laws under their stated
channels, where we can compute population risk exactly. We then test
fixed encodings on a quantum processor. The transaction study in
Section~\ref{sec:classical97} asks whether allocation learned from a
public purchase trace improves locality.

\subsection{What changes when the memory grows?}
\begin{figure}[!htb]\centering
\includegraphics[width=\linewidth]{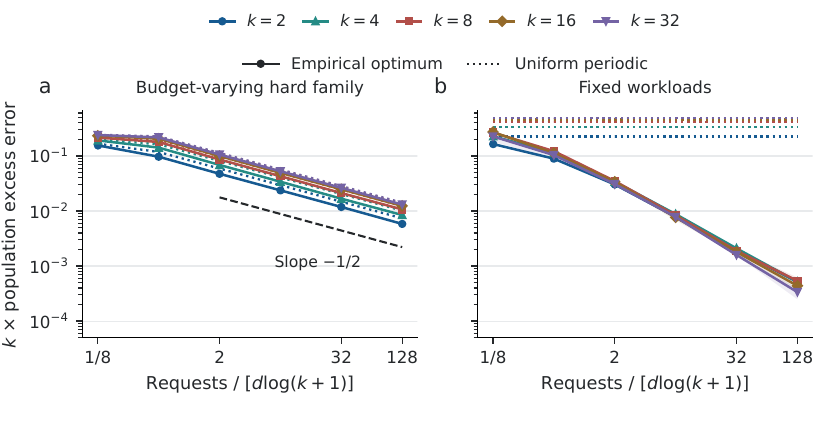}
\caption{Depth changes the error scale as well as the request requirement.
(a) On the budget-indexed hard family, both empirical and data-free
periodic allocations follow a square-root rate. The task probabilities
change with the budget. (b) On fixed Dirichlet workloads, only the
empirical learner improves. Solid lines show empirical optima; dotted
lines show uniform periodic allocations. Each line averages four
dimension means, each based on 64 independent seeds; shading spans
those dimension means, not confidence intervals. The dashed line is
$0.025x^{-1/2}$, an unfitted slope guide, not a theorem constant.
Every risk uses the exact same-depth physical optimum as comparator.}
\label{fig:growing-depth90}
\end{figure}
We varied $d$ from 128 to 1,024 and $k$ from 2 to 32, for 15,360
fits, to test the depth dependence in the path law. Figure~\ref{fig:growing-depth90} scales excess by $k$
and requests by $d\log(k+1)$. Both empirical and data-free allocations
follow the hard-family slope because the task probabilities change
with the budget. On fixed workloads, the empirical learner improves
while the periodic comparator stays flat. The fixed-workload control is essential: a hard-family slope alone
would also credit a method that never uses the requests
(Appendix~\ref{app:depth-study90}).

Connected-region experiments keep the request scale $m/(rk)$ fixed
while changing region size. The curves remain similar, as predicted
by Theorem~\ref{thm:learning86} (Appendix Figure~\ref{fig:geometry-learning86}).
Empirical ties are completed using training counts alone; retaining
only observed queries would introduce an avoidable occupancy penalty.
At the original primary condition, selective connection repair plus
local improvement is within $0.145$ error points of exact optimization
(95\% interval $[0.075,0.220]$). The full comparison uses every workload
at the same depth (Appendix~\ref{app:controls90}).

\subsection{Requests and calibration limit different decisions}
To test the two-resource prediction, we varied request and calibration
budgets separately. The prespecified protocol draws 192 independent
workload/channel instances, crossing 63 or 127 qubits with low,
heterogeneous or stress preparation noise. Each of the six conditions contains 32 instances,
four request streams and four calibration panels. Five nested request
budgets and five calibration budgets give 76,800 fits per method.
Joint uses requests, edge reliability and report visibility.
Task only uses requests; Reporting only also learns visibility.
We hold the optimizer and tie rule fixed and compare both raw counts
and smoothed estimates to isolate the value of each observation type. Same-family greedy and privileged-oracle
controls accompany the full grid (Appendix~\ref{app:experiments90}).

\begin{figure}[!htb]\centering
\includegraphics[width=\linewidth]{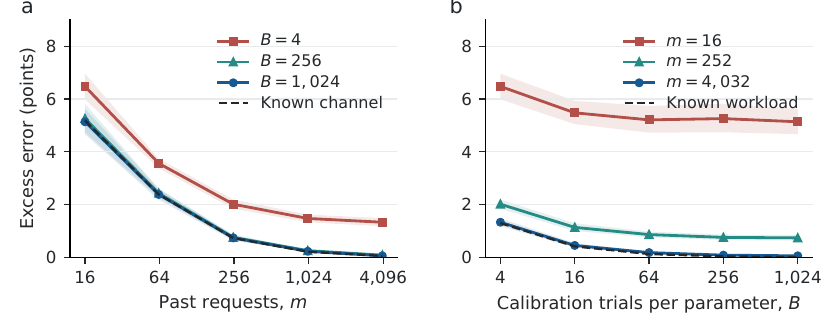}
\caption{More requests cannot remove a calibration floor.
(a) Request learning at three fixed calibration budgets; the dashed
curve supplies the true channel. (b) Calibration learning at three
fixed request budgets; the dashed curve supplies the true workload.
Both panels use 63-qubit heterogeneous-noise instances and the smoothed
Joint estimator. Excess is relative to the true-workload, true-channel
oracle in the same circuit family. Bands are pointwise 95\% bootstrap
intervals over 32 independent instances, after averaging four request
streams and four calibration panels. Calibration costs $125B$
observations, separately from $m$.}
\label{fig:two-resources90}
\end{figure}
With $B=4$, increasing requests from 16 to 4,032 reduces excess
from $6.48$ to $1.33$ points. Even exact workload knowledge leaves $1.30$ points of excess, so more
requests cannot remove this calibration floor. At $B=256$,
the same request increase reduces excess from $5.26$ to $0.085$
points. Conversely, at $m=16$, even 1,024 calibration trials leave
$5.14$ points: the request estimate is now the bottleneck
(Figure~\ref{fig:two-resources90}). The same instances thus exhibit either bottleneck, depending on which
data are scarce.
A separate near-indifference diagnostic gives approximately constant
$\sqrt B$-scaled excess over $B=4$ to 4,096, with workload known.

In the prespecified 63-qubit heterogeneous-noise comparison,
1,008 requests and 256 calibration trials per parameter reduce error
from 9.79\% for Task only to 8.58\% for Joint: a 1.21-point gain,
paired 95\% interval $[0.95,1.50]$. Calibration costs 32,000
observations; the learned code uses 22.3 preparation CZs on average,
versus 62 for Full chain (Appendix Table~\ref{tab:primary90}).
These are population risks under the specified channel.

Full chain has lower error in 76 of 150 cells, including every
low-noise cell. The benefit of supporting more queries can outweigh
their preparation cost. Smoothing affects Joint and Task only
differently, so we report all four combinations, along with depth-five
and depth-seven sensitivity and same-family greedy comparisons
(Appendix~\ref{app:experiments90}).

\subsection{Testing the encodings on a quantum processor}
We implement all 15 graph queries on a 15-data-qubit path on
\texttt{ibm\_marrakesh}. Seven query centers have degree two in the
device graph. Their parity-extraction circuits borrow data qubits and restore them
in the ideal unitary. Dropping those queries would remove the local
conflicts that define the learning problem. Preparation and extraction compile
separately. Both full and shallow chains use two parallel CZ layers.

In 7,680 commissioning shots, uniform-request error is 5.57\% for
Full chain, 14.27\% for a fixed depth-three allocation, and 25.05\%
or 30.42\% for the two product colors (Figure~\ref{fig:access-mechanism86}c).
The full chain wins this fixed-encoding comparison. Two preliminary
learning jobs timed out without usable observations, leaving the
four-workload learning/calibration design unmeasured. A separate
coherence check also exposed a failure: one witness reverses its
expected sign in two jobs, including both bits and phases in the
follow-up. The implemented extraction therefore still needs validation
across the full path (Appendix~\ref{app:hardware97}).

\section{Consequences of the allocation law}
\paragraph{Changing the detector interface.}
\label{sec:interface97}
Extra detector calls offer another way to obtain information. We compare
their cost with entangled preparation in a common statistical experiment.

For detector diagnosis, a healthy or faulty device independently flips
each reported sign with probability $e_0<e_1$. A supported query gives
an observed flip; an unsupported query erases it. Write $\mathsf E_p$
for this experiment when the supported request mass is $p$.
Appendix~\ref{app:service93} shows that the optimal $N$-call
experiment over all admissible fresh CSS probes is
$\mathsf E_{p^*}^{\otimes N}$, where $p^*$ is the maximum supported
mass. At one call and false-alarm level $e_0$, power regret is
$2(e_1-e_0)/D$ times prediction regret. Diagnosis inherits the same
request-learning laws.

Suppose a competing policy can make additional calls. If each fresh call is
dominated by $\mathsf E_p$ and a fixed probe attains it, reproducing
$N$ copies of $\mathsf E_q$, $q\geq p$, requires at least $Nq/p$
calls in expectation under either hypothesis. A stopped policy attains
equality (Theorem~\ref{thm:interface97}). The equality compares expected
call budgets for equivalent experiments. Fixed horizons at a specified
power can have a different ratio.

For the menu $(ZZI,XXX,IZZ)$ with probabilities $(.2,.6,.2)$, product
and Bell coverages are $.6$ and $.8$: the fresh-call ratio is $4/3$.
If the policy can choose its second query, a product-initialized
policy matches Bell at 1.2 expected calls. Retaining the probe
and repeating an unsupported query gives yet another signal: the
reports disagree with probability $2e_H(1-e_H)$ in health state $H$.
The disagreement diagnoses noise without recovering the stored sign.

\begin{figure}[htb]\centering
\includegraphics[width=\linewidth]{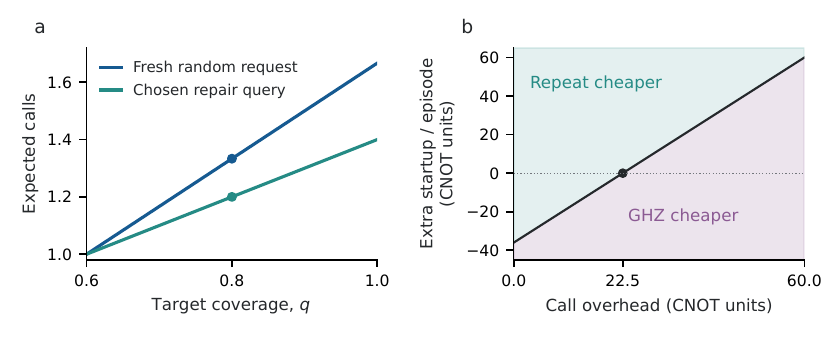}
\caption{(a) Fresh random requests and chosen repair queries have different
call costs. Points reproduce Bell coverage $q=.8$ from a product probe.
(b) Cost comparison at 80\% power and 5\% false alarms for known
$(e_0,e_1)=(.02,.164)$. GHZ uses 18 calls and 111.6 logical CNOTs;
a product-initialized repeated-query policy uses 19.6 expected calls
and 75.6 CNOTs. The boundary includes extra GHZ startup cost per
episode. These are two attaining policies, not the full adaptive frontier.}
\label{fig:interface97}
\end{figure}
At equal operating performance, entanglement reduces lifetime cost
exactly when its extra startup cost is smaller than the accumulated
saving in preparation and detection (Eq.~\eqref{eq:cost97}). Previous
IBM diagnosis runs measured GHZ power gains of 9.11 and 16.34 points
over a selected product probe, with frozen aggregate alarm rules and
dynamic coherence controls. Their different calibrations prevent a
matched-cost hardware comparison with the adaptive policy
(Appendix~\ref{app:service94}).

\paragraph{Classical transaction co-location.}
\label{sec:classical97}
A database also chooses which records to place together before requests
arrive. A transaction footprint $F_t$ succeeds when all its records
share a shard, subject to a capacity cap and no replication.
Schism formulates this hypergraph objective and uses tuple coaccess
graphs as a surrogate \citep{curino2010schism}. On structured schemas,
we recover the same allocation problem.

Theorem~\ref{thm:learning86} transfers exactly to row/column transactions
on $r$ disjoint $s\times s$ tables when capacity is
$C=(s-1)k+1$ and
$s\geq\max\{2,k,\lfloor(k-1)^2/4\rfloor+2\}$.
Allowing $rs$ shards, the minimax locality regret is
$\Theta(\min\{1,\sqrt{rk/m}\})$ (Corollary~\ref{cor:classical97}).
Larger tables need not require more requests, although transaction
width and storage resources grow. All 15,360 checks of actual record
placements reproduce the predicted masks.

\begin{figure}[htb]\centering
\includegraphics[width=\linewidth]{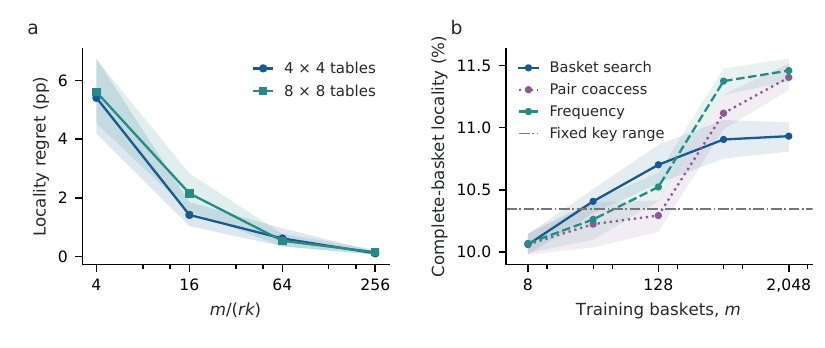}
\caption{(a) Exact empirical co-location on structured tables, with
requests scaled by $rk$. Curves average $r\in\{1,4\}$ and
$k\in\{1,2,3\}$, with 32 instances per cell; bands bootstrap
training instances. (b) Complete-basket locality on the public Online
Retail trace, with 3,922 records and eight shards of capacity 512.
Learned methods share a training prefix; fixed key range uses none.
Bands bootstrap 16 request
streams conditional on the fixed catalogue and test trace. All three
capacities and the exact small-schema study appear in Appendix~\ref{app:classical97}.}
\label{fig:classical97}
\end{figure}
On public purchase baskets \citep{chen2015retail}, we test both exact
optimization on 12/24-item projections and a heuristic on the full
3,922-record catalogue. The full study scores all 3,992 held-out
invoices without projection. At capacity 512 and 512 training baskets,
locality is 10.90\% for basket-based search, 11.37\% for same-budget
frequency grouping, and 10.35\% for fixed key-range placement.
The simpler frequency control wins this comparison. Uniform SmallBank
footprints leave no allocation to learn: all balanced placements have
the same locality.

\FloatBarrier
\section{Related work and implications}
Exact separable stabilizer polytopes and attaining product states are
known for selected GHZ and cluster families \citep{jafarizadeh2008}.
\citet[Propositions~A.3--A.4]{zander2024} give our
linear cut inequalities. \citet{knips2016} optimize weighted independent sets of
cut-anticommutativity graphs for squared correlations;
\citet{li2026restricted} derive linear graph-generator bounds using
crossing-edge matchings. Our constructive region for independent CSS families
under arbitrary partition mixtures yields sharp learning laws.
Stable-set integrality \citep{chvatal1975}
and component-order connectivity \citep{drange2014} are classical
ingredients. 

Predictive $\mathcal V$-information and the Decodable Information
Bottleneck describe predictor-dependent information
\citep{xu2020usable,dubois2020dib}; our logarithmic-loss extension
applies this prediction viewpoint with physically restricted observations
and unrestricted classical decoders. Data hiding and locally indistinguishable
product ensembles separate storage from access
\citep{terhal2001,bennett1999}. 
Channel discrimination benefits from entanglement \citep{piani2009}.
Here we identify the optimal experiment for a fixed query menu and
transfer its allocation-learning cost. 

Memory size alone does not determine the cost of learning an allocation.
Paths create more choices as they grow; biclique regions need not.
The same distinction arises in classical co-location. Practical gains
also depend on preparation costs and the workload, which can favor
simpler choices. We have yet to characterize general measurements or
test the learning laws across independent hardware jobs.

\FloatBarrier
\label{page:main-end}
\clearpage
\section*{Reproducibility statement}
The ancillary code and data provide raw observations, seeds, encoders,
resource counts and replay code. This manuscript includes the full proofs
and extended experimental record. A small installable package reproduces
the path optimizer and online sampler. The ancillary README gives the
entry points for the hardware, classical and interface evidence and
distinguishes recorded observations from unacquired experiments.

\section*{AI use statement}
Generative AI assisted with theoretical models and mathematical claims,
proof development and writing, hypothesis refinement, experimental design,
method implementation, synthetic-data generation, data processing, result
interpretation, and figure preparation. It also assisted with literature
review and manuscript drafting and editing. The author takes responsibility
for the final manuscript and artifacts.

\section*{Ethics statement}
The experiments use synthetic tasks, a quantum processor and a public
purchase dataset. The transaction study reads stock codes, invoice
identifiers, times, quantities and prices; it does not read customer
identifiers or link purchases to people. It reports aggregate locality.
The hardware results compare fixed implementations. Their interpretation
is limited to the stated measurement and control assumptions and does
not establish a general quantum learning speedup.
\clearpage
\bibliography{references}
\bibliographystyle{iclr2027_conference}
\newpage
\appendix
\noindent\textit{Appendix guide.}
The proofs establish the physical region and request-learning laws,
then extend them to logarithmic loss, finite calibration and online
allocation. The experimental appendix follows the current comparisons;
the later experimental sections preserve earlier studies and their limitations.
\begin{center}
\begin{tabular}{p{.65\linewidth}r}\toprule
Material & Page\\\midrule
Detector preservation and allowed reports & \pageref{app:detector86}\\
Exact contrast region and prior-work comparison & \pageref{app:region86}\\
Learning bounds over physical encoders & \pageref{app:learning86}\\
Preparation noise and lossless circuit family & \pageref{app:noise86}\\
Optimal logarithmic loss & \pageref{app:logloss90}\\
Growing-depth path law & \pageref{app:growing-depth90}\\
Empirical completion and two-resource guarantees & \pageref{app:completion90}\\
Online allocation & \pageref{app:online90}\\
Factorial protocol and complete comparisons & \pageref{app:experiments90}\\
Stronger baselines and model mismatch & \pageref{app:controls90}\\
Reproduction and historical evidence & \pageref{app:reproduction90}\\
Parity-service tests and physical response calibration & \pageref{app:service91}\\
Full-class diagnosis and model comparisons & \pageref{app:service93}\\
Prospective quantum-processor evaluation & \pageref{app:service94}\\
Extra detector calls & \pageref{app:interface97}\\
Transaction co-location & \pageref{app:classical97}\\
Native 15-qubit experiments & \pageref{app:hardware97}\\\bottomrule
\end{tabular}
\end{center}
\section{Detector preservation and permitted reports}
\label{app:detector86}
The instrument acts on the complete retained memory, with no additional
input correlated with $G$. The decoder observes only its classical
report and the query. The contract concerns all input density operators,
including states absent from the particular bit code. Preserving just
the actual codewords would not exclude local decoding of a product
repetition code.

\subsection{Exact compatibility}
Let $\{\mathcal J_y\}$ be CP instrument outcomes with
$\sum_y\mathcal J_y=\mathcal L_P$. For Choi matrices,
$0\preceq C(\mathcal J_y)\preceq C(\mathcal L_P)$. The latter has support
$\operatorname{span}\{|\Pi_+\rangle\!\rangle,|\Pi_-\rangle\!\rangle\}$.
Every Kraus operator of every outcome therefore has the form
$A_{yj}=a_{yj}\Pi_++b_{yj}\Pi_-$. Orthogonality of the projectors gives
\[
 E_y=\sum_j A_{yj}^\dagger A_{yj}
 =q(y|+)\Pi_++q(y|-)\Pi_-,
 \qquad
 q(y|+)=\sum_j|a_{yj}|^2,\quad q(y|-)=\sum_j|b_{yj}|^2.
\]
The effects sum to $I$, so $q$ is stochastic. Conversely, measuring
the sharp parity and applying $q$ realizes each such report.
This proves the stated equality of report classes. It is a special
case of channel--observable compatibility \citep{heinosaari2013}.
The contract is stronger than preserving the value of $P$: a refinement
inside an eigenspace can commute with $P$ while destroying its coherence.
Native parity meters motivate the preservation requirement
\citep{riste2013}; we do not infer this requirement from classical storage.

\subsection{A conservative approximate bound}
Suppose the actual nonselective channel $\mathcal E_t$ obeys
$\|\mathcal E_t-\mathcal L_{P_t}\|_\diamond\leq\delta_t$, with the
full diamond norm, not half that norm. Pad the Stinespring environments
to a common dimension at least the sum of their Kraus ranks. The
continuity theorem \citep{kretschmann2006}, in the form of
\citet[Proposition~1]{vomende2023}, provides an alignment with isometry
distance at most $\sqrt{\delta_t}$. Padding is material:
the same optimal bound need not hold in an arbitrarily prescribed
minimal environment \citep[Example~1]{vomende2023}.

An instrument's classical report is a POVM on the channel environment.
Extend that POVM to the padded space. The actual report channel
$\mathcal C_t$ is then within $2\sqrt{\delta_t}$ in diamond norm of
some parity-postprocessing channel $\mathcal C_t^0$.
For $\Delta=D(\tau_0-\tau_1)$, $\|\Delta\|_1\leq2D$, while binary risk is
\[
 R_{\mathcal C}=\frac12-\frac14
       \sum_t\pi_t\|\mathcal C_t(\Delta)\|_1.
\]
The triangle inequality gives, uniformly over all allowed encoders,
\begin{equation}
 R_{\mathcal F}^{\rm relaxed}\geq
 \max\left\{h,\ R_{\mathcal F}^{\rm sharp}
                 -D\sum_t\pi_t\sqrt{\delta_t}\right\}.
 \label{eq:robust-detector86}
\end{equation}
The reference is the ideal \emph{sharp-parity} optimum. The nonselective
channel contract does not enforce a chosen reporting visibility.
It also does not guarantee a useful actual report: the one-outcome
instrument $\mathcal J_*=\mathcal L_P$ has $\delta=0$ and risk $1/2$
for every encoder. A device comparison requires an informative predictor
below the comparator bound, in addition to channel calibration.

For the three-vertex gadget family of Appendix~\ref{app:gadget-family86},
sharp separable risk is $h+D/4$,
and the depth-two oracle attains $h+D(1-\theta)/8$.
If all $\delta_t\leq\delta$, the latter beats the relaxed separable class by
at least $D[(1+\theta)/8-\sqrt\delta]$. The sufficient radius
$\delta<1/64$ is independent of the number of gadgets.

\subsection{The selector bypass and its disturbance}
In the global-parity family retained in the historical supplement, an
$X$-basis repetition code and a query-selected $X$ outcome recover $G$
for $d\geq3$. The same local bases and a one-bit report suffice.
Software parity is therefore a substantive summary restriction.
The permanent audit enumerates this bypass for $d=3,\ldots,8$.

For an explicit physical relaxation, use the square-root instrument
$A_\pm=\sqrt{(I\pm\eta Z_1)/2}$, $0\leq\eta\leq1$, composed with
$\mathcal L_{ZZ}$. Its local report has contrast $\eta$ on the code
$|00\rangle,|11\rangle$, hence error $(1-\eta)/2$ at $h=0$.
Put $\kappa=\sqrt{1-\eta^2}$. The nonselective channel is
\[
 \mathcal E_\eta=\frac{1+\kappa}{2}\mathcal L_{ZZ}
       +\frac{1-\kappa}{2}\operatorname{Ad}_{Z_1}\circ\mathcal L_{ZZ},
 \qquad
 \|\mathcal E_\eta-\mathcal L_{ZZ}\|_\diamond=1-\kappa.
\]
The convex mixture gives the upper bound; the Bell-sector state
$(|00\rangle+|11\rangle)/\sqrt2$ attains it. Its half trace-distance
disturbance is $(1-\kappa)/2$. These formulas describe an attained curve
for this specified refinement, not an optimal frontier over all
instruments. The audit retains 21 values including $\eta=0,1$.

\subsection{Operational resource account}
\label{app:resources86}
Consider a hypothetical processor with two separated roles.
A preparation client may choose the memory encoding within its assigned
entanglement budget. A later prediction client may call a fixed catalogue
of native syndrome commands through a shared endpoint, but cannot change
the detector or access its backend. The same endpoint serves coherent
jobs on unknown input states. Its nonselective channel therefore has a
specified obligation to preserve every within-sector coherence,
independently of which job supplies the input. This is a service contract;
classical storage itself does not require it.

The preparation client uses past request identities to choose an
encoding, receives $G$, and prepares the memory. It discards the source
record and environment at handoff. Only afterward does the prediction
client receive a query. That client observes the query and the endpoint's
classical report, with no additional copy of $G$, correlated input or
backend access. Preparation can change while the shared detector service
stays fixed. This scenario motivates the formal access model; our
two-qubit acquisition establishes neither such an endpoint nor an actual
deployment.

\begin{table}[htbp]
\caption{Resources in the foundational model. A common detector is
held fixed within each allocation comparison. Comparisons between
different query families do not equate their total implementation cost.}
\centering\small
\begin{tabular}{p{.25\linewidth}p{.66\linewidth}}\toprule
Resource & Treatment\\\midrule
Stored information & Every selected nonempty code preserves the same bit;
the converse permits all encoders in the allowed state class.\\
Memory entanglement & Block size and allowed placements constrain the
stored state. They do not constrain operations inside the detector.\\
Preparation & Noiseless theory permits all allowed states. Noisy design
uses a specified circuit family and reports its CZ count separately.\\
Collective detection & The queried observable and coherent-channel
contract are fixed capabilities. Gate count, duration and detector
construction are not included in preparation cost.\\
Query locality & Graph-generator support equals degree plus one.
Bounded degree enters the learning lower bound; dense-family detectors
can require larger support.\\
Calibration & Original simulations use known parameters; the fresh
prediction study estimates preparation and reporting parameters from
finite data. True parameters are reserved for population evaluation and
privileged oracles. The detector theorem separately assumes a
diamond-distance bound, which is not measured on the device.\\
Learning & Classical query identities are the training observations.
Statistical upper bounds assume exact empirical optimization.\\\bottomrule
\end{tabular}
\end{table}
A retained classical copy, an additional correlated input, or a
different detector can change the optimum. The model therefore
quantifies storage entanglement under a specified access service,
without asserting a total-cost quantum advantage.

\FloatBarrier
\section{Proof of the exact contrast region}
\label{app:region86}
Assume $q\geq1$, independent commuting Hermitian CSS Pauli generators and a
nonempty allowed partition family. We use the conventional uniform
prior and binary prediction loss throughout this region theorem.

\subsection{Single-state bounds and downward closure}
For a state product across $\Pi$, let $x_t=|\Tr(P_t\rho)|$.
If $u,v$ are adjacent in $H_\Pi$, global commutation implies an even,
nonzero number of blocks with local anticommutation. Choose two,
$A,B$, and denote their expectation pairs by $(a_u,a_v)$ and
$(b_u,b_v)$. Anticommutation gives
$a_u^2+a_v^2\leq1$ and $b_u^2+b_v^2\leq1$. The remaining factors
have absolute value at most one, so
\[
 x_u+x_v\leq|a_ub_u|+|a_vb_v|\leq1.
\]
Together with $0\leq x_t\leq1$, these are the stable-set inequalities
for a bipartite graph \citep{chvatal1975}. Thus
$x\in\operatorname{STAB}(H_\Pi)$.

Let $\mathcal P=\operatorname{conv}\bigcup_\Pi\operatorname{STAB}(H_\Pi)$.
It is downward closed. Constructively, represent $z\in\mathcal P$
by a distribution of feasible independent sets. For $0\leq c\leq z$,
retain coordinate $t$ independently with probability $c_t/z_t$,
and omit it when $z_t=0$. The resulting subset stays feasible
and has expected indicator $c$.
For a mixture of block-product states, its absolute expectation
vector is below the average of their absolute vectors, so also lies
in $\mathcal P$. The contrast vector of any two allowed codewords
is below the mean of their absolute vectors. This proves
$c\in\mathcal P$, including codewords whose partitions and mixtures differ.

\subsection{Attainment}
For an independent set $I$ of $H_\Pi$, define $\tau_g^I$ by
Eq.~\eqref{eq:projector-code86}. Global commutation makes it positive.
Independence means every nonidentity product is traceless, so its
trace is one. It has expectation $(-1)^g$ for selected generators
and zero for every unselected generator.

Restrictions of selected generators commute inside every block.
A simultaneous eigenbasis on each block supplies a product basis
in which $\tau_g^I$ is diagonal and nonnegative. The state is
therefore separable across $\Pi$. For nonempty $I$, the two
codewords occupy opposite eigenspaces of any selected generator
and have orthogonal supports. For $I=\varnothing$, set both codewords
to $I/2^d$.

For weights $\lambda_I$, the pair
$\tau_g=\sum_I\lambda_I\tau_g^I$ has
$\Tr(P_t\tau_g)=(-1)^g\sum_{I\ni t}\lambda_I$.
It realizes every point of $\mathcal P$ exactly, completing
Theorem~\ref{thm:region86}. Maximizing a nonnegative linear workload
functional selects a nonempty indicator (also when all weights vanish),
so an information-preserving optimum always exists.

The equality concerns binary prediction-risk vectors, not equivalence
of full statistical experiments. For example, equal mixtures of the
$XX$ and $ZZ$ projector codes give
$\Tr(\tau_0\tau_1)=1/8$, despite each selected pair being orthogonal.
Randomly choosing a model after training and recording that choice
is distinct from an unobserved mixture on every trial.

\paragraph{A CSS example with incompatible allowed partitions.}
\label{app:css-mixture86}
On six qubits, consider only the three queries
\[
 P_1=XXXXII,\qquad P_2=IIXXXX,\qquad P_3=ZIZIZI.
\]
They commute globally: the respective $X$--$Z$ overlaps are
$\{1,3\}$ and $\{3,5\}$, both of even size. They are independent:
the two $X$ supports are linearly independent, and the nonzero
$Z$ support adds a third independent binary Pauli vector.
Only three generators are queried, so $q=3<d=6$.

Allow the two partitions
\[
 \Pi_A=\{\{1,2\},\{3,5\},\{4,6\}\},\qquad
 \Pi_B=\{\{1,3\},\{2,6\},\{4,5\}\}.
\]
The only edge of $H_{\Pi_A}$ is $\{1,3\}$, while the only edge
of $H_{\Pi_B}$ is $\{2,3\}$. For example, $\Pi_A$ separates the
two sites in $P_1$'s overlap with $P_3$, but keeps the two sites
in $P_2$'s overlap with $P_3$ together. The roles reverse for
$\Pi_B$. The fixed-partition contrast regions are consequently
\[
 \mathcal C_A=\{c\in[0,1]^3:c_1+c_3\leq1\},\qquad
 \mathcal C_B=\{c\in[0,1]^3:c_2+c_3\leq1\}.
\]
Their indicator vertices together comprise every vertex of the
unit cube except $(1,1,1)$. The exact region for the allowed
partition-mixture class is therefore
\[
 \operatorname{conv}(\mathcal C_A\cup\mathcal C_B)
 =\{c\in[0,1]^3:c_1+c_2+c_3\leq2\}.
\]

Put $s_g=(-1)^g$ and define the normalized projector codes
\[
 \tau_g^A=2^{-6}(I+s_gP_2)(I+s_gP_3),\qquad
 \tau_g^B=2^{-6}(I+s_gP_1)(I+s_gP_3).
\]
Each is positive, has trace one and rank sixteen.
Within every block of the stated partition, its two selected
Pauli restrictions commute. Diagonalizing those restrictions on
each block supplies a product basis in which $\tau_g^A$ or
$\tau_g^B$ is a nonnegative diagonal matrix. This proves
separability across $\Pi_A$ or $\Pi_B$, respectively.

The code pairs have contrasts $(0,1,1)$ and $(1,0,1)$.
Their unobserved equal mixture
$\tau_g=(\tau_g^A+\tau_g^B)/2$ has
\[
 c=(1/2,1/2,1).
\]
This point violates each fixed-partition edge inequality by
$1/2$, so no pair separable across either fixed partition can
realize it. Nevertheless, the mixture remains lossless:
\[
 P_3\tau_g=s_g\tau_g,\qquad \tau_0\tau_1=0.
\]
The shared selected query $P_3$ preserves the encoded bit even
when the mixture component is not recorded. With sharp reports
and $h=0$, the three prediction errors are $(1/4,1/4,0)$.
This illustrates both why arbitrary partition mixtures affect
the attainable region and why some mixtures, unlike general
mixtures of lossless codes, still preserve the information.
Every linear workload objective still has an optimal vertex.

The accompanying deterministic $64\times64$ audit checks every
nonempty generator product, local and global commutation,
normalization, positivity, projector identities, Born
probabilities, codeword orthogonality, and the common $P_3$
sectors. It reconstructs each component codeword as a positive
mixture of vectors product across its asserted partition,
and independently enumerates the vertices of the displayed
halfspace region. The reproducibility archive contains the source and complete audit.

\subsection{Relation to single-state witnesses}
Independence makes the symplectic constraints for a Pauli $Q$ with
$QP_tQ^\dagger=-P_t$ solvable. This tensor product of single-qubit
Paulis preserves every partition class.
For any allowed pair, $\rho=(\tau_0+Q\tau_1Q^\dagger)/2$ satisfies
$|\Tr(P_t\rho)|=|\Tr[P_t(\tau_0-\tau_1)]|/2$.
Conversely, $(\rho,Q\rho Q^\dagger)$ realizes the absolute-correlation
vector of any allowed $\rho$. Thus the entire binary-contrast region
equals the single-state absolute-correlation region, not only its
weighted optimum. Introducing two codewords is not a separate
mathematical extension of witness geometry.

Exact feasible polytopes with attaining product-state vertices already
appear in \citet{jafarizadeh2008}, Sections~2.2--3 of the arXiv version,
for selected GHZ and cluster stabilizer families under full separability. Projecting their
GHZ-generator polytope recovers the fully separable geometry of the
opening three-qubit example. We do not claim that example or the
general idea of an exact stabilizer polytope as new.
The linear local-anticommutativity inequality is established in
\citet[Proposition~A.3, Eq.~(A13)]{zander2024}, with its
graph-generator specialization in Proposition~A.4, Eq.~(A21).
We use this bound, rather than claim it as new.
Weighted cut-anticommutativity graphs appear in
\citet{knips2016}; linear graph-generator bounds appear in
\citet[Lemma~1, Eq.~(8)]{li2026restricted}. With $\gamma=0$, their
unweighted graph-generator bound coincides with the support-function
bound used here when the crossing graph is bipartite.
Theorem~\ref{thm:region86} gives a uniform constructive formula for
arbitrary independent commuting CSS families and arbitrary allowed
partition mixtures. Its consequences include optimal depth-constrained
allocation on $K_{3,5}$ and the geometry-dependent learning rates in
Appendix~\ref{app:geometry86}; these are outside the earlier paper's
fully separable state class. Standard concentration and testing supply
the statistics once the physical class is reduced exactly.
Binary discrimination through any quantum channel admits optimal
orthogonal pure-state encodings \citep[Lemma~4]{dallarno2020}. Graph-state clustering via ancillary resources
and rank integrity \citep{bourneuf2026} changes the state-transformation
problem; it does not optimize prediction under these fixed queries.
Biclique graph states themselves are familiar \citep{tzitrin2018}.
Local equivalence of states must track the requested observables
before it can identify two allocation problems.

\citet[Theorem~3, Proposition~4 and Theorem~5]{liu2026magic} give
a dependency-sensitive reduced stabilizer polytope and an independent-set
oracle for its sign relaxation. For perfect graphs without active
dependencies, they compute the reduced robustness of magic exactly.
Their free states are stabilizer mixtures and their graph records
global anticommutation. Here blocks may contain arbitrary states,
and conflicts arise from local restrictions of globally commuting
queries. The global graph would be empty for our CSS family;
its local partition conflicts can still constrain prediction.
The graph-polytope proof architecture is related, while the
physical free classes and resulting optimization problems differ.

\subsection{A sufficient condition beyond CSS}
The same region theorem holds for independent commuting Pauli
generators if every allowed partition satisfies:
(i) $H_\Pi$ is perfect; and (ii) for every maximal clique with
at least two vertices, there are two blocks on each of which the
clique's restrictions pairwise anticommute.
For such a clique $J$, anticommutation gives
$\sum_{t\in J}a_t^2,\sum_{t\in J}b_t^2\leq1$, hence
$\sum_{t\in J}x_t\leq\sum|a_tb_t|\leq1$.
Include singleton inequalities $x_t\leq1$.
Perfection makes these clique constraints describe
$\operatorname{STAB}(H_\Pi)$ \citep{chvatal1975}.
The preceding attainment and mixture arguments are unchanged.
For two-block partitions, global commutation makes condition (ii)
automatic. This is a sufficient condition, not a characterization
of every family admitting an exact combinatorial rule.

The independent non-CSS family $\{XXI,YYZ,ZZI\}$ satisfies both
conditions for every partition. Its conflict graph is $K_3$ when
the first two qubits are in different blocks, with anticommutation
shared by those blocks, and is empty otherwise.

Three failures delimit the result, all with $h=0$ and sharp reports.
For dependent queries $\{ZI,IZ,ZZ\}$, the conflict graph is empty,
but the equally weighted optimal error is $1/6$, not zero:
two joint eigenvalue strings can differ in at most two positions.
For $\{XZZ,ZXZ,ZZX\}$ on singleton blocks, the conflict graph is
perfect ($K_3$), but no common two blocks support its clique.
The orthogonal product pair with antipodal local Bloch vectors
$\pm(1/\sqrt3,0,\sqrt{2/3})$ has risk
$1/2-1/(3\sqrt3)=0.30754991$, below the independent-set formula's
$1/3$. Finally, mixtures of lossless code pairs need not be lossless,
as the $XX/ZZ$ example above shows.

\subsection{Graph components and depth}
For graph-state generators on a bipartite graph, local Clifford
conjugation preserves all commutation relations.
Two generator restrictions conflict across a partition exactly when
their graph edge crosses the partition. A minimum-weight vertex cover
of the crossing edges leaves every surviving connected component
inside one block. Conversely, delete any vertex set whose remaining
components have size at most $k$, group those components, and make
deleted vertices singletons. Every crossing edge is covered.
Therefore, over all depth-$k$ partitions,
\[
 R_k^*=\frac12-\frac D2
 \left[\sum_i a_i-
 \min_{\substack{U\subseteq V\\
       \max_{B\in{\rm cc}(G-U)}|B|\leq k}}\sum_{i\in U}a_i\right].
\]
Extra hardware constraints on partitions can invalidate the converse
construction and must be imposed separately.
Deleted vertices remove query requirements, not qubits.

Prepare each retained graph component in its graph state, with deleted
sites in $|0\rangle$, then apply $Z^g$ to each retained vertex.
Retained generator expectations are $(-1)^g$; a deleted generator has
zero expectation because it contains $X$ on a deleted site.
This pure block code attains the displayed optimum.
For a path, each disjoint group of $k+1$ consecutive vertices requires
a deletion, giving $\lfloor d/(k+1)\rfloor$ deletions under uniform
weights. Deleting every $(k+1)$st vertex attains the bound.
For $D>0$ and sharp reports, attaining risk $h$ requires every active connected
component to fit within a block.

\section{Learning bounds for physical encoders}
\label{app:learning86}
For a fixed measurement family and depth cap $k$, write
\[
 \mathfrak M_m(\Gamma,k)=
 \inf_{\widehat\tau}\sup_{\pi\in\Delta_{q-1}}
 \mathbb E_{T_{1:m}\sim\pi^m,\,\widehat\tau}
       [R_\pi(\widehat\tau)-R_{\pi,k}^*].
\]
The infimum includes randomized maps from $m$ query identities to
any allowed mixed codeword pair, with its optimal classical decoder.
The learner knows the source reliability, measurement family,
reporting parameters and partition constraints. The query distribution
is unknown. Block placement is unrestricted except for the stated
partition family; lower bounds use sharp reports.
\subsection{Finite-class upper bound}
Let $\mathcal A$ be the feasible independent-set indicators and
$N=|\mathcal A|\leq2^q$. The score of an indicator on query $t$ is
$v_t\mathbf1\{t\in I\}\in[0,1]$.
With probability at least $1-\delta$, every score average differs
from its expectation by at most
$\sqrt{\log(2N/\delta)/(2m)}$.
Two deviations and the factor $D/2$ in the risk give
\begin{equation}
 R(\widehat I)-\min_{\tau_0,\tau_1\in\mathcal S_{\mathcal F}}R
 \leq D\sqrt{\frac{\log(2N/\delta)}{2m}},
 \label{eq:erm86}
\end{equation}
capped by $D/2$. The expected version replaces $\delta$ by the
standard finite-class maximal-deviation bound and has
$D\sqrt{\log(2N)/(2m)}$.
Theorem~\ref{thm:region86} makes this a comparison with every allowed
mixed physical encoder. The learner receives only the stipulated
query samples and known interface parameters.
These are standard concentration arguments; exact risk completeness
supplies their physical comparator.

\subsection{Depth-two lower bound}
\label{app:gadget-family86}
For $r$ disjoint three-vertex paths, define the workload by
\[
 \pi_{j,L}=\frac{1+\sigma_j\theta}{4r},\qquad
 \pi_{j,C}=\frac1{2r},\qquad
 \pi_{j,R}=\frac{1-\sigma_j\theta}{4r},
 \quad \sigma_j\in\{-1,1\},\quad0\leq\theta\leq\tfrac12.
\]
The unknown signs determine which endpoint should share the center's
entangled pair. The parameter $\theta$ controls the endpoint imbalance;
$\theta=0$ repeats the opening three-qubit workload.

For $r$ disjoint three-vertex paths, every feasible indicator omits
a query from each gadget: retaining all three would force a component
of size three. Conversely, every subset with at most two queries
per gadget is feasible simultaneously, by using its components as
blocks. Cross-gadget entanglement cannot add another indicator.
The full contrast body is the product of
$\{c\in[0,1]^3:\sum c\leq2\}$.

For this sign family, the optimal
separable score is $1/2$, so its risk is $h+D/4$. The depth-two
oracle omits the less probable endpoint in each gadget and has risk
$h+D(1-\theta)/8$.

For any physical learner, apply the contrast-region replacement
without knowing the workload. Fill its random subset in each gadget
to size two; this cannot lower any contrast. Replace a selection
omitting the center by an equal random choice of the two endpoint
omissions. The expected deleted probability falls from $1/(2r)$
to $1/(4r)$ for every sign vector. Thus it suffices for a lower bound
to consider the $r$ binary pair choices. The transformation depends
on neither the unknown signs nor their probabilities.
A wrong choice in one gadget incurs excess $D\theta/(4r)$.

Neighboring sign vectors differ only in one endpoint pair. Their
one-sample KL divergence is
\[
 \frac{\theta}{2r}\log\frac{1+\theta}{1-\theta}
 \leq\frac{2\theta^2}{r},\qquad 0\leq\theta\leq1/2.
\]
For $m$ observations, Pinsker bounds their total variation by
$\theta\sqrt{m/r}$. Set
$\theta=\min\{1/2,\tfrac12\sqrt{r/m}\}$, making this at most $1/2$.
Averaging uniform signs, each coordinate's error is at least $1/4$.
The usual Assouad testing argument gives
\begin{equation}
 \inf_{\widehat{\tau}}\sup_{\sigma,\theta}
 \mathbb E[R(\widehat{\tau})-R^*]
 \geq\frac D{32}\min\{1,\sqrt{r/m}\}.
 \label{eq:lower86}
\end{equation}
It permits correlated estimator choices and arbitrary randomized
physical learners. Even revealing the common $\theta$ does not
invalidate this testing bound.
The upper bound uses at most $3^r$ full allocations and
Eq.~\eqref{eq:erm86}, proving the local-gadget rate used below.
With no observations, the separate binary-testing argument gives
a constant lower bound.

\subsection{Allocation geometry changes the learning rate}
\label{app:geometry86}
Let the measurement graph be the complete bipartite graph $K_{a,b}$,
with color classes $A,B$, $q=a+b$ qubits and
$1\leq k\leq\min(a,b)$. Allow every partition of depth at most $k$.
A retained query set meeting both colors induces one connected
component, so it fits a block exactly when its size is at most $k$.
Every one-color set is feasible. The exact contrast region is
therefore the convex hull of these indicators.

For nonnegative weights $w_i=\pi_i v_i$, the optimum score is
\begin{equation}
 V_k(w)=\max\left\{\sum_{i\in A}w_i,\ \sum_{i\in B}w_i,\
                         \operatorname{Top}_k(w)\right\},
 \label{eq:biclique-oracle86}
\end{equation}
where $\operatorname{Top}_k$ sums the $k$ largest coordinates.
A monochromatic top-$k$ set is dominated by its full color class.
Evaluating the three candidates costs $O(q\log(k+1))$ with a heap,
or $O(q)$ with linear-time selection. This is an elementary optimizer
for the stated complete physical encoder class, not a quantum speedup.

\paragraph{Only undominated allocations need to be learned.}
The maximal feasible sets are
\[
 \mathcal M_k=\{A,B\}\cup
 \{S:|S|=k,\ S\cap A\ne\varnothing,\ S\cap B\ne\varnothing\}.
\]
Their indicator class has VC dimension exactly $k$.
A set of $k$ vertices inside $A$ is shattered: $A$ and $B$ give the
full and empty traces; complete each nonempty proper subset to a mixed
$k$-set using vertices in $B$.
No $(k+1)$-set is shattered. If it meets both colors, its all-positive
trace is unavailable. If it is monochromatic, a trace with exactly
$k$ positives is unavailable: a mixed $k$-set has at most $k-1$
vertices of that color, whereas the full color includes all $k+1$.
This also covers $k=1$, when the frontier is simply $\{A,B\}$.

The full downward-closed feasible class has VC dimension
$\max(a,b)$, since every subset of a color is feasible, while a larger
set necessarily crosses the colors and cannot be included.
Thus counting all feasible indicators, including dominated ones,
can overstate the statistical difficulty by an arbitrarily large factor.

\paragraph{The deviation bound and risk conversion.}
\label{app:vc86}
For a set class $\mathcal M$ of VC dimension $V\geq1$ and fixed
$v_t\in[0,1]$, put $f_I(t)=v_t\mathbf1_I(t)$.
The standard expected VC bound is
\[
 \mathbb E\sup_{I\in\mathcal M}|(\widehat\pi-\pi)f_I|
 \leq C_0\min\{1,\sqrt{V/m}\}
\]
for a universal constant $C_0$.
For $v=1$, this follows by symmetrization from the Gaussian VC bound
of \citet[Theorem~6 and Lemma~4]{bartlett2002}, using the entropy
bound of \citet{haussler1995}.
Known visibility cannot increase empirical $L_2$ distances:
$\|v(\mathbf1_I-\mathbf1_J)\|_{L_2(\widehat\pi)}
\leq\|\mathbf1_I-\mathbf1_J\|_{L_2(\widehat\pi)}$.
The same covering-number argument therefore applies.
If $I^*$ maximizes the population score and $\widehat I$ maximizes
the empirical score, then
\[
 R_\pi(\widehat I)-R_\pi(I^*)
 =\frac D2\pi(f_{I^*}-f_{\widehat I})
 \leq D\sup_{I\in\mathcal M}|(\widehat\pi-\pi)f_I|.
\]
Taking expectations and setting $V=k$ gives the claimed rate.
Nonnegative weights make maximal sets sufficient; exact region
completeness supplies the full physical comparator.

Capacity-controlled learning from request identities is also familiar
in classical caching. For fractional caching with unit hit utilities
and capacity $C<N/2$ in an $N$-file catalogue,
\citet[Theorem~2]{paschos2019} give cumulative regret at most
$\sqrt{2CT}$ against the best fixed cache, independent of $N$. Here the objective is batch
excess risk, and the physical region determines a top-$k$/color-class
frontier and a lower bound over arbitrary mixed encoders.

\paragraph{A matching lower bound over physical learners.}
Set $v_i=1$. Choose $k$ active vertices in each color and pair them.
For unknown signs $\sigma_j$ and $0\leq\theta\leq1/2$, assign
\[
 \pi_{j,A}=\frac{1+\sigma_j\theta}{2k},\qquad
 \pi_{j,B}=\frac{1-\sigma_j\theta}{2k},
\]
with zero probability on all other queries.
Projection of the exact region onto these $2k$ coordinates is
$\{c\in[0,1]^{2k}:\sum_i c_i\leq k\}$.
Every physical learner can therefore be replaced, without knowing
the signs, by a randomized subset of at most $k$ active coordinates.
Fill it to size $k$. The number of doubly selected pairs then equals
the number of omitted pairs. Replace each such double and omission
by one uniformly selected endpoint in each pair.
Every pair has total mass $1/k$, so this transformation preserves
expected captured mass for every sign vector, conditionally on the
original subset. It need not preserve risks outside this hard family.

A wrong endpoint choice incurs excess $D\theta/(2k)$.
Neighboring signs have one-sample KL divergence
\[
 \frac{\theta}{k}\log\frac{1+\theta}{1-\theta}
 \leq\frac{4\theta^2}{k}.
\]
For $\theta=\min\{1/2,\tfrac14\sqrt{k/m}\}$, their $m$-sample
total variation is at most $\sqrt2/4<1/2$.
The coordinate testing error is at least $1/4$, giving
\begin{equation}
 \inf_{\widehat\tau}\sup_\pi
 \mathbb E[R(\widehat\tau)-R_{\pi,k}^*]
 \geq \frac D{32}\min\{1,\sqrt{k/m}\}.
 \label{eq:biclique-lower86}
\end{equation}
The upper and lower bounds establish the sharp rate.
They allow arbitrary workload distributions, including zeros,
and unrestricted block placement subject to the depth cap.

\paragraph{Equal memory size does not imply equal learning cost.}
A connected $q$-vertex path, $q\geq3$, contains
$r=\lfloor(q+1)/4\rfloor$ length-three active gadgets separated by
zero-probability connector queries. Extra trailing queries also
have probability zero. Active generators can act with $Z$ on connector
qubits; nevertheless, their projected contrast region is the product
of the gadget polytopes. Keeping a connector cannot make three
consecutive active vertices fit a depth-two block, and deleting all
connectors realizes every permitted combination.
Since $r\geq q/6$, Eq.~\eqref{eq:lower86} gives the conservative lower
bound $D\min\{1,\sqrt{q/m}\}/(32\sqrt6)$.
The general upper bound matches its dependence on $q,m$.

At the same depth $k=2$, complete bipartite graphs consequently have
minimax rate $\Theta(D\min\{1,m^{-1/2}\})$, whereas connected paths
have $\Theta(D\min\{1,\sqrt{q/m}\})$.
This compares connected measurement graphs, not connected positive-weight
workloads. Path queries have support at most three; complete-bipartite
queries have support up to $\max(a,b)+1$. The comparison does not
hold detector cost or locality fixed. Its conclusion concerns the
statistical geometry of useful allocations, not universal learning
difficulty as a function of physical memory size.

\paragraph{An eight-query allocation example.}
For $K_{3,5}$, take
$\pi_A=(.30,.15,.05)$ and
$\pi_B=(.30,.05,.05,.05,.05)$.
Both color totals equal $.50$. At depth three the optimum retains
$\{A_1,A_2,B_1\}$, with captured mass $.75$.
Prepare the graph state on $A_1-B_1-A_2$, put the other five sites in
$|0\rangle$, and apply $Z^g$ to the selected sites.
The original selected generators reduce to this star's stabilizers;
deleted generators have zero expectation.
At $h=0$, optimal errors are $25\%,20\%,12.5\%$ for caps one, two
and three. These are optima over all corresponding mixed physical
encoders. Original queries on $A$ have support six and those on $B$
have support four.

\subsection{Bounded query locality forces dimensional learning cost}
\label{app:locality86}
Let $\Gamma$ be a bipartite graph on $d$ vertices, with sharp
graph-generator queries and depth cap two. Let $p_3(\Gamma)$ be the
largest number of vertex-disjoint induced three-vertex paths with
no edges between different paths. Then
\begin{equation}
 \mathfrak M_m(\Gamma,2)\geq
 \frac D{32}\min\{1,\sqrt{p_3(\Gamma)/m}\}.
 \label{eq:packing-lower86}
\end{equation}
Choose any such collection of $r$ paths and put the workload on
their $3r$ vertices. A feasible depth-two query set retains at most
two vertices per path. Conversely, every such choice is feasible:
pair adjacent retained vertices, leave nonadjacent endpoints in
singletons, and put inactive qubits in singleton blocks.
There are no edges between the active paths. Edges to inactive
vertices impose no additional selected-query constraint.
The projection of the full physical contrast region is therefore
the product of $r$ gadget polytopes. The workload-independent
physical reduction and testing argument of
Eq.~\eqref{eq:lower86} apply unchanged. Maximizing over packings proves
Eq.~\eqref{eq:packing-lower86}.

Suppose now that $\Gamma$ is connected, $d\geq3$, and its maximum
degree is $\Delta$. Greedily select an induced three-vertex path
and remove its closed neighborhood. Repeat until none remains.
Bipartiteness ensures that any vertex with two remaining neighbors
supplies such a path. The chosen paths are mutually nonadjacent.
At termination every remaining component has at most two vertices.
If $R$ is the union of removed neighborhoods and $r$ paths were
chosen, then $|R|\leq3r(\Delta+1)$.
Every remaining component has an edge to $R$, since the original
graph is connected. There are at most $\Delta|R|$ such components,
each with at most two vertices. Hence
\[
 d\leq(1+2\Delta)|R|
   \leq3r(\Delta+1)(2\Delta+1).
\]
Writing $C_\Delta=3(\Delta+1)(2\Delta+1)$ gives
\[
 \mathfrak M_m(\Gamma,2)\geq
 \frac{D}{32\sqrt{C_\Delta}}\min\{1,\sqrt{d/m}\}.
\]
The finite-class upper bound matches its dependence on $d,m$.
Thus every connected bipartite family with fixed maximum degree
has rate $\Theta_\Delta(D\min\{1,\sqrt{d/m}\})$ at depth two.
The constants are not claimed sharp in $\Delta$.
Graph-generator support is one plus vertex degree, so this is also
a bounded-query-support result. A dimension-free rate along a
growing connected family requires leaving the bounded-degree regime;
growing degree alone is not sufficient.

\paragraph{General upper and lower structural bounds.}
For any bipartite graph, let $V_2(\Gamma)$ be the VC dimension of its
inclusion-maximal depth-two feasible query sets.
The deviation argument in Appendix~\ref{app:vc86} gives
\[
 \frac D{32}\min\{1,\sqrt{p_3(\Gamma)/m}\}
 \leq\mathfrak M_m(\Gamma,2)
 \leq C_0D\min\{1,\sqrt{V_2(\Gamma)/m}\},
\]
for a universal $C_0$ when $V_2(\Gamma)\geq1$; a singleton frontier
has zero excess. These bounds need not match on every graph.
The upper bound assumes exact empirical optimization and makes no
polynomial-time claim for arbitrary graphs.

\paragraph{Positive workloads.}
The constructions above use zero query probabilities.
The same minimax supremum holds over strictly positive workloads
if no common probability floor is imposed. For a fixed learner,
expected excess is continuous in $\pi$: the sample space is finite,
each encoding's risk is linear in $\pi$, and the oracle risk is
continuous as the optimum over a finite contrast polytope.
Strictly positive distributions are dense in the simplex, so their
supremum equals the unrestricted supremum for every learner.
Taking the infimum preserves that equality.
This does not establish the same rate under a fixed constraint
$\pi_i\geq c/d$ with $c>0$.

\paragraph{Independent structural checks.}
The development audit checks 285 graphs: all 70 connected bipartite
atlas graphs on three to seven vertices, plus paths, ladders, grids,
bounded-degree random trees and regular bipartite graphs up to 192
vertices. Exhaustive depth-two partitions verify 784 projected masks;
all large-graph packing and degree inequalities pass.
A separate implementation using Pauli symplectic restrictions checks
all 11,664 depth-two partitions and 775 induced-path packings on
the 70 atlas graphs, covering 6,536 projected masks.
These finite checks corroborate the proof; they do not replace it.

\subsection{The locality law at every fixed depth}
\label{app:fixed-depth86}
The depth-two argument extends to every positive depth cap $k$.
Let $\Gamma$ be connected and bipartite, with $d\geq k+1$ vertices
and maximum degree at most $\Delta$. Put
$C_{k,\Delta}=(k+1)(\Delta+1)(1+k\Delta)$.
For sharp graph-generator reports and $m\geq1$,
\begin{equation}
 \frac{D}{32k}\min\!\left\{1,
       \sqrt{\frac{kd}{C_{k,\Delta}m}}\right\}
 \leq\mathfrak M_m(\Gamma,k)
 \leq D\min\!\left\{\frac12,
       \sqrt{\frac{(d+1)\log2}{2m}}\right\}.
 \label{eq:fixed-depth86}
\end{equation}
The lower bound permits arbitrary randomized learners selecting
allowed mixed codeword pairs and grants their optimal decoder.
The upper bound is the finite-class argument in
Appendix~\ref{app:learning86}; empirical maximization over maximal
feasible sets attains it with lossless codes.
If $d\leq k$, retaining the whole graph gives zero allocation excess.

\paragraph{Connected gadgets.}
Greedily choose a connected induced subgraph of $k+1$ vertices in
the remaining graph, then remove its closed neighborhood.
Such a set exists whenever a remaining component has more than
$k$ vertices: grow a connected set until it has the required size.
The chosen sets are mutually nonadjacent. At termination, every
remaining component has at most $k$ vertices.
For $r$ chosen sets and removed union $R$,
$|R|\leq(k+1)(\Delta+1)r$.
Original connectivity ensures that every remaining component has
an edge into $R$, so there are at most $\Delta|R|$ of them. Therefore
\[
 d\leq(1+k\Delta)|R|\leq C_{k,\Delta}r.
\]

Put the workload on the $r$ chosen sets. A feasible query indicator
must omit at least one vertex of each set: retaining a connected
$(k+1)$-set exceeds the depth cap. Conversely, every proper subset
of each gadget is feasible. Its connected components have size at
most $k$, the gadgets have no edges between them, and inactive
queries can be omitted while their physical qubits remain in the
memory. The exact region thus projects to
\[
 \prod_{j=1}^r
 \left\{c^{(j)}\in[0,1]^{k+1}:\sum_i c_i^{(j)}\leq k\right\}.
\]
These are uniform-matroid independence polytopes. Partitions crossing
gadgets and arbitrary mixed codewords cannot enlarge this projection.

\paragraph{A workload-independent physical reduction.}
In gadget $j$, give two distinguished vertices probabilities
$(1+\sigma_j\theta)/(2kr)$ and $(1-\sigma_j\theta)/(2kr)$,
where $\sigma_j\in\{-1,1\}$ and $0\leq\theta\leq1/2$.
Each of the other $k-1$ vertices has probability $1/(kr)$.
Every gadget has total mass $1/r$.
Its optimum omits the less frequent distinguished vertex; the
oracle risk is $h+D(1-\theta)/(4k)$.

Decompose any physical learner's projected contrast into subset
indicators and fill each subset to size $k$. This weakly improves
every nonnegative workload objective and requires no knowledge of
$\sigma$. If such an indicator omits an ordinary vertex, replace
that omission by a fair random distinguished omission. Expected
omitted mass falls from $1/(kr)$ to $1/(2kr)$ for every $\sigma$.
Thus a randomized learner making one binary omission choice per
gadget dominates the original learner throughout this hard family.
Its choices may be correlated across gadgets. At $k=1$, there are
no ordinary vertices and this second replacement is unnecessary.

\paragraph{Testing and risk.}
A wrong distinguished omission costs $D\theta/(2kr)$ risk.
Flipping one sign gives one-sample divergence
\[
 \operatorname{KL}(\pi_\sigma,\pi_{\sigma^{(j)}})
 =\frac{\theta}{kr}\log\frac{1+\theta}{1-\theta}
 \leq\frac{4\theta^2}{kr}.
\]
Pinsker's inequality for $m$ samples gives total variation at most
$\theta\sqrt{2m/(kr)}$. Choose
$\theta=\min\{1/2,\sqrt{kr/(8m)}\}$, so this is at most $1/2$.
The average binary testing error is at least $1/4$ in each coordinate.
Summing the $r$ coordinate losses yields
\[
 \mathfrak M_m(\Gamma,k)\geq\frac{D\theta}{8k}
 \geq\frac{D}{32k}\min\{1,\sqrt{kr/m}\}.
\]
Substituting the packing bound proves Eq.~\eqref{eq:fixed-depth86}.
For fixed $k,\Delta$, its lower coefficient can be taken as
$1/(32\sqrt{kC_{k,\Delta}})$ times
$D\min\{1,\sqrt{d/m}\}$. This matches the upper rate in $d,m$;
the simultaneous dependence on growing $k$ or $\Delta$ is not sharp.
The positive-workload continuity argument in
Appendix~\ref{app:locality86} still applies without a common
probability floor. Bipartiteness remains necessary for the physical
reduction used here, although the packing argument alone does not
require it.

\paragraph{Structural verification.}
A frozen audit checks 737 graph/depth conditions on 151 graphs,
8,340 projected masks and 280 testing-constant settings, with depths
one to five and graph sizes up to 192. A separate implementation
using direct Pauli symplectic restrictions checks all 71 connected
bipartite atlas graphs on two to seven vertices: 42,351 partitions,
3,096 gadget packings and 64,532 projected masks, including $k=1$
and the $d\leq k$ boundary. These checks corroborate the proof.
The reproducibility archive retains both implementations and their complete results.

\subsection{Connected regions and the effective allocation dimension}
\label{app:modular86}
The local and complete-bipartite laws extend to a connected family
between those extremes. Partition the vertices of a bipartite graph
$\Gamma$ into $r$ regions $M_j=A_j\cup B_j$, each inducing a complete
bipartite graph. A designated set $P_j\subset M_j$ contains its
connection vertices: every edge between regions has both endpoints
in these sets. Assume $|P_j|\leq p$ and that each of
$A_j\setminus P_j$ and $B_j\setminus P_j$ contains at least $k$
vertices. The connections between regions can be arbitrary, provided
$\Gamma$ remains bipartite. The word region denotes this chosen vertex
partition; the physical encoder may use blocks crossing its boundaries.

For sharp graph-generator reports, $m\geq1$ request observations,
and all physical encoders of depth at most $k$,
\begin{equation}
 \frac D{32}\min\{1,\sqrt{rk/m}\}
 \leq \mathfrak M_m(\Gamma,k)
 \leq C D\min\left\{1,
       \sqrt{\frac{rk+\sum_j|P_j|}{m}}\right\},
 \label{eq:modular86}
\end{equation}
where $C$ is universal. For fixed $p$, this is
$\Theta_p(D\min\{1,\sqrt{rk/m}\})$, with constants independent
of the region sizes, $r$ and $k$. Both $r$ and $k$ may grow.
The upper bound is attained by exact empirical maximization over
maximal feasible query sets; every output has a lossless code.
The lower bound permits arbitrary randomized mixed physical encoders
and grants their optimal decoder.

\paragraph{The maximal class has dimension at most $rk+\sum_j|P_j|$.}
The exact physical region is the convex hull of indicators $S$ whose
selected connected components have size at most $k$.
Consider an inclusion-maximal feasible $S$. If it meets both colors
of region $j$, its selected region vertices lie in one component,
so there are at most $k$ of them. If it meets only one color,
every nonconnection vertex of that color must be selected: a missing
one would have no selected neighbor and could be added as a singleton.
Edges between connection vertices do not change this argument.

No $k+1$ nonconnection vertices in one region can be shattered by
the maximal class. A test set meeting both colors lacks the
all-positive trace, which would require a component larger than $k$.
A monochromatic test set lacks a trace with exactly $k$ positives:
a monochromatic selection includes the entire test set, while a mixed
selection containing those $k$ positives also needs an opposite-color
vertex. Its component would again exceed $k$.
The projected maximal class on each region interior therefore has
VC dimension at most $k$. A globally shattered set has at most $k$
interior vertices per region and at most $\sum_j|P_j|$ connection
vertices. The VC bound and the risk conversion in
Appendix~\ref{app:vc86} prove the upper inequality.

\paragraph{Connections do not remove the capacity lower bound.}
Choose $k$ nonconnection vertices from each color of each region.
Projection onto these $2rk$ active coordinates is exactly
\[
 \prod_{j=1}^r
 \left\{c^{(j)}\in[0,1]^{2k}:\sum_i c_i^{(j)}\leq k\right\}.
\]
Every feasible indicator satisfies each capacity inequality. A
monochromatic trace has only $k$ active vertices available, while
a mixed trace lies in one component of size at most $k$.
Conversely, select any at-most-$k$ active vertices per region and
no others. No selected interregion edge remains, and each component
fits the cap. Convexification gives the product of the uniform-matroid
independence polytopes. This proof uses the complete physical region;
partitions crossing regions cannot enlarge the projection.

Pair the active colors within each region and assign
\[
 \pi_{j,\ell,A}=\frac{1+\sigma_{j\ell}\theta}{2rk},\qquad
 \pi_{j,\ell,B}=\frac{1-\sigma_{j\ell}\theta}{2rk},
 \quad \sigma_{j\ell}\in\{-1,1\},\quad0\leq\theta\leq\tfrac12.
\]
All other queries have zero probability. Decompose any physical
learner's projected contrast into product subset indicators and
fill each region to size $k$. Within each region, pair each doubly
selected pair with an omitted pair and replace them by uniform
single-endpoint choices. Every pair has total mass $1/(rk)$, so
this replacement preserves expected captured mass for every sign
vector. Filling can only improve it. The resulting learner chooses
one endpoint in each of $rk$ pairs without knowing the signs.

A wrong choice costs $D\theta/(2rk)$ in excess risk. Neighboring
signs have one-sample divergence at most $4\theta^2/(rk)$.
Choosing $\theta=\min\{1/2,\tfrac14\sqrt{rk/m}\}$ bounds
sample total variation by $\sqrt2/4<1/2$ in both branches.
The average testing error is at least $1/4$ per pair, giving
$D\theta/8\geq(D/32)\min\{1,\sqrt{rk/m}\}$.
The oracle chooses the more frequent endpoint in every pair and
has risk $h+D(1-\theta)/4$. This argument also covers $k=1$.

\paragraph{A positive-weight allocation across regions.}
Take two $K_{3,3}$ regions with colors $A_1,B_1,A_2,B_2$ and a
single edge from $b\in B_1$ to $a\in A_2$. The weights on $B_1$
and $A_2$ are $(6,5,5)$, with six at $b,a$; all six other weights
are one. Normalize by the total 38. At depth two, retaining both
preferred colors supports weight 32. Only $a,b$ share a component;
the four other retained vertices are singletons. Any mixed selection
within a region contains at most two vertices and has weight at most
seven, while the less-favored full color has weight three. Thus 32
is the unique optimal retained-query weight pattern. At depth one,
one must also discard $a$ or $b$, reducing the best weight to 26.
Switching a whole region to its other color is worse. The exact region
therefore gives optimal errors $3/38$ and $3/19$, respectively, over
the full physical classes. Its projector codes retain the same bit.
Independent enumeration of all $2^{12}$ query subsets at both caps
reproduces these values.

\paragraph{Why the connection bound matters.}
Join two $K_{s+1,s+1}$ regions by a matching between $s$ vertices of
$B_1$ and $s$ vertices of $A_2$. Each facing color still has one
nonconnection vertex. At depth one, projection onto the $2s$ matching
endpoints is exactly a product of $s$ two-element rank-one capacities:
one cannot select both endpoints of an edge, and every other selection
is feasible by omitting all nonmatching vertices. The paired-testing
argument above now gives $\Omega(D\min\{1,\sqrt{s/m}\})$;
the finite-class upper bound on $4(s+1)$ vertices matches it.
Thus fixed $r=2$ and $k=1$ do not give a size-independent rate when
the connection count grows. The same independent source exhaustively
checks the projection for $s=1,2,3$; the argument establishes it for
every $s$. This boundary is an application of the capacity lower bound.

\paragraph{An exact optimizer for a chain of regions.}
Consider a chain of $K_{s,s}$ regions, each with a left connection
vertex in $A$ and a right one in $B$. Consecutive right and left
vertices are joined. For $s\geq k+1$, Eq.~\eqref{eq:modular86}
gives the sharp rate $\Theta(D\min\{1,\sqrt{rk/m}\})$ on
$d=2rs$ qubits. Increasing $s$ need not increase this statistical
cost; increasing $r$ does. At fixed $s,k$, the rate is dimensional,
consistent with the bounded-degree law. Larger $s$ also increases
query support, so this comparison does not fix detector cost.

The empirical optimum can be computed without enumerating all
allocations. Process regions from left to right, with state
$t\in\{0,\ldots,k\}$ equal to the size of the selected component
touching the preceding right connection vertex; $t=0$ means that
vertex is absent. A mixed-color option with $n$ selected vertices
and left/right selection bits $L,R$ forms one local component.
Set $z=n+Lt$, reject $z>k$, and send state $Rz$ to the next region.
If $L=0$, the incoming component closes.

For an $A$-only option, include all nonconnection $A$ vertices and
optionally its left connection vertex. Those interior vertices are
separate singletons. Selecting the left vertex joins only one vertex
to the incoming component, so require $t+1\leq k$; the outgoing
state is zero. A $B$-only option similarly includes its interior
singletons and has outgoing state zero or one, according to whether
the right vertex is selected. Empty options can be included for
boundary cases with no interior vertices.

Sort the interior weights in each color. For every mixed size and
connection pattern, prefix sums give the best split between colors.
There are $O(k^2)$ candidate splits but only $O(k)$ resulting options;
$O(k)$ incoming states give $O(rk^2)$ scalar work across the chain,
after $O(d\log d)$ sorting. This is a sufficient boundary state,
so induction proves the recurrence exact. Equal objectives prefer
greater selected cardinality, ensuring a maximal optimum; a fixed
mask order resolves remaining ties. Integer-mask costs are additional.
The implementation uses floating-point comparisons and does not
certify exact ordering of nearly equal real objectives.

\paragraph{Independent checks.}
The released chain optimizer passes
1,300 exhaustive weighted objectives on 22 graphs, including 1,000
exact integer-weight cardinality and maximality checks.
A separate implementation checks 336 connected port graphs,
9,496 capped physical partitions of a ten-qubit example, and
20,109 weighted optimizations. It compares the released optimizer
against 44,938 exhaustive graph-oracle workloads, including zero
weights, asymmetric regions, empty interiors and $k=1$.
These finite checks support the derivations above; they do not
establish the asymptotic claims. Both implementations and their complete results are in the reproducibility archive.

\section{Preparation noise, algorithms and preserved information}
\label{app:noise86}
\subsection{Exact circuit-family objective}
For a fixed selected set $S$, prepare graph-state runs, put deleted
sites in $|0\rangle$, and apply $Z^g$ to every selected site.
A $ZZ$ error after a retained edge flips precisely its two incident
generator expectations. Independent errors therefore give
Eq.~\eqref{eq:noisy-contrast86}. An unselected expectation remains zero.
The model assumes $0\leq\gamma_e,v_i\leq1$.

Reporting noise is a classical postprocessing of the parity report.
Writing $v_i$ as a product of participating-site factors is an error
budget model; it does not require destructive local measurements on
the retained memory. For path queries and site rates $p_j$,
$v_i=\prod_{j\in{\rm supp}(K_i)}(1-2p_j)$ has at most three factors.
The nonselective detector contract and the noisy report remain
distinct specifications.

\subsection{Dynamic programming}
After deciding the first $i$ sites, retain the best partial score for
each current run length $\ell\in\{0,\ldots,k\}$.
When choosing the next bit $b\in\{0,1\}$, finalize the previous site's
contribution
\[
 \mathbf1\{\ell>0\}\,w_i v_i\,
 \gamma_{i-1}^{\mathbf1\{\ell\geq2\}}\,
 \gamma_i^{b}.
\]
The new length is $\ell+1$ for $b=1$ and zero for $b=0$.
Reject lengths exceeding $k$, and append a final zero to close the
last run. Optimal substructure follows because future contributions
depend only on the current run length. This gives $O(dk)$ transitions;
backpointers give an implementation without growing-mask arithmetic.
Our small-instance implementation retains integer masks for deterministic
ties, whose bit-operation cost is additional to the transition count.
The all-odd control rejects a zero that closes a positive even run.
The primary family requires only one odd run: retain an additional
Boolean flag, set it whenever a positive odd run closes, and accept
only final states with that flag set. Even runs remain allowed.
This doubles the state count and preserves the $O(dk)$ transition bound.

An independent verifier uses interval scheduling: at site $i$, either
skip it, or retain a run of allowed length $\ell$ and continue after
its separator. It explicitly sums that run's weighted noisy contrasts.
Agreement of these two recurrences checks optimization separately
from direct Born-rule evaluation of the physical model.

\subsection{Information equality after preparation errors}
For a retained odd component $B$, let $W_B=\prod_{i\in B}K_i$.
The codeword before noise satisfies
$W_B|\psi_g\rangle=(-1)^{g|B|}|\psi_g\rangle=(-1)^g|\psi_g\rangle$.
Every retained-edge $ZZ$ commutes with $W_B$, since it flips two
constituent stabilizer signs. The noisy codewords remain in opposite
$W_B$ eigenspaces. They therefore preserve the exact classical
experiment $G$, including for mixed noisy states.
One odd component suffices; requiring all components odd is a clean
stronger control. The support argument permits correlated mixtures
or coherent operations generated by the $ZZ$ errors, although the
contrast product assumes independent errors.

A fixed nonempty architecture is required. A hidden mixture over
architectures can have different distinguishing observables and
need not preserve the bit. In the experiment, training selects one
architecture and records it before future trials.
For an even two-site component,
$\tau_g'=(1-p)\tau_g+p\tau_{1-g}$, so the codewords overlap at
$0<p<1/2$. At $p=.2$, their trace overlap is $.32$.
This counterexample is retained beside the odd-code control.

\subsection{Learning with calibrated parameters}
Within any finite selected-mask class $\mathcal A$, known parameters
give the same bound as Eq.~\eqref{eq:erm86}, with $|\mathcal A|\leq2^d$.
It compares with the best circuit in that family, not every noisy
physical encoder. If independent calibration yields parameters in
$[0,1]$ satisfying
\[
 |v_i-\widehat v_i|
 +\sum_{e\ni i}|\gamma_e-\widehat\gamma_e|\leq e_i,
\]
then $|b_i(S)-\widehat b_i(S)|\leq e_i$ uniformly over masks.
Adding and subtracting risks under the estimated model adds
$D\sum_i\pi_i e_i$ to the excess bound. Clipping estimated parameters
to $[0,1]$ preserves the absolute-error guarantees.
The original experiments use known parameters; they do not estimate
these calibration terms or provide a device-valid certificate.

\subsection{A homogeneous preparation-error threshold}
For uniform queries on a long path, a retained run of length $\ell$
followed by one deleted vertex has contrast density
$v f_\ell/(\ell+1)$, where $f_1=1$ and
$f_\ell=2\gamma+(\ell-2)\gamma^2$ for $\ell\geq2$.
For those longer runs,
\[
 \frac{f_\ell}{\ell+1}
 =\gamma^2+\frac{\gamma(2-3\gamma)}{\ell+1}.
\]
For $\gamma\leq2/3$, every longer run is inferior to singleton
density $1/2$. For $\gamma>2/3$, its density increases with length.
At cap $k$, the optimal asymptotic density is therefore the larger
of $1/2$ and $[2\gamma+(k-2)\gamma^2]/(k+1)$.
Mixtures of run types average these densities with weights
proportional to $\ell+1$, including separators; extra deleted sites
contribute zero. Boundary terms vanish with path length.

For odd runs, replace $k$ by the largest allowed odd length $K$.
At $K=3$, the switch occurs at $\gamma=\sqrt3-1$, or
$\epsilon=(2-\sqrt3)/2$, when $Dv>0$.
The full odd chain has asymptotic contrast $v\gamma^2$.
Under homogeneous noise and no depth cap, small intermediate groups
do not improve on both the product and full-chain alternatives.
Heterogeneous costs permit the more selective choices tested in the
workload study. This is a statement about the specified preparation
channel, not a general runtime or gate-infidelity threshold.

\FloatBarrier
\section{Prediction under logarithmic loss}
\label{app:logloss90}
Logarithmic loss evaluates a probability assigned to the realized label,
after observing the realized detector output. For a code pair $\tau$, define
\[
 L_\pi(\tau)=\inf_q\mathbb E[-\log_2 q(C\mid T,Y)]
             =H(C\mid T,Y),\qquad
 J_\pi(\tau)=1-L_\pi(\tau)=I(C;Y\mid T).
\]
Here $q$ ranges over all classical probabilistic decoders, $T$ is
independent of $C$, and the expectation includes the stochastic detector.
This is an operational restricted-prediction information quantity. It
uses the same log-loss principle as predictive $\mathcal V$-information
\citep{xu2020usable}, with the physical experiment stated explicitly.
It does not score a forecast obtained by averaging unobserved detector
outputs before observing a label.

\begin{corollary}[Optimal logarithmic loss]
\label{cor:logloss90}
Under the assumptions of Theorem~\ref{thm:region86}, with symmetric
report visibility $v_t$, put
$\beta_t=1-H_2((1-Dv_t)/2)$. Then
\begin{equation}
 \inf_{\tau_0,\tau_1\in\mathcal S_{\mathcal F}}L_\pi(\tau)
 =1-\max_{\Pi\in\mathcal F,\ I\in\mathcal I(H_\Pi)}
                                      \sum_{t\in I}\pi_t\beta_t.
 \label{eq:logloss90}
\end{equation}
An existing projector code attains the optimum and preserves $G$.
\end{corollary}
\begin{proof}
Fix a query and orient its sharp parity $S\in\{+1,-1\}$ so that
$a_0=\mathbb E[S\mid G=0]\geq a_1=\mathbb E[S\mid G=1]$.
Set $c=(a_0-a_1)/2$ and $\mu=(a_0+a_1)/2$. The binary channel
from $G$ to $S$ is a mixture of three channels: $S=(-1)^G$
with weight $c$, constant $+1$ with weight $(1+\mu-c)/2$,
and constant $-1$ with weight $(1-\mu-c)/2$. The weights are
nonnegative because $|a_g|\leq1$. Introduce the mixture selector
independently of $(C,G)$. Reveal that choice as an extra output.
The informative branch, composed with the input crossover and
report noise, is a binary symmetric channel of correlation $Dv_t$;
the constant branches carry no information. Discarding the extra
output cannot increase information, so
\begin{equation}
 I(C;Y\mid T=t)\leq c_t\beta_t.
 \label{eq:information-majorant90}
\end{equation}
Theorem~\ref{thm:region86} bounds the weighted sum by its best
independent-set indicator. For the projector code of that set,
selected parities equal $(-1)^G$, and independence of the generators
makes every unselected parity uniform and independent of $G$.
Thus every selected query attains $\beta_t$ and every unselected
query contributes zero. If the optimal value is zero, any permitted
nonempty singleton code is also optimal and preserves $G$.
\end{proof}

\subsection{What contrast does and does not determine}
The corollary optimizes over codewords; it does not give the log loss
of every code from its contrast alone. With $h=0$, $v=1$ and $c=1/2$,
the expectation pairs $(a_0,a_1)=(1/2,-1/2)$ and $(1,0)$ both
have binary prediction error $1/4$, but their optimal log losses
are $0.811278$ and $0.688722$ bits. Both can be embedded in
fully separable, information-preserving encodings by adding an
orthogonal unmeasured flag. Output bias matters for codewise
information; Eq.~\eqref{eq:information-majorant90} handles it.

The loss can also change the chosen representation. On two qubits,
let the queries be $XX$ and $ZZ$, restrict to product codes, and
take $h=0$, $\pi=(0.3,0.7)$ and $v=(1,0.5)$. The feasible
indicator choices support either query. Accuracy chooses $ZZ$:
its weighted contrast is $0.35$, compared with $0.30$ for $XX$.
Log loss chooses $XX$: its information is $0.30$ bits, compared
with $0.7[1-H_2(3/4)]=0.132105$ bits for $ZZ$.
Both chosen codes store the original bit perfectly.

\subsection{Learning under logarithmic loss}
With sharp reports, $\beta_t=J_0=1-H_2(h)$ for every query.
The learning laws in Theorem~\ref{thm:learning86} then hold for
log-loss excess with $J_0$ replacing $D$, up to universal factors.
For the upper bound, empirical maximization returns the same indicator
code and its log-loss excess is $2J_0/D$ times its binary-risk excess.
For the lower bound, Eq.~\eqref{eq:information-majorant90} gives
every physical learner information at most $J_0\pi\cdot c$.
The workload-independent complete-class reductions in
Appendices~\ref{app:geometry86}--\ref{app:modular86} and
\ref{app:growing-depth90} replace that
contrast by a randomized indicator with no smaller expected score
on the hard family. The same testing lower bounds apply, with the
same rescaling. Thus the result covers asymmetric mixed codewords
as well as the attaining projector codes. Near $h=1/2$, $J_0$
is quadratic in $D$; the two loss scales are not interchangeable.

\section{A sharp path law when entanglement depth grows}
\label{app:growing-depth90}
All reports in this section are sharp. The comparator is the full
depth-$k$ physical encoder class of Section~\ref{sec:learning86},
not the one-odd circuit family used for noisy preparation.
All logarithms in rates and divergences are natural.
\begin{theorem}[Growing entanglement depth on a path]
\label{thm:growing-depth90}
For $d\geq3$, $2\leq k\leq d-1$, and $m\geq1$,
\begin{equation}
 \frac{D}{128k}\min\left\{1,\sqrt{\frac{d\log(k+1)}m}\right\}
 \leq \mathfrak M_m(P_d,k)
 \leq \frac{16D}{k}\min\left\{1,\sqrt{\frac{d\log(k+1)}m}\right\}.
 \label{eq:growing-depth90}
\end{equation}
The constants are uniform in $d,k,m$. The upper bound is attained by
a learner returning information-preserving codes. Its empirical
optimization costs $O(d)$ arithmetic operations after counting the
requests. The lower bound permits arbitrary mixed physical encodings.
For $k\geq d$, allocation excess is zero.
\end{theorem}

\subsection{Upper bound: the oracle sacrifices little mass}
A feasible indicator leaves no selected run longer than $k$.
Write $J$ for its deleted vertices and $\ell_J(i)=\mathbf1\{i\in J\}$.
For an inclusion-minimal deletion set with $t$ vertices, let
$l_0,\ldots,l_t$ be its intervening selected-run lengths, including
possibly empty endpoint runs. Minimality implies
$l_{j-1}+l_j\geq k$ at each deleted vertex. Hence
$tk\leq2\sum_jl_j=2(d-t)$, or $t\leq2d/(k+2)$.
If $N$ is the number of these deletion sets and
$q=\lfloor2d/(k+2)\rfloor\leq d/2$, then
\[
 \log N\leq \log\sum_{j=0}^q\binom dj
 \leq q\log(ed/q)
 \leq \frac{4d}{k}\log(k+1).
\]
Deleting one residue class modulo $k+1$ is feasible. Averaging its
deleted workload mass over the $k+1$ offsets shows that the optimal
deletion loss $r_* = \min_J\pi(J)$ is at most $1/(k+1)$.

Let $J_*$ minimize population deletion loss and let $\widehat J$
minimize empirical deletion loss over the minimal deletion sets.
For a fixed candidate with excess $\delta=\pi(J)-r_*>0$, put
$Z=\ell_J(T)-\ell_{J_*}(T)$. Then $\mathbb EZ=\delta$,
$|Z|\leq1$, and $\operatorname{Var}Z\leq\delta+2r_*$.
Bernstein's inequality therefore gives the conservative bound
\[
 \Pr(\widehat\pi(J)\leq\widehat\pi(J_*))
 \leq\exp\left(-\frac{m\delta^2}{4r_*+4\delta}\right).
\]
For $x=\log N+t$, a union bound implies
\[
 \Pr\left(\pi(\widehat J)-r_*>
          \frac{4x}{m}+2\sqrt{\frac{r_*x}{m}}\right)\leq e^{-t}.
\]
Integration of this tail, followed by Jensen's inequality, yields
\[
 \mathbb E[\pi(\widehat J)-r_*]
 \leq\frac{4(\log N+1)}m+
          2\sqrt{\frac{r_*(\log N+1)}m}.
\]
Since $\log N+1\leq(5d/k)\log(k+1)$, writing
$x=d\log(k+1)/m\leq1$ bounds this expectation by
$(20x+2\sqrt5\sqrt x)/k\leq25\sqrt x/k$.
Binary prediction excess is $D/2$ times deletion excess.
When $x>1$, instead sample a uniform periodic deletion offset;
its expected excess is at most $D/[2(k+1)]$. These two learners
give the stated upper bound, with a choice determined by known
$d,k,m$.

For empirical optimization, add virtual deleted positions $0,d+1$
of weight zero. A valid next deletion lies at distance at most $k+1$.
The shortest-path recurrence is
\[
 f(j)=w_j+\min_{\max(0,j-k-1)\leq i<j} f(i).
\]
A monotone deque maintains the sliding minimum in amortized constant
time per vertex. Optimize lexicographically by deletion weight and
then number of real deleted vertices. Any removable zero-weight
deletion would improve this second objective, so the resulting set
is inclusion-minimal. Backtracking takes $O(d)$ time. Each selected
indicator has the lossless projector code of
Theorem~\ref{thm:region86}; no odd-component restriction is needed
in this noiseless construction.

\subsection{Lower bound: identify the least requested query}
Set $r=\lfloor(d+1)/(k+2)\rfloor\geq d/(4k)$.
Place $r$ disjoint $(k+1)$-vertex path segments, separated by single
zero-weight vertices. All unused vertices also have weight zero.
In segment $j$, one unknown vertex $U_j\in\{1,\ldots,k+1\}$ has
weight $(1-k\theta)/[r(k+1)]$; each other vertex has weight
$(1+\theta)/[r(k+1)]$. Choose
\[
 \theta=\min\left\{\frac1{2k},
              \sqrt{\frac{r\log(k+1)}{8mk}}\right\}.
\]
These workloads are nonnegative and sum to one. Each physical
contrast projects into the product of the segment polytopes
$\{c\in[0,1]^{k+1}:\sum_i c_i\leq k\}$: a full connected
segment cannot fit in a depth-$k$ block, and every proper subset
is feasible. Decompose each projected contrast into a randomized
indicator and complete it to exactly $k$ selected vertices per
segment, setting every connector to zero. This transformation
depends on the learner's output but not on the unknown workload.
It gives a randomized feasible physical indicator whose conditional
expected score is no smaller for every workload in the construction. The unique omitted vertex
in each segment is therefore an estimate $\widehat U_j$.
Every error $\widehat U_j\ne U_j$ costs excess $D\theta/(2r)$.

Give the $U_j$ independent uniform priors. To bound the information
about one coordinate, reveal all the other coordinates. Compare
each resulting request distribution with an auxiliary distribution
that makes segment $j$ uniform and leaves the other segments
unchanged. Its one-sample chi-squared divergence is
$k\theta^2/r$. KL is no greater, so the information in $m$
requests about $U_j$, even with the revealed coordinates, is at
most $mk\theta^2/r\leq\log(k+1)/8$.
Fano's inequality \citep{yu1997} gives
\[
 \Pr(\widehat U_j\ne U_j)
 \geq 1-\frac18-\frac{\log2}{\log(k+1)}\geq\frac18.
\]
Thus Bayes excess, and hence minimax excess, is at least $D\theta/16$.
The bound $r\geq d/(4k)$ gives
\[
 \theta\geq\frac1{\sqrt{32}\,k}
          \min\{1,\sqrt{d\log(k+1)/m}\},
\]
which implies the stated constant $1/128$.

The theorem distinguishes absolute error from resolving the remaining
allocation choices. Depth $k$ reduces worst-case excess to order
$1/k$, even without requests. Learning a sufficiently small fixed fraction of that scale
requires order $d\log(k+1)$ requests. The fixed-depth path law is
recovered when $k$ is held constant.

\paragraph{Relation to sparse-loss prediction.}
When $k=d-1$, every maximal allocation omits one vertex. A query then
assigns loss one to its own omission action and zero to every other
action: the loss vector is one-hot. This is the sparse-loss setting
studied by \citet{kwon2016sparse}. Their full-information regret bound
(Theorem~4) gives $O(\sqrt{m\log(d)/d}+\log d)$ for this case;
online-to-batch conversion recovers the large-sample upper rate
$O(\sqrt{\log(d)/(dm)})$ when $m\geq d\log d$.
The small-loss improvement is therefore a classical ingredient.
Our theorem handles coupled deletion choices at every depth, gives
uniform finite-sample bounds, constructs physical attaining encodings,
and includes arbitrary mixed encoders in its lower bound.

\section{Completing an empirical allocation}
\label{app:completion90}
Let $\mathcal A$ be a finite downward-closed family of query sets,
with nonnegative empirical weights $\widehat w$ and population
weights $w$. If $\widehat S$ maximizes empirical weight over
$\mathcal A$, every feasible inclusion-maximal extension
$M\supseteq\widehat S$ satisfies
\[
 \widehat w(M)=\widehat w(\widehat S),\qquad
 w(M)\geq w(\widehat S).
\]
The first claim follows because a strict increase would contradict
empirical optimality; the second follows from nonnegativity.
Every added coordinate has zero empirical weight. Consequently,
full-feasible and maximal-set ERM are not intrinsically different
statistical problems. The smaller maximal class gives a tighter
analysis and a useful tie rule. A poor tie rule can nonetheless
create an arbitrarily large sample-requirement penalty.

For an exact example, take $K_{a,a}$ at depth one, with total
probabilities $3/4$ and $1/4$ on the two colors, uniform within
each color. Let $K\sim\mathrm{Binomial}(m,3/4)$ and break color
ties toward the first color. An ERM retaining only observed vertices
of the winning color has expected captured mass
\[
 \mathbb E\left[
 \mathbf1\{2K\geq m\}\frac34\{1-(1-1/a)^K\}
 +\mathbf1\{2K<m\}\frac14\{1-(1-1/a)^{m-K}\}\right].
\]
Its excess is $(D/2)(3/4-\text{expected captured mass})$. Completing
the winning color instead has expected excess
$(D/4)\Pr(2K<m)\leq(D/4)e^{-m/8}$, independent of $a$.
The sparse rule's excess is at least
$(3D/8)(1-m/a)$. The experiment compares these exact expectations
with independently sampled counts; it does not claim a new minimax
separation between the two feasible classes.

\section{Two statistical resources under preparation noise}
\label{app:two-resources90}
Let $\mathcal A$ be the finite one-odd path circuit family, of size
$N$, and let $b_i(S)=\mathbf1\{i\in S\}v_i\prod_{e\ni i,\ e\text{ retained}}\gamma_e$.
Calibration supplies $B$ independent Bernoulli observations for each
edge error probability and each native report-success probability.
Clip the resulting estimates of $\gamma_e$ and $v_i$ to $[0,1]$.
The learner maximizes $\widehat\pi\cdot\widehat b(S)$ using
$m$ independent request identities, independently of calibration.
Suppose the true effective contrasts $b_i^{\rm true}(S)\in[0,1]$
also satisfy
\[
 \sup_{S\in\mathcal A}\sum_i\pi_i
       |b_i^{\rm true}(S)-b_i(S)|\leq\beta.
\]
The mismatch condition concerns prediction contrasts. It does not
assert information preservation under an unspecified true channel.
We define its Bayes risk as
$R^{\rm true}_\pi(S)=1/2-(D/2)\pi\cdot b^{\rm true}(S)$;
the proposition does not bound an arbitrary estimated decoder.

\begin{proposition}[Requests, calibration, and model error]
\label{prop:two-resources90}
For the raw empirical maximizer, expected excess over the best
member of the same circuit family under the true contrasts obeys
\begin{equation}
 \mathbb E[R^{\rm true}_\pi(\widehat S)-\min_{S\in\mathcal A}R^{\rm true}_\pi(S)]
 \leq D\min\left\{\frac12,
       \sqrt{\frac{\log(2N)}{2m}}+\frac3{\sqrt B}+\beta\right\}.
 \label{eq:two-resources90}
\end{equation}
The calibration cost is $B(2d-1)$ scalar observations, in addition
to the $m$ requests.
\end{proposition}
\begin{proof}
Conditional on calibration, each function $i\mapsto\widehat b_i(S)$
takes values in $[0,1]$. The exponential-moment form of Hoeffding's
inequality gives expected uniform request deviation at most
$\sqrt{\log(2N)/(2m)}$. Telescoping products in $[0,1]$ gives
\[
 |b_i(S)-\widehat b_i(S)|\leq
 |v_i-\widehat v_i|+\sum_{e\ni i}|\gamma_e-\widehat\gamma_e|.
\]
Clipping cannot increase estimation error. Each estimated factor
has expected absolute error at most $1/\sqrt B$, by its binomial
variance and Cauchy--Schwarz. Each query has at most two incident
edges, so the expected workload-weighted deviation is at most
$3/\sqrt B$. Adding model mismatch gives the displayed sum as
an expected uniform score error. Empirical maximization costs at
most twice that error; conversion to risk multiplies by $D/2$.
All true contrasts lie in $[0,1]$, so excess is also at most $D/2$.
\end{proof}
The guarantee has no sharp claim about constants or graph-dependent
calibration complexity. It is for raw counts; the smoothed methods
are separately evaluated estimators. For half-count smoothing of
request counts, with calibration estimates unchanged,
the same proof adds a deterministic score perturbation at most
$d/(2m+d)$ inside the parentheses. This bound is deliberately
conservative and need not certify a small observed improvement.

\subsection{A sharp joint law on depth-three paths}
\label{app:joint-law90}
The two terms in Proposition~\ref{prop:two-resources90} cannot in
general replace one another. Let $\mathcal A_3$ be the same one-odd
path circuit family with depth cap three. Write
\[
 \mathfrak J_{m,B}(d)=\inf_{\widehat S}\sup_{\pi,\gamma,v}
 \mathbb E\!\left[R_{\pi,\gamma,v}(\widehat S)
       -\min_{S\in\mathcal A_3}R_{\pi,\gamma,v}(S)\right],
\]
where the learner observes $m$ request identities and $B$ independent
Bernoulli observations for each edge and native visibility, as above.
The supremum is over all workloads and factors in $[0,1]$; $D$ is
fixed and $\beta=0$. The learner may randomize its allocation.
The comparator and learner both belong to $\mathcal A_3$.

\begin{theorem}[Joint request and calibration complexity]
\label{thm:joint-law90}
For $d\geq4$ and $m,B\geq1$,
\begin{equation}
 \frac D{256}\min\!\left\{1,\sqrt{d/m}+B^{-1/2}\right\}
 \leq \mathfrak J_{m,B}(d)
 \leq 3D\min\!\left\{1,\sqrt{d/m}+B^{-1/2}\right\}.
 \label{eq:joint-law90}
\end{equation}
\end{theorem}
\begin{proof}
The upper bound follows from Proposition~\ref{prop:two-resources90}
with $N\leq2^d$ and $\beta=0$.

For the request lower bound, set $\gamma=v=1$ and place
$r=\lfloor(d+1)/5\rfloor\geq d/8$ four-vertex subpaths along the
path, separated by unused vertices. Each gadget has one of two
possible low-probability vertices. Its low vertex has probability
$(1-3\theta)/(4r)$; its other three vertices have probability
$(1+\theta)/(4r)$, where
$\theta=\tfrac18\min\{1,\sqrt{r/m}\}$. All remaining queries have
probability zero. A feasible allocation retains at most three
vertices of each gadget. Every choice of three is attainable by
omitting the separators: three selected vertices on a four-vertex
path always include an odd run. Thus the oracle that omits every
low vertex belongs to $\mathcal A_3$.

Replace a learner's allocation by completing its projected set to
three vertices per gadget and setting every other coordinate to zero.
The replacement is feasible and cannot decrease its score for any
workload in the construction.
Map its omitted vertex to the corresponding low-vertex guess;
if it omits neither candidate, choose either guess. A wrong guess
costs at least $D\theta/(2r)$ in excess risk. Flipping one gadget's
low vertex changes the one-request distribution by KL divergence
\[
 \frac\theta r\log\frac{1+\theta}{1-3\theta}
 \leq\frac{8\theta^2}{r}.
\]
The $m$-sample KL is at most $1/8$, so Pinsker's inequality bounds
neighboring total variation by $1/4$. Averaging over all sign
vectors, each guess errs with probability at least $3/8$.
The total excess is at least
$3D\min\{1,\sqrt{r/m}\}/128$, hence at least
$D\min\{1,\sqrt{d/m}\}/128$.
Calibration data are deterministic here and reveal nothing about
the unknown workload.

For the calibration lower bound, put probability $1/3$ on each
of the first three queries and zero elsewhere. Set $v=1$ and all
edge factors except the first two to one. Those two factors are
both $\gamma_\pm=\sqrt3-1\pm\delta$, where
$\delta=1/(16\sqrt B)$. On the active subpath, every pattern except
$111$ and $101$ is dominated by $101$ for both parameter values.
Set all outside coordinates to zero and replace every active prefix
other than $111$ by $101$. Both resulting allocations are feasible,
have an odd run, and weakly improve the original score in both worlds.
The outside edge factors equal one, so removing outside vertices
does not change active contrasts. Their score difference is
\[
 f(\gamma)=\frac{2\gamma+\gamma^2-2}{3},\qquad
 |f(\gamma_\pm)|\geq\delta.
\]
Thus the preferred allocation changes with the sign, and a wrong
choice costs at least $D\delta/2$.

Only two edge-calibration panels carry information about the sign.
Their Bernoulli parameters are $(1-\gamma_\pm)/2$ and differ by
$\delta$. Both lie between $0.10$ and $0.17$, so each observation
has KL at most $12\delta^2$. The full calibration KL is at most
$24B\delta^2=3/32$ and total variation is less than $1/4$.
Request observations have the same distribution in both cases.
The two-point testing bound therefore gives expected excess at
least $3D\delta/16=3D/(256\sqrt B)$. Randomized allocations obey
the same bound: dominated patterns can be replaced by $101$, leaving
a randomized choice between the two candidates.

Taking the larger of the request and calibration lower bounds
proves the stated sum bound, using
$\max\{\min(1,a),b\}\geq\tfrac12\min(1,a+b)$ for $a\geq0$
and $0\leq b\leq1$.
\end{proof}
The rate is sharp for this path family, not the displayed constants.
Since total calibration cost is $Q=B(2d-1)$, its second term is
$\sqrt{(2d-1)/Q}$. The theorem fixes the per-parameter sampling
scheme; it does not rule out gains from adaptive calibration or
additional shared structure in a particular device's parameters.

\section{Online allocation on a path}
\label{app:online90}
The learner may also choose an encoding for each new incoming record.
Assume a fixed known path channel and the same finite circuit family
$\mathcal A$. At round $t$, a fresh pair $(C_t,G_t)$ follows the
original signal model. Draw an allocation $S_t$, prepare its code
from $G_t$, and discard that classical record. The encoder and decoder
know the sampled allocation. Then observe query identity $i_t$ and
incur expected binary risk
$1/2-(D/2)b_{i_t}(S_t)$. Observing $i_t$ reveals the reward
$b_{i_t}(S)$ of every candidate under the known model. This is
full-information feedback. No prediction labels are needed.

Choose
\[
 \Pr(S_t=S)\propto
 \exp\!\left(\eta\sum_{s<t}b_{i_s}(S)\right).
\]
The uniform initial distribution is over nonempty feasible allocations
with at least one odd component. Conditional on the known allocation,
every realized encoding preserves $G_t$ under the stipulated channel.
This is a stream of newly prepared records, not a procedure that
re-encodes a previously stored bit after discarding its classical copy.
Only query history carries between rounds. An unlabelled mixture of
the sampled codes need not preserve the bit.

\begin{proposition}[Efficient full-information allocation]
\label{prop:online90}
For any fixed query sequence of length $T$, exponential weighting
has expected cumulative risk excess over the best fixed allocation
at most
\begin{equation}
 \frac D2\left(\frac{\log N}{\eta}+\frac{\eta T}{8}\right).
 \label{eq:online90}
\end{equation}
For $N>1$, setting $\eta=\sqrt{8\log N/T}$ gives
$(D/2)\sqrt{T\log N/2}\leq(D/2)\sqrt{Td\log2/2}$.
The distribution, its query-wise expected contrasts, and an exact
sample can be computed in $O(dk)$ scalar operations per round
without enumerating the encodings. If $N=1$, regret is zero.
\end{proposition}
\begin{proof}
Rewards lie in $[0,1]$. Apply Hoeffding's lemma to the log partition
function increment and telescope:
\[
 \eta\max_S\sum_{t=1}^T b_{i_t}(S)-\log N
 \leq \log(Z_{T+1}/Z_1)
 \leq \eta\sum_{t=1}^T\mathbb E[b_{i_t}(S_t)]+\eta^2T/8.
\]
Rearrange and multiply by $D/2$.
For computation, scan the path from left to right. A state records
the current selected-run length $0,\ldots,k$ and whether an odd
run has closed. Appending a bit finalizes the previous site's
contribution, since both its neighbors are now known. A zero
sentinel closes the last run; accept only states with a closed odd
run. Replace the maximization recurrence by log-sum-exp to compute
$Z_t$, with $\eta$ times the cumulative query counts as weights. Forward and
backward messages give transition probabilities, exact expected
contrasts, and backward or forward conditional sampling. There
are $2(k+1)$ states and at most two transitions per state per site.
\end{proof}
This is an application of exponential weighting \citep{freund1997hedge} with a tractable
physical action distribution, not a new general online-learning
principle. It permits changing workloads but compares with one fixed
allocation; it does not provide dynamic regret against a changing
oracle. The rate is dimension dependent and does not claim the
regional $rk$ refinement. Detector costs and repeated state
preparation remain separate physical resources.

\FloatBarrier
\section{Experimental protocols and complete comparisons}
\label{app:experiments86}
\label{app:experiments90}
The experiments use two distinct kinds of evidence. The earlier
geometry and connected-region studies are retained without changing
their observations. New request/calibration, online and noise studies
were specified before their new draws. Port-repair comparisons use
the earlier connected-region observations and are retrospective.
The reproducibility archive preserves every protocol, source hash,
raw array and analysis, including adverse comparisons. No quantum
hardware measurements enter these synthetic studies;
Appendix~\ref{app:service94} reports the subsequent physical acquisition.

\subsection{Growing-depth scaling diagnostic}
\label{app:depth-study90}
A separate protocol varies $d\in\{128,256,512,1024\}$ and
$k\in\{2,4,8,16,32\}$, with sharp reports and $h=0$.
Rounded request budgets are
$m=d\log(k+1)\{1/8,1/2,2,8,32,128\}$.
Each dimension/depth/family condition has 64 independent seeds,
for 15,360 empirical fits across two workload families.
The local-alternative family uses the exact separated-segment
construction and budget-dependent $\theta$ in
Appendix~\ref{app:growing-depth90}; its low-query identities remain
fixed across budgets but its probabilities change. The fixed-workload
control draws $\pi\sim\operatorname{Dirichlet}(0.5,\ldots,0.5)$
and uses nested request prefixes. The empirical learner minimizes
omitted count, then number of omissions, with the linear-time
recurrence. The oracle knows the workload; a uniform periodic
allocation supplies a data-free comparator. All are evaluated by
exact population risk in the full ideal depth-$k$ family.

Figure~\ref{fig:growing-depth90} presents both families in the main text.

For normalized budgets at least two,
$k\sqrt{m/[d\log(k+1)]}$ times empirical excess ranges from
$0.0656$ to $0.1578$ across all hard-family cells.
At the largest budget it ranges from $0.0028$ to $0.0071$ on
fixed Dirichlet workloads. The latter are not minimax-hard populations.
The periodic baseline shares the hard-family slope because $\theta$
decreases with $m$. Only the fixed-workload control tests learning
on unchanged populations.
All dimension-specific curves, bootstrap intervals, masks, counts
and periodic baselines are retained. Before acquisition, 744
exhaustive or independent-recurrence objectives passed; every fitted
mask and population risk was subsequently replayed. Acquisition took
8.63 seconds with one numerical thread, excluding verification and
plotting. These are finite diagnostics of a proved rate, not estimates
of a universal learning-curve exponent.

\subsection{Fixed-depth geometry comparisons}
The retained fixed-workload experiments compare paths with bicliques
and vary the number and size of connected regions
(Figure~\ref{fig:geometry-learning86}). Section~6.1 reports their
population-risk interpretation and stronger repair controls.

\begin{figure}[!htb]\centering
\includegraphics[width=\linewidth]{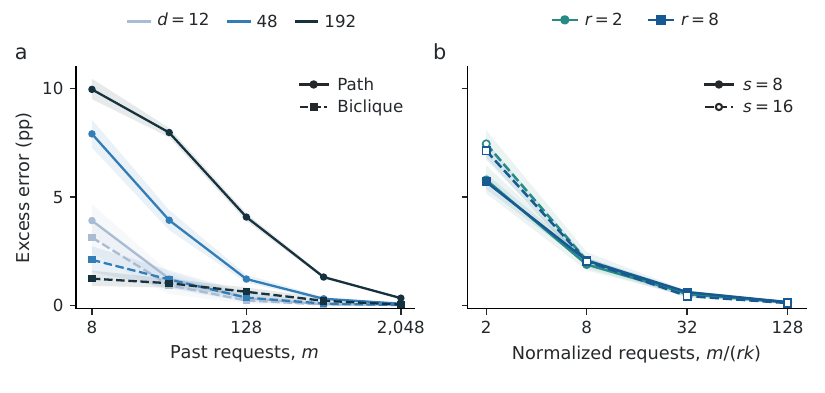}
\caption{Enlarging a region differs from adding allocation decisions.
(a) Fixed full-support Dirichlet workloads on paths and bicliques,
at depth two; 20 paired workload instances with ten training repetitions.
(b) Chains of $r$ biclique regions, each with $s$ vertices per color,
at depth four and total connection-query mass $0.2$; 24 paired fixed
instances with ten repetitions.
Population excess is measured against the physical-class optimum;
bands are pointwise 95\% instance-bootstrap intervals. Reports are
sharp and $h=0$. These fixed-workload curves do not estimate minimax
rates; biclique query support grows with $s$.}
\label{fig:geometry-learning86}
\end{figure}

\subsection{Two-resource factorial design}
For each $d\in\{63,127\}$ and each of three noise regimes, draw
32 independent workload/channel instances. Requests have
$\pi\sim\operatorname{Dirichlet}(0.5,\ldots,0.5)$.
To generate native query visibility, draw local outcome-error
probabilities uniformly from $[0.002,0.02]$ and multiply their
$1-2e$ factors over the path query's two or three sites.
Calibration observes each resulting native visibility directly; it
does not estimate those underlying local errors.
Low preparation noise has independent edge-error probabilities
uniform in $[0.001,0.01]$. Heterogeneous noise replaces an independently
marked quarter of edges by errors uniform in $[0.08,0.18]$.
Stress noise uses $[0.2,0.4]$ on the marked edges. These are controlled
synthetic distributions, not fitted device noise.

Each instance has four request streams and four independent
calibration panels. Request budgets are the rounded positive integers
$m=d\{1/4,1,4,16,64\}$; calibration budgets are
$B\in\{4,16,64,256,1024\}$ per parameter. Independent multinomial
increments give nested request prefixes. Independent binomial
increments give nested edge-error and native-report-success counts.
The two sources are independent and crossed, producing
$192\cdot5\cdot5\cdot4\cdot4=76{,}800$ fitted rows per method.
The two prefix axes pair the learning curves; they do not increase
the number of independent workload/channel units.

Edge counts estimate $\gamma_e=1-2\epsilon_e$ and native successes
estimate $(1+v_i)/2$. Clip both estimated factors to $[0,1]$.
Total calibration is $Q=B(2d-1)$, separately from $m$ requests.
At the two dimensions, $B=256$ costs 32,000 and 64,768 scalar
observations. The study does not supply a physical exchange rate
between calibration, request acquisition, detector operations and
preparation gates. Its mean gate counts concern one preparation.

\subsection{Matched estimators and privileged diagnostics}
Joint uses empirical requests and both estimated channel factors.
Task only sets both factors to one. Each is evaluated using raw
frequencies and half-count smoothing,
$\widetilde\pi_i=(n_i+1/2)/(m+d/2)$, with calibration unchanged.
Reporting only uses smoothed requests, estimated native visibility,
and ideal preparation factors. Joint, Task only and Reporting only use
the same dynamic program and first-predecessor tie rule. The selected
runs may have length one, two or three; at least one run must be odd.

Greedy starts with the greatest standalone-benefit singleton and
visits the remaining sites in descending estimated standalone benefit.
It accepts an addition only if it strictly raises estimated score and
keeps the same one-odd feasible family. This is a one-pass heuristic,
not a local optimum: a beneficial multi-site addition can be missed
when its intermediate single-site additions violate the odd condition.
Product solves the cap-one problem with smoothed requests and native
visibility. Full chain retains every site, uses depth $d$, and needs
neither request fitting nor calibration. Its excess relative to a
cap-three oracle may be negative because its resource budget differs.

Known channel uses true factors with smoothed empirical requests.
Known workload uses true requests with estimated factors.
The circuit oracle knows both. Their risk differences diagnose the
two error sources, but need not add exactly: changing either estimate
can change the selected mask. None is tuned with population outcomes.
The primary comparison, fixed before the draws, is Task only smoothed
minus Joint smoothed at $d=63$, heterogeneous noise, $m=1{,}008$,
and $B=256$.

\subsection{Uncertainty, adverse cells, and depth sensitivity}
Risk is evaluated exactly from the true synthetic factors after each
mask is chosen. The 16 request/calibration repetitions are averaged
within each latent instance. Five thousand paired bootstrap resamples
of the 32 instances give pointwise 95\% percentile intervals.
Secondary intervals are descriptive and have no multiplicity
adjustment. All 150 dimension/noise/budget cells are released,
including both raw and smoothed estimators.
The generation of binary future outcomes from this same contrast
model would add sampling variance, not an independent validation of
the model; population risk is therefore the primary outcome here.

\begin{table}[!htb]\centering
\caption{Corners of the complete request/calibration grid. Entries are
Joint smoothed population excess in percentage points relative to the
true-workload, true-channel oracle in the same circuit family. The last column
counts cells where Full chain has lower error across all 25 budgets.}
\label{tab:grid90}
\begin{tabular}{llrrrrr}\toprule
& & \multicolumn{2}{c}{$m/d=1/4$} & \multicolumn{2}{c}{$m/d=64$} & Full chain\\
$d$ & Noise & $B=4$ & $B=1024$ & $B=4$ & $B=1024$ & wins\\\midrule
63 & Low & 5.880 & 5.017 & 0.703 & 0.065 & 25/25\\
63 & Heterogeneous & 6.477 & 5.137 & 1.332 & 0.059 & 13/25\\
63 & Stress & 6.988 & 5.396 & 1.922 & 0.058 & 0/25\\
127 & Low & 6.008 & 5.772 & 0.549 & 0.048 & 25/25\\
127 & Heterogeneous & 6.509 & 5.087 & 1.462 & 0.052 & 13/25\\
127 & Stress & 6.998 & 5.495 & 1.939 & 0.054 & 0/25\\
\bottomrule\end{tabular}\end{table}

The main matched-smoothed gain is $1.214$ points, with paired interval
$[0.947,1.496]$. Joint's primary error is $8.576\%$ versus the
same-family oracle's $8.320\%$. The raw-count comparison is also
reported: Joint has $8.618\%$ error and Task only $9.616\%$.
Smoothing is useful for Joint in this cell but increases Task only's
error; a single favorable smoothing comparison would not identify
the channel contribution.

The prespecified depth sensitivity evaluates caps five and seven
on the same 32 heterogeneous-noise instances at $d=63$, retaining
the complete budget matrix and identical request/calibration arrays.
Each cap has its own true-workload, true-channel oracle. Larger
capacity expands the family but need not improve a fitted estimator
at every finite budget. Table~\ref{tab:depth90} reports the primary
budget without substituting the best depth into the declared primary
comparison.
\begin{table}[!htb]\centering
\caption{Depth sensitivity on 63-qubit heterogeneous-noise instances at
$m=1{,}008$ requests and $B=256$ calibration trials per parameter. Joint uses half-count smoothing; each oracle has the same cap.}
\label{tab:depth90}
\begin{tabular}{rrrrr}\toprule
Cap & Joint error (\%) & Oracle error (\%) & Excess (points) & Mean CZ gates\\\midrule
3 & 8.576 & 8.320 & 0.256 & 22.3\\
5 & 7.515 & 7.312 & 0.203 & 31.6\\
7 & 7.239 & 7.050 & 0.189 & 36.3\\
\bottomrule\end{tabular}\end{table}

\subsection{Completion and a calibration-hard diagnostic}
\begin{figure}[!htb]\centering
\includegraphics[width=\linewidth]{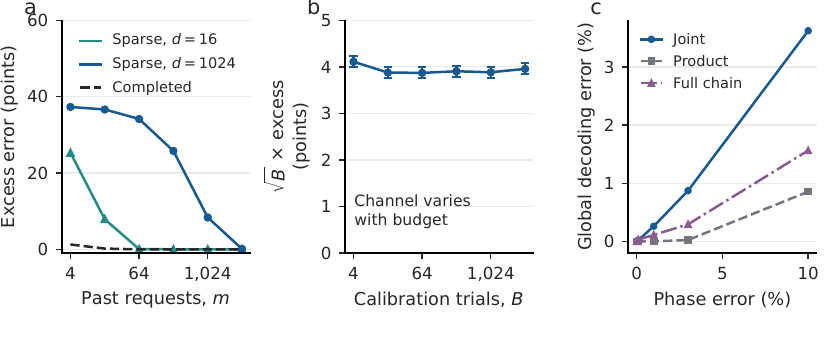}
\caption{Three diagnostics distinguish implementation, estimation and
information loss. (a) Sparse biclique ERM can omit unobserved queries;
completion removes the size penalty in this example. Lines are exact
expectations, markers average 2,048 count draws. (b) Calibration-hard
three-site channels give approximately constant $\sqrt B$-scaled
excess; bars are $1.96$ standard errors over 8,192 independent panels.
The channel separation changes with $B$, so these are not fixed-channel
learning curves. (c) Independent phase noise creates global decoding
error in exact nine-qubit calculations; these are deterministic model
values, with no sampling intervals.}
\label{fig:diagnostics90}
\end{figure}
The completion diagnostic uses depth one on $K_{a,a}$, with
$a\in\{8,32,128,512\}$, color masses $3/4,1/4$ and
$m\in\{4,16,64,256,1024,4096\}$. It compares 2,048 independent
count draws per condition with Appendix~\ref{app:completion90}'s
binomial occupancy formula. Completed full-feasible ERM and maximal
ERM choose the same full color on every draw.

For calibration alone, a uniform three-site path has product score
$2/3$ and triple score $(2\gamma+\gamma^2)/3$. They meet at
$\gamma_* = \sqrt3-1$. At each
$B\in\{4,16,64,256,1024,4096\}$, draw 8,192 independent signs
and set both true edge factors to $\gamma_*\pm0.2/\sqrt B$.
Calibrate the two edges independently; give the learner the true
uniform workload and ideal report visibility. The mean
$\sqrt B$-scaled excess ranges from $0.0387$ to $0.0411$ in risk
units. This budget-indexed construction tests the switching mechanism
in Theorem~\ref{thm:joint-law90}. The proof uses a smaller perturbation;
this curve measures finite-sample behavior rather than its lower-bound
constant.

\section{Stronger baselines, online learning and noise mismatch}
\label{app:controls90}
\subsection{Selective repair of connected allocations}
The original connected-region study has 96 full-support cells:
$r\in\{2,8\}$, $s\in\{8,16\}$ vertices per color,
$k\in\{2,4\}$, total connection-query mass in
$\{0.05,0.2,0.5\}$ and $m/(rk)\in\{2,8,32,128\}$.
Each cell contains 24 workload instances with ten request repetitions.
The main figure and original observed outcomes remain unchanged.

The new retrospective control first optimizes each region independently.
When a selected component exceeds the cap, it removes the lowest-count
endpoint of a selected interregion edge in an oversized component.
It repeats until feasible, then adds feasible queries in descending
training count. A second version makes two deterministic passes of
strictly improving one-for-one swaps. These rules use training counts
alone, with graph connectivity checked explicitly. All 96 original
cells are evaluated; no workload is selected using its improvement.

At the original primary condition ($r=s=8$, $k=4$, connection-query
mass $0.2$ and $m=1024$), exact allocation has $14.998\%$
population error. Repair gives $15.170\%$ and repair with local
improvement $15.143\%$. Their paired differences from exact allocation
are $0.172$ points ($[0.093,0.252]$) and $0.145$ points
($[0.075,0.220]$), compared with $5.242$ points for deleting every
connection query. Thus most of that earlier gap reflected the coarse
connection-removal baseline. Raw masks, preparation gates and sampled
future outcomes for the new controls are included in the archive.

\subsection{Online protocol and results}
Each online condition has a known sharp path channel, depth cap three,
512 requests and 12 seeds. Dimensions are 15 and 63. Stationary
sequences use a fixed Dirichlet-$0.5$ workload. Round-robin sequences
cycle deterministically through all sites. At each dimension, all
12 seeds repeat that one query sequence and vary only allocation
randomness; they are not 12 independently drawn workloads. Changing sequences switch
once, halfway through, between workloads that favor alternate pairs
of sites by a factor of four. Every full sequence is fixed before
learner randomization. Each round encodes a fresh incoming record.

The primary learner quantity is exact expected Gibbs reward, computed
from forward/backward messages. A sampled allocation is also saved
each round. Smoothed follow-the-leader chooses the empirical optimum
after adding half a count to every query. The comparator is the best
single allocation for the entire realized request sequence. Mean
cumulative excesses, in expected additional mistakes, are:
\begin{center}
\begin{tabular}{llrrr}\toprule
Qubits & Sequence & Gibbs & Smoothed FTL & Bound\\\midrule
15 & Stationary & 13.19 & 2.21 & 25.19\\
15 & Round robin & 16.50 & 0.00 & 25.19\\
15 & One change & 15.38 & 2.33 & 25.19\\
63 & Stationary & 26.72 & 8.71 & 51.49\\
63 & Round robin & 45.92 & 0.00 & 51.49\\
63 & One change & 36.42 & 11.71 & 51.49\\\bottomrule
\end{tabular}
\end{center}
Every sequence satisfies the bound; the simpler FTL baseline has
lower mean loss in all six conditions. These observations support
implementation correctness and show that a worst-case guarantee need
not give the best behavior on these sequences. They do not establish
a dynamic regret guarantee or empirical superiority of Gibbs.

\subsection{Retained information and correlated errors}
The odd-component invariant holds for arbitrary joint distributions
over retained-edge $ZZ$ errors: each error pattern commutes with the
same witness. Independence is required for the factorized prediction
formula, not for that information-preservation argument.
To separate these roles, an exact nine-site control applies a common
Bernoulli shock that flips every retained edge, in addition to
independent $ZZ$ errors. At shock probability $\rho$, a selected
query with exactly one retained neighbor acquires factor $1-2\rho$;
with two retained neighbors, the common flips cancel. Using the true
edge marginals in an independent model incorrectly attenuates the
interior query twice. At $\rho=0.1$, the resulting chosen allocation
has $9.47\%$ error, versus $8.23\%$ for the correct correlated-model
oracle. Global decoding remains exact for both.

Independent single-qubit $Z$ errors give a different failure. They
need not commute with the witness. The selected graph-state syndrome
distribution is an independent binary convolution; a Walsh transform
computes it exactly. Under $G=1$, its indices are shifted by the
all-ones syndrome, so unrestricted Bayes error follows from the total
variation between those distributions. At phase-error probability
$0.03$ in the fixed nine-site example, global errors are $0.874\%$
for Joint, $0.026\%$ for Product and $0.295\%$ for Full chain.
Restricted prediction errors are $15.14\%$, $19.52\%$ and $13.32\%$.
None of these representations remains lossless at this noise level;
their access comparison is no longer information matched.

\section{Reproduction and historical evidence}
\label{app:reproduction90}
The small installable allocation package contains the path optimizer,
Gibbs distribution and sampler, input validation, and a runnable
example. Its README provides the clean-environment commands. The full
reproduction entry point verifies frozen acquisition hashes, regenerates
request and calibration prefixes, checks population risks and compares
the main optimizer against an independent interval implementation.
It also checks manuscript references and the compiled artifact.

Before acquisition, exhaustive enumeration verified 864 max-plus
optimizer cases and 36 complete Gibbs distributions. The binary
information bound was checked on 100,000 asymmetric channels.
After acquisition, 11,136 independent interval optimizations checked
the primary fitted method, and every population-risk row was
recomputed. These finite checks corroborate the mathematical proofs;
they do not replace them. Factorial acquisition took 50.50 seconds on six workers; controls took
179.61 seconds. Numerical libraries used one thread per worker; these
times exclude development, audits, analysis and plotting.
All study seeds and exact package versions are recorded in the machine-readable protocols.

The ancillary data preserve the two-qubit benchmark, earlier allocation
studies and the 215,040-shot fixed IBM acquisition, together with their
protocols and chronology. These device observations
compare fixed implementations in a different measurement family;
they are not evidence for learned allocation, an all-input detector
guarantee, or exact information preservation in noisy device states.

\begin{table}[!htb]\centering
\caption{Matched estimators isolate the value of calibration. Population
error and mean preparation CZ counts in the prespecified 63-qubit cell
($m=1{,}008$, $B=256$). Intervals are pointwise 95\% bootstrap
intervals over 32 workload/channel instances. The oracle knows the
true workload and channel. Allocation rows have depth at most three
and an odd run; Product has depth one and Full chain depth 63.}
\label{tab:primary90}
\begin{tabular}{lrr}\toprule
Design & Error (\%) [95\% interval] & Mean CZ gates\\\midrule
Joint, smoothed & 8.58 [8.19, 8.99] & 22.3\\
Task only, smoothed & 9.79 [9.24, 10.37] & 25.8\\
\midrule
Joint, raw & 8.62 [8.24, 9.03] & 19.5\\
Task only, raw & 9.62 [9.09, 10.17] & 22.6\\
Reporting only, smoothed & 9.77 [9.23, 10.34] & 25.8\\
Greedy, smoothed & 8.94 [8.45, 9.50] & 22.3\\
\midrule
Product, smoothed & 15.64 [15.06, 16.22] & 0.0\\
Full chain & 9.54 [8.91, 10.16] & 62.0\\
\midrule
Workload and channel oracle & 8.32 [7.96, 8.71] & 21.5\\
\bottomrule\end{tabular}\end{table}

\FloatBarrier
\section{A parity-service health test and physical response calibration}
\label{app:service91}

Entanglement allocation can change how quickly a detector fault becomes
visible. We test this application in a simulated four-qubit parity
service and validate response calibration on archived IBM observations.
The physical experiment tests response prediction; it does not validate
the four-qubit service's coherence-preserving instrument.

\subsection{The service and its health test}
A service accepts a three-qubit probe and returns one requested parity.
Its menu is $ZZI$, $XXX$, and $IZZ$. The two $Z$ parities are the
stabilizers of the three-qubit repetition code; $XXX$ is a logical-state
diagnostic, not another syndrome check \citep{riste2015stabilizer}.
Preserving information inside a parity sector matters when the same
interface processes coherent inputs. Parity repeatability alone does
not impose this requirement: measuring local outcomes and reporting
their parity also repeats the ideal parity perfectly.

Our client supplies a batch of precommitted, randomly signed probes.
An evaluator retains each sign and checks the returned report.
The probe factory can use past request frequencies and calibration
observations, but its preparation cannot depend on the current request.
We enforce this contract with identical preparation prefixes across
query-specific circuits. This is a benchmark constraint, rather than
a claim that ordinary error-correction schedules conceal their checks.
The circuits do not implement an external request arriving after
physical preparation.

At entanglement depth two, the catalog contains the $Z$-product, left-pair,
and right-pair codes of Figure~\ref{fig:access-mechanism86}, together with
an $X$-product code. We use
$|000\rangle,|010\rangle$ for the product code and
$|\Phi^+\rangle,|\Psi^-\rangle$ on the selected pair, with $|+\rangle$
on the remaining qubit. Their ideal contrasts, in query order, are
$(1,0,1)$, $(1,1,0)$, and $(0,1,1)$; the $X$-product pair
$|+++\rangle,|-++\rangle$ has contrast $(0,1,0)$.
The product baseline learns which of its two preparations to use from
the same requests and healthy calibration observations as the pair
selector. All candidates use the same calibration shots per setting.
An ancilla reports the selected parity through compute--copy--uncompute:
the middle data qubit is connected to both outer data qubits and the
ancilla. The $Z$ checks use three CNOTs and the $X$ check uses five,
with Hadamard basis changes. These are detector gates; preparing a
Bell probe uses one CNOT. All-input Kraus-matrix checks give the
projectors $(I\pm P)/2$ to a maximum residual of $4.5\times10^{-16}$.

We inject a report-polarity fault by applying $X$ to the ancilla before
its measurement. Ideally, this reverses the report without changing
the nonselective data channel. The evaluator's task is to detect the
fault from signed reports and query identities. We also test a weaker
fault: a rotation $R_y(2\arcsin\sqrt{0.15})$ on the ancilla produces
a 15\% report-flip probability in the ideal instrument.
Thus an uninformative probe has a
concrete cost: it makes healthy and faulty operation harder to tell
apart. This deliberately injected fault tests the diagnostic method;
it is not an estimated field-failure rate.

Two controls expose the interface restrictions. A product preparation
that knows the requested query is perfectly informative in the ideal
case. A product probe followed by a second, nondestructive $ZZI$
service call also recovers the bit, because all three queries commute.
The benchmark therefore charges one call per probe. We include both
controls, together with a full-GHZ preparation, and report their
different resources in Table~\ref{tab:service91-controls}.

\subsection{Calibration from directly observed responses}
We calibrate what the service returns, rather than assuming separately
observed edge errors. For candidate $a$, query $t$, and balanced
label $g$, let $Y\in\{-1,+1\}$ be the report and define
\[
 b_{at}=\tfrac12\bigl(\mathbb E[Y\mid a,t,0]-
                          \mathbb E[Y\mid a,t,1]\bigr).
\]
This signed contrast includes preparation, gate, assignment, and
report errors. The two label-specific means also retain output bias.
Bias cancels from balanced-label classification risk; it would still
matter for unequal priors or logarithmic loss.

\begin{samepage}
\begin{proposition}[Direct response calibration]
\label{prop:service91-calibration}
Fix a catalog of $A$ encoders and $q$ requests. Suppose their response
channels are stationary, and acquire $B$ independent shots for each
$(a,t,g)$, independently of $m$ iid request samples. Let
$\widehat b_{at}=(\overline Y_{at0}-\overline Y_{at1})/2$,
choose $\widehat a\in\arg\max_a\sum_t\widehat\pi_t|\widehat b_{at}|$,
and decode with the learned signs of $\widehat b$, taking the sign to
be $+1$ at zero.
For the balanced source of Section~\ref{sec:access-model86}, with
probability at least $1-\delta$,
\begin{equation}
 R(\widehat a,\widehat s)-R^*_{\rm catalog}
 \leq D\min\left\{1,
 \sqrt{\frac{\log(4A/\delta)}{2m}}+
 \sqrt{\frac{\log(4Aq/\delta)}{B}}\right\}.
 \label{eq:service91-bound}
\end{equation}
The comparator uses the best candidate and its true Bayes binary
decoder. Total calibration cost is $Q=2AqB$, so the calibration term
equals $\sqrt{2Aq\log(4Aq/\delta)/Q}$.
\end{proposition}
\end{samepage}

\begin{proof}
For any orientation $s$, balanced labels give
$R(a,s)=1/2-(D/2)\sum_t\pi_ts_tb_{at}$.
Hoeffding's inequality and a union bound give
$\Pr(\max_{at}|\widehat b_{at}-b_{at}|>\eta)
\leq2Aq\exp(-B\eta^2)$.
Conditioning on calibration, each function
$t\mapsto|\widehat b_{at}|$ lies in $[0,1]$, so the uniform
request-estimation error exceeds $\epsilon$ with probability at most
$2A\exp(-2m\epsilon^2)$. On their intersection, compare the true
optimal score to its estimated score, the selected estimated score,
and the selected signed true score. The total difference is at most
$2\eta+2\epsilon$. The last step uses
$|\operatorname{sign}(\widehat b)b-|\widehat b||=|b-\widehat b|$,
so it includes incorrectly learned signs. Multiply by $D/2$ and use
the two failure budgets $\delta/2$. The trivial cap is $D$, since a
wrongly learned orientation can have risk above $1/2$.
\end{proof}

This finite-catalog result makes no factorization assumption and
does not replace the full-class lower bounds in the main text.
If every test contrast differs from its calibration contrast by at
most $\kappa$, the same argument adds $D\kappa$. A few observed
blocks cannot certify such a condition for future device operation.

For $A=4$, $q=3$, $m=128$, $B=2048$, $D=1$ and $\delta=0.05$,
Eq.~\eqref{eq:service91-bound} gives an excess-risk bound of $0.208$,
using 49,152 healthy calibration shots. It guarantees consistency;
it does not certify the much smaller empirical improvement.

Every probe uses the same query- and label-conditioned likelihood
alarm. Separate response-fit data estimate healthy and deliberately
faulty report probabilities. A third, healthy calibration block
sets the threshold. For an episode, we sum
$\log[\widehat p_{\rm fault}(Y\mid t,g)/
\widehat p_{\rm healthy}(Y\mid t,g)]$, using Jeffreys smoothing.
With $n_0$ independent healthy calibration episodes, the threshold is
their $\lceil0.95(n_0+1)\rceil$th smallest score; we alarm only on
strict exceedance. Exchangeability bounds the probability of a future
healthy score exceeding this rank, giving marginal false-alarm
probability at most $0.05$, with ties conservative. This is not a
guarantee conditional on one realized threshold. The alarm is trained
for a specified injected fault; Proposition~\ref{prop:service91-calibration}
concerns recovery error, not optimal multi-report fault power.

\subsection{Circuit experiments and escape controls}
Each of two declared gate-noise models uses 1,068 circuits and 1,093,632
simulated shots. Six blocks contain 1,024 shots per circuit setting.
Submitted blocks 0--1 fit responses, block 2 calibrates the alarm,
and blocks 3--5 evaluate both. Settings are shuffled within each block. One model
has equal qubit error scales; the other assigns scales $(1,1,5,1.5)$
to left, middle, right, and ancilla qubits. These are hypothetical
channels, not fits to a quantum processor. They combine gate
depolarization, amplitude damping, asymmetric assignment errors, and
phase damping on measured qubits. They omit spectator measurement
crosstalk. Exact parameters and raw counts accompany the code.

The workload menu includes syndrome-heavy, balanced, coherence-heavy,
left-heavy, and right-heavy request distributions. The displayed
coherence-heavy mixture is $(0.2,0.6,0.2)$. Over 4,096 independently
sampled request histories of length 128, calibration reduces mean
held-out recovery error from $16.88\%$ to $16.28\%$ in the asymmetric
model. A 1,000-replicate bootstrap resamples calibration and test
shots within each setting, refits the selector, and reuses the request
histories; its 95\% interval for the gain is $[0.323,0.928]$ percentage
points. Both allocation rules use the calibrated decoder orientations;
``request-only'' here ablates allocation calibration. In the symmetric
control, calibration changes error by
$+0.035$ points, with gain interval $[-0.121,0.020]$ points: there is
no resolved benefit. These histories share one acquisition per noise
model. The bootstrap describes shot uncertainty under that channel,
not variation across devices. With the single fixed
history used for the alarm, both the calibrated and request-only
rules select the left pair. That alarm comparison therefore tests
probe choice, not an additional calibration gain.

We replay 256 test episodes from distinct shots within each
probe/workload/fault condition, sampling queries and labels independently.
The polarity-reversal task uses eight reports per episode and 256
healthy threshold-calibration episodes. The weaker fault uses 32
reports and 64 calibration episodes. The extreme fault is nearly
saturated: both pair probes detect all 256 test faults, while the
$X$-product detects 255. It supplies little evidence of a large advantage.

The weaker fault separates the probes. With independently calibrated
thresholds, the left pair detects 184/256 faults with 15/256 healthy
alarms; the $X$-product detects 129/256 with 9/256 healthy alarms.
Because realized false-alarm rates differ, we also report a descriptive
held-out ROC comparison. At an empirical 5\% false-alarm rate, powers
are $70.70\%$ and $54.30\%$. A paired bootstrap over the 256 shared
episode indices gives a gain of $16.4$ points, with interval
$[3.9,28.9]$. The ROC threshold is an evaluation statistic, not the
independently calibrated operating threshold. Figure~\ref{fig:service91}c
plots the descriptive ROC; fixed operating counts are reported above.
Averaging descriptive ROC power over
the 4,096 request histories gives $69.14\%$ for the calibrated selector,
$64.36\%$ for the request-only selector, and $54.23\%$ for the learned
product selector, conditional on the same response and test acquisition.
All episodes are offline replays, not live service calls. Rules selecting
the same probe share the same replay, and pools are shared across comparisons.

\begin{table}[t]\centering
\caption{Escape controls for the asymmetric circuit model at request
probabilities $(0.2,0.6,0.2)$. Errors use separate test shots.
The calibrated and request-only rules choose the left pair for the
displayed fixed history. Depth is stored-state entanglement depth.}
\label{tab:service91-controls}
\begin{tabular}{lcccr}\toprule
Probe/control & Depth & Early query & Calls & Error (\%)\\\midrule
Z-product & 1 & No & 1 & 31.24\\
X-product (selected product) & 1 & No & 1 & 24.58\\
Left pair & 2 & No & 1 & 16.09\\
Right pair & 2 & No & 1 & 17.73\\
Full GHZ & 3 & No & 1 & 9.99\\
Query-aware product & 1 & Yes & 1 & 6.60\\
Z-product with second $ZZI$ call & 1 & No & 2 & 5.36\\\bottomrule
\end{tabular}
\end{table}

The back-action experiment uses the same Bell factories. After a
$Z$ check we measure the supported pair's $XX$ coherence; after
$XXX$ we measure its $ZZ$ coherence. Both parity sectors and two
opposite coherent phases within each sector are included, with
all outcomes retained. Ancilla extraction retains signed coherence
$0.893$, $0.764$, and $0.649$ for $ZZI$, $XXX$, and $IZZ$; local
refinement gives $-0.0004$, $-0.0016$, and $0.0072$. No-instrument
references and repeated parity checks accompany these witnesses.
The no-instrument reference is not duration matched. The simulation
demonstrates the distinction between the two implemented channels;
it is not a physical instrument certificate. Full instrument
characterization must also address preparation and analysis errors
\citep{pereira2023tomography}.

Finally, we test the main text's multiplicative contrast model on
these circuit responses. Its supported entries require
$b_{L0}b_{P2}b_{R1}=b_{R2}b_{P0}b_{L1}$, with unsupported entries zero.
A constrained likelihood fit includes independent output biases
and uses calibration shots only. On the asymmetric test set, the
direct-minus-factorized mean log score is $-1.2\times10^{-5}$ nats per
shot, with conditional normal interval
$[-4.4,2.0]\times10^{-5}$. This experiment does not resolve a
predictive difference between the models. We retain both results;
direct calibration is useful without claiming the factorization fails.

\begin{figure}[t]\centering
\includegraphics[width=\linewidth]{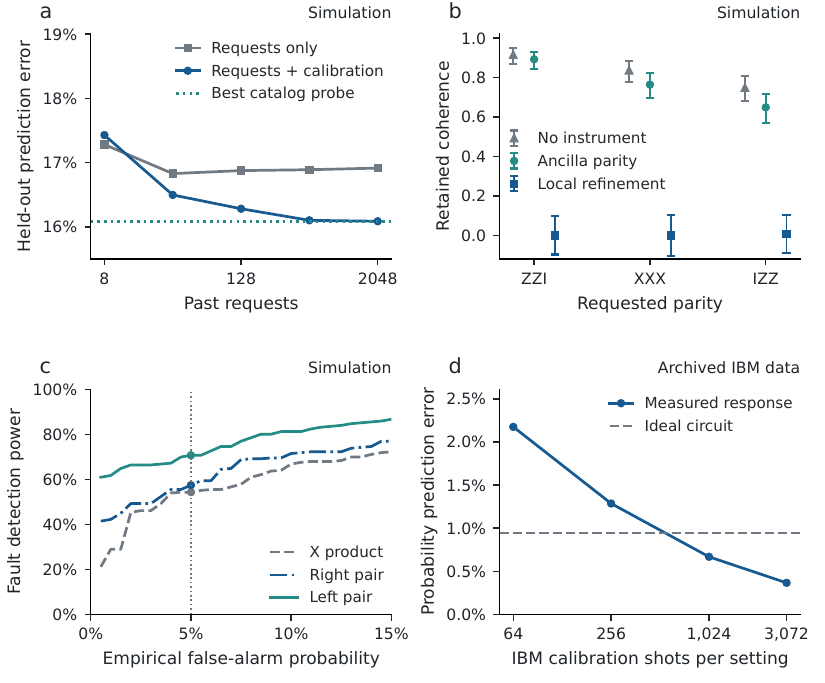}
\caption{(a--c) Declared-noise simulation. (a) Mean recovery error over
4,096 request histories sharing calibration and test shots.
(b) Unconditional signed Bell coherence; shot-weighted 95\% simultaneous
binomial bounds over prepared-state/block strata. (c) Query-conditioned
alarms for the 15\% report-flip fault: 256 disjoint 32-shot test episodes
per condition. Held-out ROC curves and empirical 5\% false-alarm points
are descriptive, not prospective thresholds. $X$-product is the selected
product comparator. (d) Archived physical IBM observations: probability
MAE across 24 settings on disjoint test blocks, with nested calibration
prefixes. This single retrospective two-qubit split tests response
prediction, not the instruments in (b).}
\label{fig:service91}
\end{figure}

\subsection{Physical calibration on archived IBM observations}
The earlier fixed-code experiment acquired 215,040 shots on
\texttt{ibm\_fez}, qubits 22 and 23, on 6 September 2026. Its raw
counts, receipt, and circuit hashes remain unchanged. We now use the
147,456 prediction shots for a retrospective response-calibration
test. There are 24 settings: two conditions, three preparations,
two queries, and two labels. Submitted blocks 0--2 supply 3,072
calibration shots per setting; blocks 3--5 supply 3,072 test shots.
The split was selected after the original aggregate results were
known. Per-circuit execution timestamps are unavailable, so these
are disjoint submitted blocks, not a demonstrated later-day test.

We estimate each report probability with Jeffreys smoothing,
$\widehat p=(K+1/2)/(B+1)$, and score predictions on the test shots.
Mean absolute probability error decreases from $0.00942$ for ideal
circuits to $0.00365$ for measured responses. At budgets 64 and 256,
estimation error is larger than the discrepancy of the ideal model;
the larger calibration budgets improve prediction
(Figure~\ref{fig:service91}d). The held-out Brier score decreases
from $0.070311$ to $0.070197$. Its paired gain is
$0.000114$, with conditional shot-level normal interval
$[0.000064,0.000164]$. All 20 possible three-block
calibration/three-block test splits yield positive Brier gain, ranging
from $7.43\times10^{-5}$ to $1.41\times10^{-4}$. These overlapping
splits are a sensitivity analysis, not independent replications.
These data support response calibration, without identifying isolated
noise parameters.

The original 1,068-circuit, 273,408-shot acquisition bundle was
superseded before submission. Appendix~\ref{app:service94} reports the
new four-qubit acquisition, including backend-bound compilation,
measurement-exposure controls, and calibration transfer to later jobs.
The original bundle and its ideal checks remain in the artifact.

\section{Learning probes for detector diagnosis}
\label{app:service93}

The same allocation problem governs two tasks: recovering a stored bit
and detecting a change in the detector. This section establishes the
connection over the full physical encoder class, then tests its limits
under more general detector noise. Numerical results in this section
are calculations for declared models. Appendix~\ref{app:service94}
reports the prospective quantum-processor evaluation, including
calibration transfer, coherence controls, and acquisition costs.

\subsection{An exact diagnostic experiment for every CSS allocation class}
Let $H\in\{0,1\}$ denote healthy or faulty operation. Before each call,
the client prepares a state in $\mathcal S_{\mathcal F}$, the partition
class of Section~\ref{sec:access-model86}. An independent request
$T\sim\pi$ then selects one of the independent commuting CSS queries.
The detector measures its ideal parity and flips the report with
probability $e_H$, where $0\le e_0<e_1\le1$. Flips are independent
across calls and of the request, preparation and ideal outcome,
conditional on $H$. The evaluator sees the request, report and any
classical preparation tag. Each state must belong to
$\mathcal S_{\mathcal F}$ conditional on that tag and the complete
past history. The client retains no quantum system between calls.
Classical randomization, biased states and adaptation to every past
report are allowed. Training requests are independent of $H$.

For partition $\Pi$, let $I$ be a locally compatible set of queries,
as in Theorem~\ref{thm:region86}, and define
\[
 p_{\mathcal F}(\pi)=\max_{\Pi\in\mathcal F,\ I\text{ compatible on }\Pi}
                      \sum_{t\in I}\pi_t.
\]
Write $\mathsf E_p$ for an experiment that returns an erasure with
probability $1-p$ and an observed flip bit distributed as
$\operatorname{Bernoulli}(e_H)$ with probability $p$. Its erasure
indicator is independent of $H$. Let $B_{p,N}(\alpha)$ be its optimal
$N$-call detection power at false-alarm probability at most $\alpha$.

For one-call learning, let $S_m\sim\pi^m$ be the training requests.
A learner $A$ maps $S_m$ to an admissible probe, its retained classical
tag, and an alarm rule. Write $\beta_\pi(A;S_m)$ for its power,
averaged over preparation randomness, the new request and the report.
Define
\[
 \mathfrak D_{\mathcal F}(m)=\inf_A\sup_{\pi\in\Delta_q}
 \mathbb E_{S_m}\!\left[
 B_{p_{\mathcal F}(\pi),1}(e_0)-\beta_\pi(A;S_m)\right].
\]
Here $A$ cannot use $\pi$, and its false-alarm probability must be at
most $e_0$ conditional on every training record, tag and request.
The comparator knows $\pi$. The identity below also holds if this
conditional requirement is weakened to false-alarm control averaged
over training and test randomness, uniformly over $\pi$.

\begin{theorem}[Allocation determines diagnostic power]
\label{thm:diagnostic93}
Under the preceding model, the optimal $N$-call experiment over all
admissible, classically adaptive probe policies is equivalent, in
Blackwell's comparison of experiments, to
$\mathsf E_{p_{\mathcal F}(\pi)}^{\otimes N}$. A fixed signed projector
probe supported on a maximizing $I$ attains it.
At one call and false-alarm level $\alpha=e_0$,
\begin{equation}
 B_{p,1}(e_0)=e_0+(e_1-e_0)p.
 \label{eq:diagnostic-linear93}
\end{equation}
Consequently, if $\mathfrak R_{\mathcal F}(m)$ is the sharp-report
minimax prediction regret of the main text, the minimax loss of
one-call diagnostic power after $m$ training requests is exactly
\begin{equation}
 \mathfrak D_{\mathcal F}(m)
       =\frac{2(e_1-e_0)}{D}\mathfrak R_{\mathcal F}(m),\qquad D>0.
 \label{eq:diagnostic-learning93}
\end{equation}
The attaining alarm has false-alarm probability $e_0$ conditional on
every training record, preparation tag and request; it does not
require knowledge of $\pi$.
\end{theorem}

\begin{proof}
The signed single-state expectation region is the convex hull of
$s\odot\mathbf1_I$, where $I$ ranges over compatible sets and
$s_t\in\{-1,+1\}$. For the outer bound, the main theorem's
single-state cut inequalities bound absolute correlations. For
attainment, the normalized signed projector
\[
 \tau_s^I=2^{-d}\prod_{t\in I}(I+s_tP_t)
\]
is separable across the chosen partition. Independence of the global
Paulis makes its expectation $s_t$ on $I$ and zero off $I$.
Thus every admissible state's entire report channel is a mixture
of these vertex channels under both hypotheses. This is a channel
decomposition; it need not decompose the original density matrix.
Revealing the mixture component can only help the evaluator.

A vertex gives a known ideal report on $I$ and a fair,
hypothesis-independent report otherwise. Its experiment is
$\mathsf E_{\pi(I)}$: the remaining request identity and known sign
are ancillary conditional on the informative flag and flip bit.
If $p'\le p$, erase an informative output of $\mathsf E_p$ with
probability $1-p'/p$ to obtain $\mathsf E_{p'}$. Draw a compatible
query identity from its conditional distribution and restore its
known sign to reproduce the full vertex channel. The resulting
kernel never uses $H$. Applied conditionally on the simulated history
at each call, it reproduces every adaptive policy from independent
copies of $\mathsf E_{p_{\mathcal F}}$. The maximizing fixed vertex
attains this experiment, proving optimality for every terminal test.

For Eq.~\eqref{eq:diagnostic-linear93}, alarm on an informative
flipped report, never alarm on an informative unflipped report, and
alarm with independent probability $e_0$ on an erasure. Its false-alarm
probability is $e_0$ even conditional on the request. Its power is
$e_0+(e_1-e_0)p$. Conversely the total variation between the two
one-call laws is $(e_1-e_0)p$, so no level-$e_0$ test has greater power.

More generally, a conditional state with expectations $r_t$ has
one-call total variation $(e_1-e_0)\sum_t\pi_t|r_t|$. Averaging over
training and retained tags preserves this upper bound. The mean
absolute-correlation vector remains in the main theorem's convex
feasible region. Every vector in that region is realized by a tagged
mixture of signed projectors, with the preceding alarm attaining the
bound. Prediction regret is $D/2$ times the loss of weighted coverage;
diagnostic-power regret is $e_1-e_0$ times that same loss. Taking the
same infimum over learners and supremum over workloads proves
Eq.~\eqref{eq:diagnostic-learning93}.
\end{proof}

In particular, the path law becomes
$\Theta((e_1-e_0)k^{-1}\min\{1,\sqrt{d\log(k+1)/m}\})$, uniformly
over $2\le k<d$. The connected-region law transfers with the same
geometric conditions. This is a learning result for probe allocation:
the hypotheses $e_0,e_1$ are known. Estimating a healthy baseline or
an operating threshold costs separate calibration data.

For arbitrary $N$, a learned compatible set of mass $\widehat p$
also satisfies the same-level transfer bound
\[
 0\le B_{p^*,N}(\alpha)-B_{\widehat p,N}(\alpha)
 \le 1-\{1-(e_1-e_0)(p^*-\widehat p)\}^{N}.
\]
To see this, promote erased outputs of $\mathsf E_{\widehat p}$ to
informative outputs with probability
$(p^*-\widehat p)/(1-\widehat p)$ and generate their flip bits using
$e_0$. Under $H=0$ this exactly reproduces $\mathsf E_{p^*}$; under
$H=1$ the one-use total-variation error is
$(e_1-e_0)(p^*-\widehat p)$. Product coupling gives the bound.
It compares optimal population tests, not estimated thresholds.

\subsection{The complete separable benchmark and its resource cost}
For $ZZI,XXX,IZZ$, the full separable signed region is
$|r_0|+|r_1|\le1$, $|r_1|+|r_2|\le1$.
Its vertices are $(s,0,t)$ and $(0,s,0)$, with $s,t\in\{-1,+1\}$.
The full depth-at-most-two region adds $(s,t,0)$ and $(0,s,t)$,
giving $|r_t|\le1$ and $\sum_t|r_t|\le2$. Therefore
\[
 p_1=\max\{\pi_1,\pi_0+\pi_2\},\qquad
 p_2=1-\min_t\pi_t,\qquad p_3=1.
\]
These are full-class optima, including arbitrary mixed states and
adaptation, rather than optima over two named product preparations.
At the uniform workload, $p_1=p_2$: entangling a pair has no benefit.
At $(0.2,0.6,0.2)$, $p_1=0.6$ and $p_2=0.8$.

We calculate the exact finite-call ROC from
\[
 q_h(k,j)=\binom Nk p^k(1-p)^{N-k}\binom kj e_h^j(1-e_h)^{k-j}.
\]
Here $k$ counts informative reports and $j$ counts their flips. Sorting
$q_1/q_0$ and randomizing at the boundary gives the Neyman--Pearson
test. We use $e_0=0.02$ and an additional independent 15\% fault,
so $e_1=0.164$, with $\alpha=0.05$. At 32 calls, optimal power is
$82.39\%$ for all separable policies and $88.45\%$ at depth two.
The fixed-horizon minimum counts attaining 80\% power are 30 and 22. Independent
linear-program checks and 100,000 random product states test the
implementation; the proofs supply global optimality.

\begin{table}[ht]\centering
\caption{Exact-model resources for fixed-horizon tests attaining 80\% power at 5\%
false alarms. Preparation and detector columns count logical CNOTs.
Fractional entries are expectations over requests, not fractional gates.
These known-channel costs exclude calibration and dynamic-control latency.}
\label{tab:diagnostic-resources93}
\begin{tabular}{lrrrr}\toprule
Probe policy & Calls & Preparation & Detector & Total CNOTs\\\midrule
Fixed-horizon optimal separable & 30 & 0 & 126.0 & 126.0\\
Bell pair & 22 & 22 & 92.4 & 114.4\\
Full GHZ & 18 & 36 & 75.6 & 111.6\\
Query-aware product & 18 & 0 & 75.6 & 75.6\\
Adaptive second $ZZI$ call & 28.8 & 0 & 108.0 & 108.0\\\bottomrule
\end{tabular}
\end{table}

The first service call uses 4.2 detector CNOTs on average at this
workload. A Bell preparation adds one. The adaptive product control
makes a second $ZZI$ call only after $XXX$, so it averages 1.6 calls
and six detector CNOTs per probe. Over $N$ probes its call count is
$N+\operatorname{Binomial}(N,0.6)$. It costs fewer CNOTs than the Bell
policy in Table~\ref{tab:diagnostic-resources93}, but more detector
calls. These are not lower bounds on expected cost with early stopping:
a sequential implementation can stop when the terminal decision is
already determined. Early query access removes the preparation problem altogether.

For $L$ diagnostic episodes, the total cost of policy $a$ is
$C_a(L)=C_{\rm cal}(a)+L N_a[C_{\rm prep}(a)+C_{\rm detector}(a)]$.
Thus extra pair-calibration cost $\Delta C_{\rm cal}$ is amortized
against the fixed-horizon one-call separable policy only when
$11.6L>\Delta C_{\rm cal}$, if cost is measured in logical CNOTs.
The corresponding detector-call condition is $8L>\Delta Q_{\rm cal}$.
This comparison charges calibration rather than assigning it zero
cost. Appendix~\ref{app:service94} records actual calibration
shots, compiled native gates, measurement/reset slots, latency and
QPU usage separately; CNOT counts are not a substitute for elapsed time.
The frontier also assumes ideal probes. An additional independent
1\% flip of the Bell probe's supported parity signs reduces its
32-call power from 88.45\% to 84.84\%; at 1.6\% it falls to 82.00\%,
below the ideal separable value of 82.39\%. Preparation error can
therefore remove a detector-call saving before cost amortization matters.

The operational constraint describes a factory that commits probes
before a multiplexed detector's next request is selected. It can use
past demand but cannot rebuild the current probe. This can be
implemented and tested as a service contract. Ordinary syndrome
schedules need not satisfy it, and a deployed diagnostic system with
early query access should use the cheaper query-aware control.

\subsection{A certificate for arbitrary binary detector effects}
The scalar coverage reduction need not survive biased, query-dependent
or non-Pauli errors. We therefore allow arbitrary healthy and faulty
binary effects $E_{h,t,y}$, with no CSS or independent-error assumption.
The request still arrives after preparation, and these effects specify
the conditional response on every use; a stationary one-use marginal
alone does not imply a memoryless experiment.

For a fixed $\lambda\ge0$ and decision vector
$d\in\{0,1\}^{2q}$, define
\[
 A_{d,\lambda}=\sum_{t,y}\pi_t d_{ty}
                     (E_{1,t,y}-\lambda E_{0,t,y}).
\]
Every level-$\alpha$ alarm using a separable state has power at most
$\lambda\alpha+U(\lambda)$, where
\[
 U(\lambda)=\max_d\ \max_{\rho\succeq0,\ \operatorname{Tr}\rho=1,
                         \rho^{\Gamma_i}\succeq0\ \forall i}
                  \operatorname{Tr}(A_{d,\lambda}\rho).
\]
The positive-partial-transpose (PPT) class contains every fully
separable state. Enumerating binary decisions suffices because the
objective is linear in each alarm probability. The same upper bound
holds for arbitrary revealed classical mixtures.

For finitely many slopes, the concave polygon
$\overline B(\alpha)=\min\{1,\min_\lambda[\lambda\alpha+U(\lambda)]\}$
is the ROC of a dominating binary-hypothesis experiment. Its outcome
masses are the horizontal and vertical increments of the polygon,
including a fault-only atom of mass $\overline B(0)$. ROC ordering is
Blackwell ordering for classical experiments with two hypotheses
\citep{blackwell1953}; the classical garbling criterion is reviewed
in \citet[Section~2]{buscemi2012}.
Consequently the product of this experiment dominates every
classically adaptive separable $N$-call policy: condition on its
history and apply the one-use garbling at each step. This gives a
finite-report upper bound without fitting a separable probe catalog.
Its computation scales exponentially in the number of outcome/query
decisions and in Hilbert-space dimension; here $q=3$ and the dimension
is eight.

Solver output is not itself a certificate. We store Gram factors
$Z_i=L_iL_i^\dagger$ and verify
$uI-A_{d,\lambda}-\sum_i Z_i^{\Gamma_i}\succeq0$ using interval
arithmetic after a numerical change of basis. Gram positivity is
exact; interval Gershgorin bounds also certify that the basis change
is invertible. A separate matrix-formation allowance covers floating
products and exact-complement rounding. Polygon intersections and
likelihood-ratio refinement use rational arithmetic. Each likelihood
atom is split between its enclosing geometric-grid ratios while
preserving both hypothesis masses, an explicit Blackwell refinement.
Positive convolutions propagate outward error bounds. The grid ratio
is $41/40$, with $21/20$ as a coarser check. These are numerical upper
bounds, not confidence intervals from sampled data.

\subsection{Optimized probes under correlated and non-Pauli errors}
We perturb the detector with the unitary
$U_s=\exp[-is(XYI+0.7IYX+0.3YIZ)]$, followed by independent amplitude
damping of strength $0.12s$. If $\mathcal N_s$ is this channel, the
healthy positive effect for query $t$ is
\[
 E_{0,t,+}=\tfrac12\{[1+0.03s(t-1)]I+
                         (0.96-0.08s)\mathcal N_s^*(P_t)\},
 \quad t=0,1,2.
\]
Faulty reports undergo a further 15\% flip. The stored outcome-zero
matrices and their exact complements define the numerical POVMs.
These declared channels include correlated coherent errors and biased
responses; they are not fitted device models. Response probabilities
alone do not certify the post-measurement instrument.

We optimize six Bloch angles for product probes. A one-call search
starts from all 216 Pauli product states; a separate search directly
optimizes a 32-report likelihood alarm. Both use seed 930907 and retain
the best known feasible probes. Likelihood-score grouping gives an
explicit achievable test. The separable upper bound, rather than the
optimizer's convergence, supplies the comparison to the full class.

\begin{table}[ht]\centering
\caption{Detection at 32 calls and 5\% false alarms under three declared
POVM models. Product and Bell columns are achievable powers; the final
column bounds every separable policy, including classical adaptation.
No column reports a hardware estimate or a sampling confidence interval.}
\label{tab:general-detector93}
\begin{tabular}{lrrrr}\toprule
Perturbation $s$ & Fixed $X$ product & Optimized product & Best fixed Bell & Separable upper\\\midrule
0 & 82.39\% & 82.39\% & 88.45\% & 82.42\%\\
0.08 & 66.89\% & 75.04\% & 79.96\% & 77.77\%\\
0.20 & 34.42\% & 68.09\% & 57.31\% & 70.78\%\\\bottomrule
\end{tabular}
\end{table}

At $s=0.08$, the Bell probe exceeds the full separable upper bound
by more than 2.19 percentage points. At $s=0.20$, the optimized
product outperforms both fixed Bell probes. Comparing only with the
named $X$ product would miss this reversal. The upper bound can be
loose: it relaxes separability, uses finitely many slopes and refines
likelihood ratios. At $s=0$, its excess over the exact separable
optimum is 0.023 points, providing a quantitative implementation check.

\begin{figure}[ht]\centering
\includegraphics[width=\linewidth]{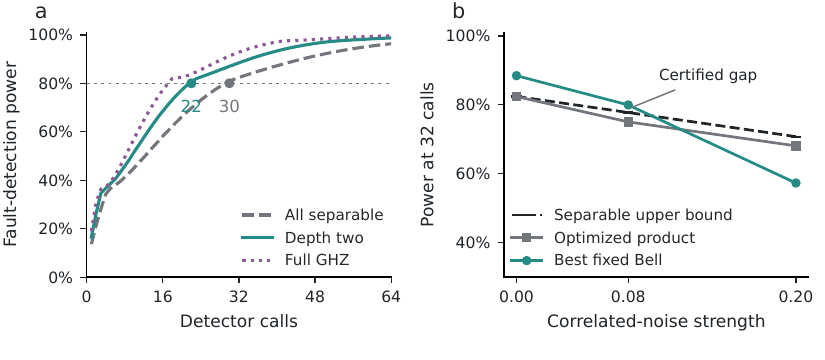}
\caption{(a) Exact optimal detection power at 5\% false alarms in the
report-flip model, including the full separable and depth-two classes.
(b) At 32 calls, optimized product and fixed Bell probes under the
declared correlated-noise models, with an upper bound on every
classically adaptive separable policy. The positive gap at $s=0.08$
is certified; the reversal at $s=0.20$ motivates direct probe optimization.
All curves are model calculations.}
\label{fig:diagnostic93}
\end{figure}

\subsection{From ideal planning to physical acquisition}
The initial plan selected among 50 fixed designs by integrating
threshold randomness and test counts under the ideal detector model.
Its Bell-versus-product rule had an 88.69\% probability of passing in
two independently calibrated cycles. An extra 0.5\% Bell
parity-preparation flip reduced that probability to 40.79\%; at 1\%
it was 8.53\%. These are model planning probabilities, not observations.
Five ideal dynamic-circuit simulations, each with 32,768 shots,
reproduced every supported parity exactly and left unsupported
contrasts within 0.021 of zero.

Physical commissioning required a shorter request generator and
measurement-interval protection. Direct response calibration also
showed weaker and differently ranked probes than the ideal plan
predicted. Appendix~\ref{app:service94} gives the superseding acquisition
and its operating rules, fixed before the first threshold or test job.
The original circuit checks and planning calculations remain in the
artifact. Neither limited response calibration nor passing a finite
probe comparison supplies a full separable device certificate; that
would require a valid tomography and preparation-error confidence
region.

\subsection{Relation to other learning problems}
Request-distribution-aware access codes optimize recovery of one
requested bit from a many-bit message for a supplied distribution
\citep{tanasescu2020}. Our task stores one bit; we learn an entanglement
allocation from classical request observations under a fixed detector
and depth constraint.

Preparation games optimize sequential measurement protocols against
constrained, history-dependent state preparation
\citep{weilenmann2021preparation}. Our fixed query distribution and
memoryless detector hypotheses yield a common dominating experiment;
this permits tensorizing a one-use bound and identifies the allocation
learning cost. We do not claim the general idea of optimizing tests
over constrained preparation strategies.

Our request-only allocation learner observes classical query identities.
Property testing, state learning and classical shadows use measurement
data from copies of an unknown state
\citep{bubeck2020testing,aaronson2007learnability,huang2020shadows}.
Quantum-memory separations compare learners that retain quantum information
across experiments with learners that retain only classical outcomes
\citep{chen2021memory}. Embedding and classifier generalization study
information quantities or trainable circuits
\citep{banchi2021generalization,caro2022generalization}; our action
class follows from physical feasibility. Classical fractional caching also admits
capacity-dependent online regret bounds \citep{paschos2019}. At
$k=d-1$, our path problem becomes sparse-loss prediction
\citep{kwon2016sparse}, choosing the least requested query to omit.

Binary discrimination through any quantum channel admits optimal
orthogonal pure-state encodings \citep[Lemma~4]{dallarno2020};
here the optimum must also respect the entanglement class.

\FloatBarrier
\section{Prospective evaluation on a quantum processor}
\label{app:service94}

The model gives a criterion for useful probe allocation. A physical
implementation must also preserve the relevant coherences and transfer
its calibration to later observations. We test these requirements on
\texttt{ibm\_marrakesh}, using data qubits $(98,111,110)$ and ancilla
112. The detector queries remain $(ZZI,XXX,IZZ)$ with probabilities
$(0.2,0.6,0.2)$. These experiments evaluate the implemented detector
service under this known nominal workload. They test calibration
transfer and the one-call probe comparison; the request-learning rate
is established by the theory and its separate scaling experiments.
The full separable guarantees of Appendix~\ref{app:service93}
apply to its declared models.

\subsection{A request that arrives after preparation}
Each circuit commits its probe before generating the request. The
ancilla first supplies a Bernoulli bit with probability 0.6, then a
Hadamard and a second measurement supply a uniform bit. An $X$ gate
conditioned on the second bit returns the ancilla to zero. The four
request values have ideal probabilities $(0.2,0.3,0.2,0.3)$, independently
of the data. Flat conditional unitaries extract the selected parity;
all measurements occur outside those branches. Alarm scores use only
the recorded request and permitted detector reports. Terminal data-qubit
measurements are retained for diagnostics and never enter those scores.

Initial commissioning exposed a practical obstacle: the longer
measurement-and-reset request generator left Bell $X$ correlations
near zero. Shortening the request generator and applying local $XX$
echo sequences during ancilla measurement intervals restored both
complementary contrasts. The compiler resolves the echo delays using
its measurement schedule. The sequences are identities on the data
in the ideal circuit and add no entangling gates. Four commissioning
jobs selected this implementation; their observations are excluded
from every scientific calibration and test below.

Healthy operation uses the physical detector as implemented. The
fault condition adds $R_y(2\arcsin\sqrt{0.15})$ to the ancilla before
its report measurement. With ideal gates and measurement, this gives
the report probabilities of a 15\% parity-label flip. Physical errors
before and after the rotation enter the calibrated fault channel. It need
not equal a classical relabeling of the conditional quantum state on
arbitrary cross-sector superpositions. The coherence tests therefore
use the healthy instrument and make their own state-preservation
measurement.

\subsection{Calibration selects a probe; a second job checks it}
Each cycle starts with 64 informationally complete product inputs
under each hypothesis, using 1,024 shots per input, plus named probe
controls using 2,048 shots. Positive-semidefinite least squares fits
the three binary detector effects. One six-angle product search fits
the hypotheses separately. A second search constrains the effects by
the injected-flip relation and uses two independently seeded runs.
These are numerical candidate searches, without a global physical
separability certificate.

A separate job measures eight candidates under both hypotheses, with
4,096 shots per candidate and hypothesis: the two continuous product
candidates, $Z$ and $X$ products, both Bell allocations, GHZ, and the
adaptive $Z$ control. It selects the product candidate and entangled
allocation with greatest estimated 32-report power at 5\% false
alarms. All eight candidates remain in the record. The likelihood
scores use these directly measured conditional frequencies with
Jeffreys smoothing; tomography predictions do not substitute for
measured performance. Probe choice and scores are then frozen.

This validation step matters. In the first cycle, the unconstrained
fit predicted nearly perfect detection for its product candidate;
direct measurement estimated 15.01\% power. The structured candidate
reached 33.74\%, compared with 43.23\% for $Z$ product and 51.74\%
for GHZ. These calibration estimates selected $Z$ product and GHZ for
the first held-out comparison. They are not the prospective alarm
rates reported below. The second cycle again selects $Z$ product and
GHZ, with calibration estimates of 40.19\% and 58.54\%; its
unconstrained product candidate reaches only 11.78\% on direct
validation. Both cycles therefore evaluate the same selected probes
with separately fitted scores and thresholds; they do not demonstrate
a learning gain over fixed GHZ. In the physical comparison, ``selected product''
means the best validated candidate from this search and control set.

\subsection{Frozen alarms and independent test jobs}
Each primary policy receives 2,048 healthy calibration episodes of
32 consecutive shots. Its alarm threshold is healthy order statistic
2,009, with a strict exceedance rule. Under independent stationary
healthy episodes, this gives 95.55\% per-policy tolerance coverage for
a 2.5\% false-alarm target. It is not a simultaneous guarantee across
policies or cycles. Four secondary policies receive 128 episodes and
use rank 126, targeting 5\%.

Later jobs acquire 1,536 healthy and 1,536 faulty episodes for each
primary policy; secondary policies receive 128 of each. Every episode
is one original 32-shot circuit instance (PUB). Each primary shot
prepares a fresh probe and executes one requested detector call;
the alarm combines 32 such reports. Records are
randomly ordered within a stage before submission. A second
cycle repeats fitting, candidate validation, threshold calibration,
and testing, with no shared job or observation. A cycle passes only
if a one-sided Fisher test favors the selected entangled probe at
$p\le0.025$ and both primary healthy one-sided 95\% binomial upper
limits are at most 5\%. At this sample size the latter permits at
most 62 alarms in 1,536 healthy episodes. Both cycles must pass;
no pooled result rescues a failed cycle. A GHZ result does not confirm
the earlier Bell-specific hypothesis.

Both cycles share the backend, layout, implementation, and acquisition
day. Separate jobs establish observation separation, not independent
physical noise. Within-PUB shot chronology was not retained.

The enlarged sample sizes were fixed before the first threshold job.
They reflect weaker separation in physical calibration than in the
ideal model. Calibration-only sensitivity calculations propagate
conditional-response and request uncertainty through random
thresholds and test counts. They leave substantial uncertainty about
power and do not account for drift or selection uncertainty. Test
observations never change the probes, scores, thresholds, or sample
sizes. Per-job rates and temporal diagnostics assess the stationarity
assumption separately from the frozen primary rule.

\subsection{Coherence within both parity sectors}
For each Bell allocation we test its supported $Z$ query and $XXX$;
for GHZ we test all three queries. Each preparation uses both parity
sectors and both opposite coherence phases. Three instruments follow
the same request generator: ancilla parity extraction, local
refinement on the query's nonidentity support, and a nominal idle
reference. Each cell uses 256 shots, giving 84 cells per cycle.
These circuits retain the request generator but apply fixed-query
extraction, testing the corresponding primitives. There is no
postselection on requests, reports, or final outcomes.

After a $Z$ query, the witness is $XX$ on the Bell pair or $XXX$ on
GHZ. After $XXX$, it is the appropriate pair's $ZZ$ witness; for GHZ
we use $ZZI$. Multiplying by the prepared sign puts all cells on the
same scale. Ideal extraction and the idle channel preserve a value
of one; the destructive local refinement gives zero. The idle
reference is not duration matched, so its ratio to the extraction
arm is not a preservation estimate. We report absolute residual
coherence and all sectors and phases. Cell intervals use
Clopper--Pearson bounds with Bonferroni correction over the 84 cells
in each cycle, assuming independent stationary shots. These witnesses
test the stated input family, not a complete quantum instrument.

\subsection{Measured outcomes and resource accounting}
\begin{table}[ht]\centering
\caption{Prospective physical alarms under the frozen rule. Each primary cell contains 1,536 separate 32-shot circuit instances. Intervals are pointwise 95\% binomial intervals under independent stationary episodes; job-level sensitivity is reported separately.}
\label{tab:physical94}
\begin{tabular}{llrr}\toprule
Cycle & Probe & Fault alarms (\%) & Healthy alarms (\%)\\\midrule
1 & Selected Z product & 30.53 [28.24,32.91] & 1.43 [0.90,2.16]\\
1 & GHZ & 39.65 [37.19,42.15] & 1.95 [1.32,2.78]\\
2 & Selected Z product & 29.82 [27.54,32.18] & 1.69 [1.11,2.47]\\
2 & GHZ & 46.16 [43.64,48.69] & 2.28 [1.59,3.15]\\
\bottomrule\end{tabular}
\end{table}
The two-cycle superiority rule passes.
In cycle 1, the power difference is 9.11 percentage points (conservative 95\% interval [3.60,14.58]); the one-sided Fisher $p$ is $7.20\times10^{-8}$. The healthy upper bounds are 2.04\% for Selected product and 2.64\% for GHZ.
In cycle 2, the power difference is 16.34 percentage points (conservative 95\% interval [10.77,21.83]); the one-sided Fisher $p$ is $5.79\times10^{-21}$. The healthy upper bounds are 2.34\% for Selected product and 3.01\% for GHZ.

\begin{figure}[ht]\centering
\includegraphics[width=\linewidth]{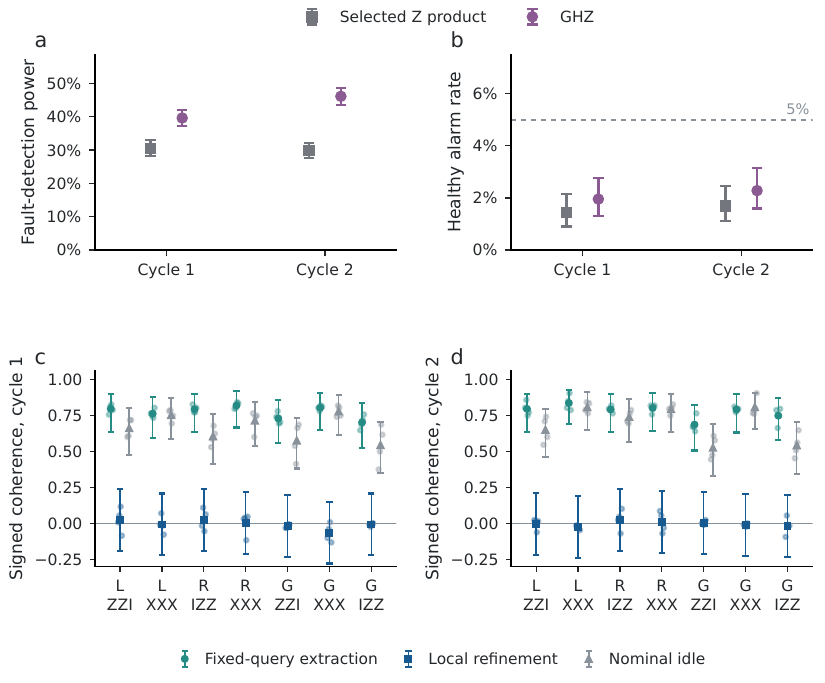}
\caption{(a,b) Held-out detection and healthy alarm rates: 1,536 episodes per policy and hypothesis, 32 shots per episode, with pointwise 95\% binomial intervals. (c,d) Fixed-query coherence controls: columns identify the Bell allocation (L or R) or GHZ (G) and the query. Faint points retain both sectors and phases. Solid points average those four fixed inputs; intervals average the simultaneous cell bounds, preserving their 95\% coverage separately within each cycle under independent stationary shots in each cell. Idle is nominal, not duration matched.}
\label{fig:physical94}
\end{figure}
The full acquisition uses 1,482,496 shots and 522 charged QPU seconds across 40 distinct jobs, including commissioning and coherence controls.
The interval from the first submission to the last completion is 1.99 hours, including compiler, queue, analysis, and client gaps; it is distinct from charged processor time.

\begin{table}[ht]\centering
\caption{Complete new QPU acquisition budget. Shared calibration is charged once to the procedure, not separately to every policy.}
\label{tab:physical94-resources}
\begin{tabular}{lrrr}\toprule
Stage & Jobs & Shots & QPU seconds\\\midrule
Commissioning & 4 & 128,768 & 44\\
Fit and candidate validation & 4 & 442,368 & 130\\
Healthy thresholds & 14 & 294,912 & 128\\
Held-out tests & 14 & 458,752 & 168\\
Coherence controls & 2 & 43,008 & 16\\
Dynamic-branch follow-up & 2 & 114,688 & 36\\
\bottomrule\end{tabular}
\end{table}
Product, Bell, and GHZ preparation use 0, 1, and 2 native CZ gates. Under the nominal request law, their detector circuits each average 4.2 CZ gates, giving totals of 4.2, 5.2, and 6.2 per trial. The adaptive product uses 6.0 CZ gates and 1.6 logical detector calls on average. Every response trial also contains two request measurements and three terminal diagnostic measurements; the adaptive control reserves two detector-report slots. These circuit counts exclude automatic initialization and finalization; their cost is included in the charged job usage above.

\begin{table}[ht]\centering
\caption{Secondary physical comparisons. Each cell has 128 separate 32-shot episodes, with thresholds calibrated from 128 healthy episodes at rank 126. These target 5\% false alarms, whereas the primary policies target 2.5\% with larger calibration samples. Differences from Table~\ref{tab:physical94} therefore reflect operating rules as well as probes. Adaptive Z uses extra detector calls.}
\label{tab:physical94-secondary}
\begin{tabular}{llrr}\toprule
Cycle & Probe & Fault alarms (\%) & Healthy alarms (\%)\\\midrule
1 & Z product & 40.62 [32.04,49.66] & 2.34 [0.49,6.70]\\
1 & X product & 36.72 [28.38,45.69] & 0.00 [0.00,2.84]\\
1 & Left Bell & 25.00 [17.77,33.42] & 0.00 [0.00,2.84]\\
1 & Adaptive Z & 52.34 [43.34,61.24] & 3.91 [1.28,8.88]\\
2 & Z product & 26.56 [19.15,35.09] & 0.00 [0.00,2.84]\\
2 & X product & 32.03 [24.06,40.85] & 6.25 [2.74,11.94]\\
2 & Left Bell & 31.25 [23.35,40.04] & 0.78 [0.02,4.28]\\
2 & Adaptive Z & 65.62 [56.72,73.79] & 6.25 [2.74,11.94]\\
\bottomrule\end{tabular}
\end{table}
Cycle 1's minimum ancilla-extraction cell lower bound is 0.462, with simultaneous coverage over its 84 cells. 
Cycle 2's minimum ancilla-extraction cell lower bound is 0.453, with simultaneous coverage over its 84 cells. 

The pooled rule and per-job stability answer different questions. In the first cycle, GHZ healthy and fault alarm rates vary across jobs and increase with acquisition order after Holm correction. Its final job has 14/224 healthy alarms (6.25\%; pointwise 95\% interval [3.46,10.26]\%). Passing the pooled healthy gate therefore does not establish control below 5\% in every job.
None of the 42 listed cycle-2 drift diagnostics survives Holm correction; this does not establish stationarity.

Conditioning on each job's fault-test margins gives a one-sided exact $p=1.01\times10^{-7}$ in cycle 1; GHZ has higher power in 6 of 7 jobs.
Conditioning on each job's fault-test margins gives a one-sided exact $p=7.53\times10^{-21}$ in cycle 2; GHZ has higher power in 7 of 7 jobs.
This sensitivity permits different baseline rates across jobs but still requires conditional exchangeability within jobs.

\begin{figure}[ht]\centering
\includegraphics[width=\linewidth]{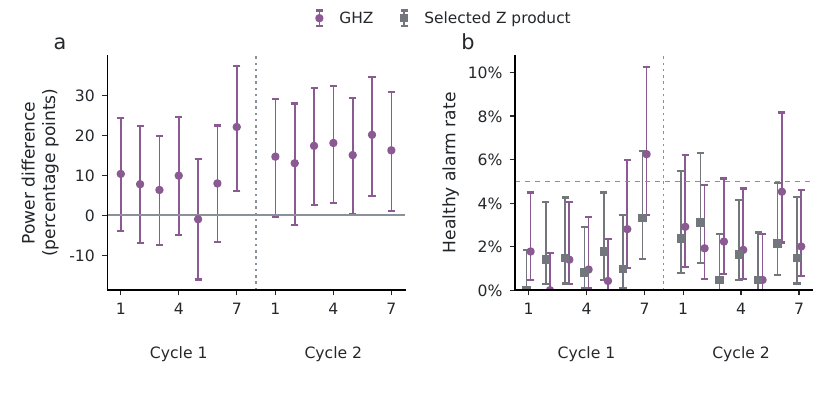}
\caption{Every held-out hardware job, in acquisition order within each cycle. (a) Entangled-minus-product detection power with conservative pointwise 95\% intervals. (b) Healthy alarm rates with pointwise 95\% binomial intervals; the dashed line marks 5\%. Intervals retain their independent-sampling assumptions. Job-specific variation limits extrapolation from aggregate alarm control.}
\label{fig:physical94-drift}
\end{figure}

The adaptive $Z$ control requests a second parity on $XXX$ requests,
relaxing the one-call interface. It uses 1.6 logical detector calls
on average and two physical report slots in every circuit, including
ancilla reset overhead. Its alarm power is informative about that
resource tradeoff, not an equal-resource separable counterexample.
Preparation gates, detector circuits, shared calibration shots,
terminal diagnostic measurements, and charged QPU time are counted
separately. A fixed 32-report comparison alone is not a physical
sample-complexity or total-cost advantage.

\subsection{Coherence through the dynamic service}
The fixed-query controls do not exercise the request-conditioned
extraction branches. A subsequent follow-up tests those branches
directly, with two new jobs of 28 input/target-query cells and 2,048
shots per cell. Both protocols were fixed before either job. The
follow-up uses the same sectors, phases, and witnesses, but executes
the complete four-branch request selection used for alarms.

For each input, we estimate coherence conditional only on its
recorded target request. Every request and outcome remains in the
raw record; neither detector agreement nor witness value selects
shots. The 56 target cells across both jobs form one Bonferroni
family. Conditional binomial intervals assume independent stationary
shots given the observed request counts. Expected-parity agreement
from the detector report is a separate diagnostic on the same shots.
Unconditional witness values are also retained: ideally they are
0.8 for left-Bell/$ZZI$ and right-Bell/$IZZ$, and one for the other
input families. This follow-up leaves the original alarm rules and
their outcomes unchanged.

56 of 56 target-request coherence cells have positive simultaneous lower bounds; the smallest is 0.302. Expected-parity agreement ranges from 78.80\% to 95.63\% across the cells. These outcomes measure residual coherence through the implemented dynamic branches on the stated input family.

\begin{figure}[ht]\centering
\includegraphics[width=\linewidth]{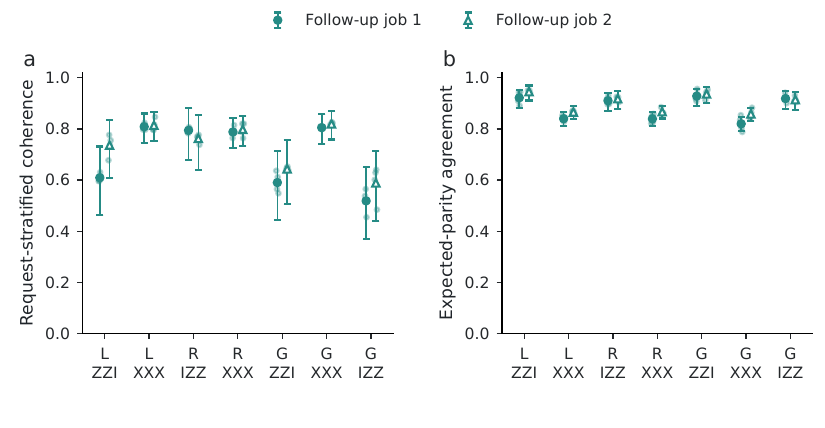}
\caption{Dynamic-branch follow-up on the quantum processor. Each point averages both sectors and phases for one preparation and target query, with four faint input-level values. Every input receives 2,048 shots; the analysis stratifies only by the recorded request. (a) Signed coherence, with bounds averaged from the simultaneous 56-cell family across both jobs. (b) Expected-parity agreement; each mean has a conservative pointwise 95\% interval obtained from four Bonferroni-adjusted cell bounds. L/R denote Bell allocations and G denotes GHZ. Neither observable selects the observations used by the other.}
\label{fig:physical94-dynamic}
\end{figure}

\subsection{Acquisition provenance and timing audit}
Every acquisition retains its frozen circuit file, backend calibration,
preparation-prefix checks, protocol and source hashes, actual IBM job
identifier, joint outcome counts, and charged usage. Stages are split
into partitions before their first submission; each record retains
its actual job identifier. The same free Open-plan instance supplies
all jobs. No mitigation or synthetic resampling replaces physical
observations.

The experimental scheduler metadata requires an additional check.
Some returned timing payloads correspond to a different circuit in
the same job, despite matching top-level circuit labels. In the first
candidate-validation job, two payloads exchange healthy-GHZ and
faulty-Bell signatures. Their measured $IZZ$ report-one rates,
9.05\% and 52.94\%, respectively, agree with the intended probes and
independent fit controls, rather than the swapped timing identities.
Deterministic request controls give the same conclusion. We retain
the original payloads and the audit; unverified per-record timing
cannot establish duration matching. This does not change the counts,
calibration choices, or frozen alarm rule.

\FloatBarrier
\section{The cost of changing the detector interface}
\label{app:interface97}
Extra calls can match an entangled probe's diagnostic performance.
To determine their cost, we compare the complete statistical experiments.
Let
$\mathsf E_p$ denote the experiment with probabilities
\[
 P_h(\bot)=1-p,\qquad P_h(0)=p(1-e_h),\qquad P_h(1)=pe_h,
\]
where $h\in\{0,1\}$ indexes detector health, $0<e_0<e_1<1$, and
the informative flag is observed. The nonerased bit records whether
the reported sign disagrees with the prepared sign. We use Blackwell equivalence under both hypotheses throughout this
section. It requires more than matching power at one operating point.

\subsection{Fresh requests and optimal expected call conversion}
\begin{theorem}[Exact fresh-call conversion]
\label{thm:interface97}
Suppose every admissible fresh call, conditional on its preceding
classical history, is dominated by $\mathsf E_p$, and some fixed
admissible probe attains $\mathsf E_p$. A fresh probe precedes each
new independent request; no quantum system is retained between calls.
For $0<p\leq q\leq1$, any terminating policy whose output dominates
$\mathsf E_q^{\otimes N}$ uses at least $Nq/p$ calls in expectation
under each hypothesis. A stopped policy attains equality and produces
an equivalent complete experiment.
\end{theorem}
\begin{proof}
Write $d_h=D_{\rm KL}(\operatorname{Bern}(e_h)\Vert
\operatorname{Bern}(e_{1-h}))$. Conditional one-call domination
bounds the KL increment by $pd_h$. The chain rule for a stopped
transcript, first at a truncated horizon and then by monotone
convergence, gives
\[
 D_{\rm KL}(P_h^{\rm transcript}\Vert P_{1-h}^{\rm transcript})
 \leq pd_h\,\mathbb E_h\tau.
\]
The action and stopping kernels are functions of the observed history,
so they add no likelihood ratio. By data processing, the target
experiment requires the left side to be at least $Nqd_h$. Dividing
proves the lower bound. Infinite expected cost satisfies it trivially.

For one target output, make a fresh call. Return an informative bit
if obtained. Following an erasure, with probability
$(q-p)/(1-p)$ restart until the first informative bit; otherwise
return an erasure. The informative probability is $q$ and expected
cost is $1+(q-p)/p=q/p$. Conditional on the output, the discarded
erasures and stopping details are independent of health. The full
transcript is therefore equivalent to $\mathsf E_q$. Independent
repetition proves the claim for $N$ outputs. The case $p=1$ is trivial.
\end{proof}

For the menu $(ZZI,XXX,IZZ)$ with probabilities $(.2,.6,.2)$,
the separable and Bell coverages are $.6$ and $.8$.
Here a bounded construction suffices: after an erased first product
call, make one fresh second attempt with probability $5/6$ and then
stop. Its coverage is $.6+.4(5/6)(.6)=.8$, and its expected cost
is $4/3$. For $N$ target outputs the cost is
$N+\operatorname{Binomial}(N,1/3)$, with maximum $2N$.
The expectation is $4N/3$ calls; the policy may exceed that number
on a particular run.

At a fixed fresh horizon $M$, the sufficient statistic is the number
$K$ of informative calls and the number $J$ of flips:
\[
 P_h(K=k,J=j)={M\choose k}p^k(1-p)^{M-k}
                    {k\choose j}e_h^j(1-e_h)^{k-j}.
\]
Sorting these atoms by likelihood ratio and randomizing at the
boundary gives exact Neyman--Pearson power. At $(e_0,e_1)=(.02,.164)$,
80\% power at 5\% false alarms requires 30 fresh product calls or
22 Bell calls. This finite-horizon ratio is $30/22$.
Equal expected informative counts alone do not imply equivalence.
The ratio $q/p$ also compares fixed-false-alarm Stein exponents,
since $D_{\rm KL}(\mathsf E_p^0\Vert\mathsf E_p^1)=pd_0$.
Equal-prior Bayes error instead has Chernoff exponent
\[
 -\log\!\left[1-p+p\min_{0\leq s\leq1}
 \{e_0^{1-s}e_1^s+(1-e_0)^{1-s}(1-e_1)^s\}\right],
\]
whose Bell/product ratio is $1.33860$ at these parameters.

\subsection{Chosen extra queries and retained probes}
Prepare $|+++\rangle$ and serve the first random request.
If it is $XXX$, its ideal sign is known. After either unsupported
$Z$ request, ask for $XXX$ with probability $u$.
All three observables commute, so every branch of the first ideal
parity measurement preserves the $XXX$ sign. The unsupported first
report is fair under either health hypothesis. The complete experiment
is therefore $\mathsf E_{.6+.4u}$, at $1+.4u$ expected calls per
original probe. In particular, $u=1/2$ reproduces the Bell experiment
at 1.2 expected calls. More generally, a supported commuting repair
query gives coverage $q$ at cost $1+q-p$.
This is an attainable policy, not an optimum over all adaptive ones.
Chosen-query control matters as well as retention: a newly prepared
product eigenstate queried through a chosen supported observable also
achieves this call count, but consumes an additional preparation.

Keeping the probe also makes a different signal available. Repeat an
unsupported nondemolition $Z$ query. If its unknown ideal parity is
$S$ and the report flips are $F_1,F_2$, then
\[
 Y_1=S\oplus F_1,\quad Y_2=S\oplus F_2,\quad
 Y_1\oplus Y_2\sim\operatorname{Bern}(\theta_h),\qquad
 \theta_h=2e_h(1-e_h).
\]
Each report alone is fair, but their disagreement diagnoses detector
noise. The stored sign remains unknown. This is why the one-call
identity between prediction and diagnostic regret need not hold when
a policy retains the probe.

Keep the first report on $XXX$ and repeat the first query otherwise.
For $N$ original probes, let $K$ count unsupported requests, $J$ count
flips on the supported first requests, and $V$ count disagreements.
The exact likelihood is
\[
 P_h(k,j,v)={N\choose k}.4^k.6^{N-k}
 {N-k\choose j}e_h^j(1-e_h)^{N-k-j}
 {k\choose v}\theta_h^v(1-\theta_h)^{k-v}.
\]
At the same known-channel 5\% false-alarm level, 14 probes give
80.2647\% power; 13 give 78.8579\%. The 14-probe policy uses
19.6 calls on average and at most 28. It exceeds 22 calls with
probability 5.8319\%, so it is not a hard-budget counterexample to
a 22-call comparison. Independent $8\times8$ instrument calculations
verify the transcript factorizations to $1.2\times10^{-16}$.

\subsection{Preparation, calls, and calibration}
Let $S_a$ be policy $a$'s startup cost, including workload collection,
response fitting, threshold calibration and setup. Let $V_a(h)$ be
its expected operating cost per diagnostic episode under health $h$,
including every preparation, detector call, report, reset, control
operation and retained-state holding cost. To combine these costs, assign a price to each operation and count it
once. Over $L$ episodes,
entanglement is cheaper precisely when
\begin{equation}
 S_{\rm ent}-S_{\rm sep}<L\{V_{\rm sep}(h)-V_{\rm ent}(h)\}.
 \label{eq:cost97}
\end{equation}
Both policies must meet the same performance requirement and use
the same expected or hard resource convention. With a uniform call
price and preparation-CZ price, equal other costs reduce this to
$\kappa_{\rm call}(N_{\rm sep}-N_{\rm ent})>
\kappa_{\rm CZ}(G_{\rm ent}-G_{\rm sep})$, where $G$ counts
\emph{all} preparation CZs in the episode.

\begin{table}[ht]\centering\small
\caption{Known-channel policies reaching 80\% power at 5\% false alarms.
Calls and gate counts are expectations; retained policies have larger
hard maxima. CNOT counts are logical and exclude calibration.}
\label{tab:interface97}
\begin{tabular}{lrrr}\toprule
Policy & Preparations & Calls & Total CNOTs\\\midrule
Fresh product &30&30&126.0\\
Bell &22&22&114.4\\
Fresh stopped Bell simulation &$29\tfrac13$&$29\tfrac13$&123.2\\
Product, half repair &22&26.4&114.4\\
GHZ &18&18&111.6\\
Product, full repair &18&25.2&111.6\\
Product, repeat unsupported query &14&19.6&75.6\\\bottomrule
\end{tabular}
\end{table}
The repeated-query policy can therefore be cheaper than GHZ. If each call
has overhead $x$ in units of one logical CNOT, its operating cost
is $19.6x+75.6$, versus $18x+111.6$ for GHZ.
GHZ is cheaper exactly when $x>22.5$, before any extra startup cost.
These are two attaining policies, not an optimized adaptive frontier.
The physical GHZ and Adaptive-$Z$ studies in
Appendix~\ref{app:service94} use different threshold calibrations;
their observed powers cannot be inserted into this common-operating-point
comparison without a new experiment.

\section{Transaction locality as an allocation problem}
\label{app:classical97}
An allocation problem also arises when a database must decide which
records to store together before seeing the next transaction.
For record partition $P$ and transaction footprint $F_t$, define
$\ell_t(P)=\mathbf1\{F_t\text{ lies in one shard}\}$.
The objective is $L_\pi(P)=\sum_t\pi_t\ell_t(P)$, under a
record-capacity constraint and without replication. The learner
maximizes empirical locality using a log of transaction identities.
The cost $1-L_\pi(P)$ is the distributed-transaction fraction.
It does not by itself measure throughput or latency.

Schism formulates the transaction hypergraph objective in its
Appendix B, then uses a tuple coaccess graph under balance constraints
\citep{curino2010schism}. Our graph vertices are
queries, and deleting a vertex sacrifices a query requirement.
These graphs encode different decisions. A single transaction touching
two records has locality zero at shard capacity one, whereas deleting
one endpoint of its tuple graph retains positive vertex weight.
We therefore need an explicit construction to transfer the learning law.

\begin{corollary}[Classical transaction instances]
\label{cor:classical97}
For adjacent-pair transactions on $q+1$ ordered records, shard
capacity $k+1$ gives exactly the feasible locality masks of
component-order deletion on $P_q$ at cap $k$.
For $r$ disjoint $s\times s$ tables of unit records, take every row
and column as a transaction. If
\[
 C=(s-1)k+1,\qquad
 s\geq\max\{2,k,\lfloor(k-1)^2/4\rfloor+2\},
\]
then allowing $rs$ shards of capacity $C$ realizes every maximal
allocation of $r$ disjoint $K_{s,s}$ regions at cap $k$.
The minimax expected locality regret is
$\Theta(\min\{1,\sqrt{rk/m}\})$ after $m$ transaction observations.
The path instance inherits Eq.~\eqref{eq:path-learning90}, with
$d=q$ and without the quantum scale $D$.
\end{corollary}
\begin{proof}
A run of $a$ adjacent-pair transactions forces its $a+1$ records
into one shard. Distinct selected runs have disjoint record supports.
Thus the path masks are exactly the runs of at most $k$ transactions.

In a table, any selected row and column share a record. A mixed
selection of $a$ rows and $b$ columns must occupy one shard and
contains $s(a+b)-ab$ records. If $a+b\leq k$, this is at most
$(s-1)k+1$. If $a+b\geq k+1$, a mixed subcollection of size $k+1$
has at least $s(k+1)-\lfloor(k+1)^2/4\rfloor>C$ records.
Monochromatic selections use separate row or column shards.
These are precisely the biclique constraints.

Every maximal monochromatic selection uses $s$ shards. A maximal
mixed selection uses one shard for its union; the remaining records
fit in at most $s-1$ further shards because $C\geq s$ and the union
contains at least $s$ records. Summing over regions proves feasibility
with $rs$ shards. Cross-region grouping cannot evade the within-region
union bound. The optimum and empirical optimum therefore equal their
component-order counterparts. The disjoint-region upper and lower
bounds apply; locality regret is twice quantum prediction regret
at $D=1$.
\end{proof}
The strict capacity condition matters: $s=k=5$, $C=21$ also fits
three rows and three columns, contradicting a cap-five interpretation.
Increasing $s$ grows transaction width, shard capacity and the server
budget. The sample bound holds under this increase in storage resources.
Independent enumeration checks 57,994 small tuple partitions,
28 path cases and 104,129 capacity comparisons.

\subsection{Structured experiments and a public negative control}
We use $r\in\{1,4\}$, $s\in\{4,8\}$, $k\in\{1,2,3\}$ and
$m/(rk)\in\{4,16,64,256\}$. Each setting has 32 independent
instances. A fixed full-support Dirichlet family and a separate
budget-indexed paired-sign family test ordinary learning and
worst-case scaling, respectively. Exact empirical maximization
compares the row total, column total and largest $k$ weights in
each region. Every chosen mask is converted into an actual record
placement, and locality is recomputed from its transaction footprints.
All 15,360 such checks agree with the selected masks. Static rows,
static columns, empirical color choice and a known-workload optimum
use the same capacity and server envelope.

BenchBase continues the OLTP-Bench framework \citep{difallah2013oltpbench}.
Its SmallBank generator uses uniform distinct account pairs; the pinned
PostgreSQL sample assigns 40\% of requests to two-account procedures
\citep{benchbase2026}. We replay those footprints with 12 customer
bundles and four shards of capacity three bundles, charging three
logical records per bundle. Enumerating all 15,400 balanced packings
gives locality $.6+.4(2/11)=.672727$ for every packing.
Every packing has the same locality, so learning cannot improve it.
We replay the published generator's transaction footprints; we do not
execute BenchBase or measure throughput.

\subsection{Observed purchase baskets}
The public Online Retail data contain 541,909 line observations
\citep{chen2015retail}. We retain positive-quantity, positive-price
lines outside cancelled invoices, group distinct stock codes by
invoice, and sort by timestamp. This gives 19,960 purchase baskets.
We define an inventory workload that accesses the items in a basket;
the data are observed purchases, not a captured SQL trace.

The first 3,992 baskets choose the most frequent 12 or 24 stock codes.
The next 40\% supply training requests, the next 20\% train a
validation comparator, and the last 20\% form a common held-out
trace. We project each basket onto the selected codes and score
nonempty projections. The test denominators are 1,548 and 2,069
of the 3,992 final invoices. Oversized projections remain failures:
136 and 341 test baskets exceed the three-record shard capacity.
Projection makes the optimization tractable but changes the task.
It retains only 2.45\% and 4.31\% of test item incidences; just 10
and 23 original invoices lie wholly inside the selected catalogues.
The reported locality therefore applies to these subcatalogues.

Each method uses exactly $n/3$ shards with three records each.
A binary variable for every three-item block indicates whether that
block is a shard; each item must appear once. Its weight counts
training baskets contained in it. This set-partition MILP is exactly
the empirical locality objective, including singleton baskets and
the unavoidable failures larger than three items. All 128 fits
(16 request streams, four budgets, two schemas) and both validation
fits terminate with zero optimality gap. We use nested samples of
$m\in\{8,32,128,512\}$ requests from the training pool.

At 512 requests, mean test locality is 67.25\% and 52.21\% for the
12- and 24-item schemas, versus 63.57\% and 48.19\% for fixed
frequency-ranked grouping. Fixed hash placement gives 62.08\% and 47.70\%.
The validation-trained comparisons give 68.02\% and 53.02\% after
1,725 and 2,232 additional observed baskets. They are not population
or test-selected oracles. Intervals resample the 16 training streams,
conditional on the schema and the held-out trace. We do not assign
an $rk$ dimension to these general basket footprints.

\subsection{Complete baskets on the full catalogue}
We next retain all 3,922 stock-code identifiers and score every one
of the 3,992 held-out invoices in full. We supply the catalogue as metadata, extracted retrospectively from
the complete dataset. The learner knows the item identifiers but
receives no future frequencies or co-occurrences. The experiment
therefore tests allocation within a known catalogue.
The chronological pilot/training/validation/test split is unchanged.
Shard capacities are 128, 256 and 512 records. Every fitted and fixed
method actually uses 31, 16 and eight shards, respectively, the
minimum counts at those capacities. No method replicates records.

General basket partitioning is combinatorial. Our full-catalogue
implementation searches five training-only candidates: four greedy
transaction-merge orders and a pair-coaccess merge. Each merge obeys
the record cap; first-fit decreasing packs the remaining components.
The chosen candidate maximizes whole-basket locality on its training
prefix. This is a heuristic, not an exact empirical maximizer.
We evaluate $m\in\{8,32,128,512,2048\}$, with 16 nested request
streams, yielding 240 fits. All footprint scores, capacities, shard
counts, selected empirical objectives and training indices pass
independent replay against the original public archive.

\begin{figure}[ht]\centering
\includegraphics[width=\linewidth]{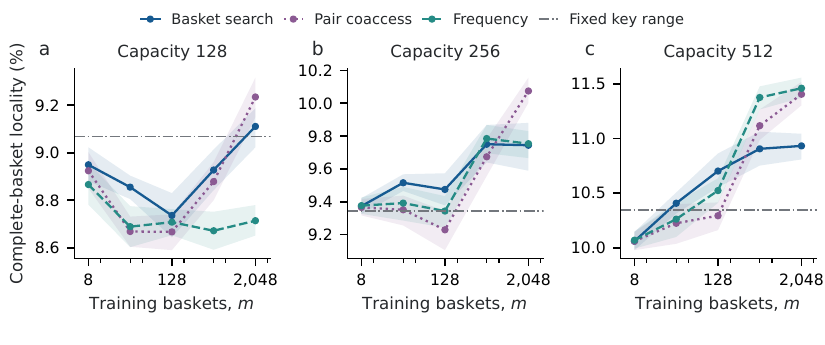}
\caption{Complete-invoice locality on the full Online Retail catalogue.
All three learned methods use the same training requests and actual
shard count; fixed key range uses none. Pointwise 95\% bootstrap
bands resample 16 request streams, conditional on the fixed catalogue
and held-out trace. Axes retain each capacity's local scale.}
\label{fig:classical-full97}
\end{figure}
The frequency baseline ranks items by counts from the identical
training prefix and groups consecutive items. We added this control after the primary study and retain its
implementation and every output. At capacity 512, it beats basket search
in 14/16 streams at $m=512$ and 15/16 at $m=2048$, by 0.470 and
0.528 points on average. At capacity 128, basket search instead has
higher mean locality at every budget. The pair-coaccess candidate
also exceeds the selected candidate on later data in some settings:
selecting by training locality need not select the best future partition.

Independently pilot-ranked and validation-trained controls each consume
3,992 additional baskets and are charged separately. At capacity 512,
their localities are 11.85\% and 11.65\%, compared with 10.35\%
for fixed key range and 7.57\% for fixed balanced hash. The full study shows modest locality gains from observing requests,
with the better method depending on capacity. It measures neither
throughput nor latency.

\begin{figure}[ht]\centering
\includegraphics[width=\linewidth]{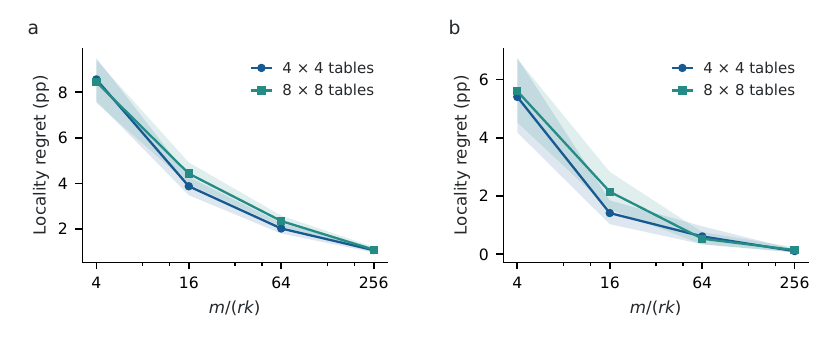}
\caption{Structured classical transfer. (a) Budget-indexed paired-sign
family; changing $m$ changes the population. (b) Fixed full-support
Dirichlet workloads. Curves average balanced cells in $r\in\{1,4\}$
and $k\in\{1,2,3\}$, with 32 instances per cell. The horizontal scale
is the proved $m/(rk)$; it is not applied to the public basket schema.}
\label{fig:classical-structured97}
\end{figure}

\section{Native 15-qubit path: execution and pending learning study}
\label{app:hardware97}
The data path is $(33,39,53,52,51,58,71,70,69,78,89,88,87,86,85)$
on \texttt{ibm\_marrakesh}. We selected the path using backend calibration alone, before observing
its experimental responses. Every one of the 15
queries is retained. Endpoint parity extraction uses three logical
CNOTs. Internal queries use five CNOTs to fold parity toward an
available ancilla and undo the borrowed-data transformation. Seven
centers have physical degree two. Ideal Clifford identities verify
the entire unitary, including restoration of arbitrary borrowed data;
a separate binary-linear calculation verifies the parity map.

The request is fixed in each circuit. Preparation is compiled
separately and has an identical prefix across requests for a given
codeword. The compiled program therefore fixes preparation independently of
the request, although the request is not generated late on the device.
Both full and shallow graph states use even/odd CZ layers. We use
native gates, active initialization, and no twirling, dynamical
decoupling or error mitigation. Default-delay commissioning uses
250 microseconds between shots; the failed learning pilots used
25 microseconds. Those conditions are separate.

Job \texttt{dagas0hhvn6c73cqbhbg} returned 7,680 shots and charged
five QPU seconds. Each of four fixed encodings has 64 shots per bit
and query. We fixed the decoder before acquisition: it predicts the encoded bit
directly from the reported parity.
Full chain has 107 errors in 1,920 trials, fixed depth three 274,
product even 481 and product odd 584. For the three simultaneous
contrasts against Full chain, independent-shot Hoeffding bounds give
depth-three minus full error $8.70$ points, interval $[3.70,13.69]$;
even minus full $19.48$, $[14.49,24.47]$; and odd minus full
$24.84$, $[19.85,29.84]$. These bounds assume independent shots within this fixed job.
Request-learning variability and changes between device calibrations
require separate acquisitions.

\subsection{Frozen learning design and execution status}
The request-only learner receives counts from $m\in\{8,32,128,512\}$
requests and solves the depth-three ideal allocation problem. Four
fixed full-support workloads are uniform, alternating, centrally
concentrated and bimodal. There are 16 independent nested request
streams per workload. The product comparator solves the depth-one
problem on the same counts; Full chain uses every edge. A
known-workload ideal depth-three comparator is privileged and is not
the best unknown physical encoding. A separately calibrated learner
at $m=128$ varies $B\in\{4,16,64\}$ response trials per local context.
It has additional observations and is never called request-only.

The planned 292,608-shot budget comprises 7,680 commissioning shots,
153,600 core comparisons, 46,080 calibrated tests, 57,344 calibration
shots, 15,360 finite-catalog selection shots, 11,520 coherence controls
and 1,024 input-bit controls. Core inference uses 16 separate
request/acquisition blocks; the calibration arm has eight independent
panels. Shared panels and response records must remain shared in
resampling. The complete allocation list, workload probabilities,
request prefixes, compiler and analysis plan are frozen in the artifact.

Two exploratory shared-response libraries were submitted before the
confirmatory acquisition. Job \texttt{dagb0sj9k43c73ad42j0}
planned 32,520 shots in 15,900 circuit publications; it timed out.
Job \texttt{dagb5m8mhr3c73e4taag} retained every learned mask,
used a stronger known-workload product comparator, and reduced the
submission to 2,340 publications and 18,720 planned shots. It also
timed out. Neither job returned usable observations, so neither contributes data
to the learning curves. The supplement preserves the frozen circuits,
receipts, error messages and usage records. The free-plan limit
prevents completing the confirmatory learning and calibration curves
in this revision.

The proposed success criterion is a positive paired advantage over
both extremes on every workload at $m\geq128$. We will evaluate this
hypothesis on every workload and retain failures.
A fixed-workload slope is descriptive: the minimax theorem does not
predict $-1/2$ on every population. Any budget-indexed hard-family
test must be labeled separately from learning a fixed workload.

\subsection{A limited coherence check exposes a sign reversal}
Job \texttt{dagb9efi3e6s738m5tf0} returned 1,920 shots in 30 circuits,
charging three QPU seconds. Every query has a coherent-extraction and
local-refinement arm, each with 64 shots. This small diagnostic fixes
the encoded bit and witness phase to zero; it does not cover both
parity sectors or certify the all-input detector contract.

After extracting query $K_i$, we measure the neighboring commuting
generator $K_j$ ($j=i+1$, except at the last site). Ideal extraction
preserves its signed expectation at one, whereas local refinement
sets it to zero. There is no postselection. Direct parity parsing of
the saved six-bit strings agrees with independent reversed-string
parsing, and all QPY metadata match the returned circuit metadata.
Exact Heisenberg propagation of the submitted native circuits verifies
the complete 64-outcome ideal distributions. A separate ideal
stabilizer simulation gives the same witness predictions. These
checks concern the circuit specification, not the physical observations.

\begin{figure}[ht]\centering
\includegraphics[width=\linewidth]{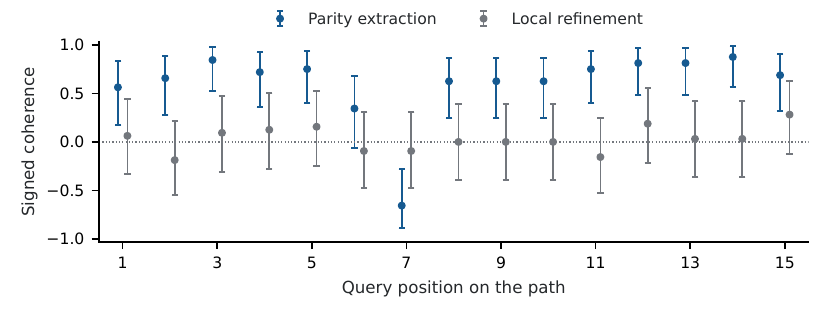}
\caption{All 15 native-path coherence diagnostics. Bars are 95\%
simultaneous Clopper--Pearson intervals over 30 cells, conditional on
independent shots. The query numbered seven reverses the expected
sign; query six is inconclusive. No sign is selected from these data.}
\label{fig:coherence97}
\end{figure}
Mean signed coherence is 0.602 for extraction and 0.029 for refinement.
Thirteen extraction cells have positive simultaneous lower bounds.
Query six gives 0.344, interval $[-0.062,0.675]$; query seven gives
$-0.656$, $[-0.887,-0.282]$. The latter measures the witness on
physical data qubits $(71,70,69)$. Its ideal value is positive, so
the reversal remains an unexplained acquisition/device result.
It rules out claiming that this check validates preservation across
the whole path. The next hardware stage must diagnose that cell and
repeat both sectors and phases before attributing allocation performance
to a coherence-preserving physical service.

\subsection{The reversal repeats on both bits and witness phases}
We used the remaining free allocation for a focused follow-up,
selected after the preceding result. Job
\texttt{dagbgdni3e6s738m65ig} returned 3,072 shots in 12 circuits,
charging three QPU seconds. The selected query is tested at both
encoded bits and both witness phases, with extraction, local refinement
and a preparation-only witness reference. Each cell has 256 shots.
The reference removes the extraction/measurement interval; it is not
duration matched. The protocol, signs, circuits and all cells were
frozen before submission. Ideal simulation of the exact submitted
circuits again predicts signed coherence one for extraction and reference,
and zero for refinement.

\begin{table}[ht]\centering\small
\caption{Signed witness in the focused second job. Each value uses
256 shots. The query and follow-up were chosen after the first job;
these are exploratory replication data, not a new confirmatory endpoint.}
\label{tab:coherence-followup97}
\begin{tabular}{ccrrr}\toprule
Encoded bit & Phase & Extraction & Refinement & Preparation only\\\midrule
0&0&$-0.570$&$-0.016$&$0.914$\\
0&1&$-0.680$&$0.023$&$0.883$\\
1&0&$-0.664$&$-0.055$&$0.906$\\
1&1&$-0.688$&$0.117$&$0.875$\\\bottomrule
\end{tabular}
\end{table}
All four extraction intervals are negative after Bonferroni correction
over the 12 cells; their largest upper endpoint is $-0.407$.
All preparation-only intervals are positive, with smallest lower
endpoint $0.763$; all refinement intervals contain zero.
The sign reversal therefore repeats in a separate job over this input
family. The comparison implicates the added extraction/measurement
interval, including its timing, rather than preparation alone; it
does not identify a specific faulty gate or distinguish measurement
backaction from coherent evolution during the added interval.
Both jobs use the same device and day. A duration-matched control
and a separately frozen repair test remain necessary.

\end{document}